\pdfoutput=1
\documentclass[openacc]{rsproca_new}%%%%where rsproca is the template name
\jname{rspa}
\Year{2024}
\esm{}

\newtheorem{proposition}{\bf Proposition}[section]

\usepackage[percent]{overpic}
\usepackage{pdfpages}

\usepackage{tikz,xcolor,hyperref}

\definecolor{lime}{HTML}{A6CE39}
\DeclareRobustCommand{\orcidicon}{%
	\begin{tikzpicture}
	\draw[lime, fill=lime] (0,0) 
	circle [radius=0.16] 
	node[white] {{\fontfamily{qag}\selectfont \tiny ID}};
	\draw[white, fill=white] (-0.0625,0.095) 
	circle [radius=0.007];
	\end{tikzpicture}
	\hspace{-2mm}
}

\foreach \x in {A, ..., Z}{%
	\expandafter\xdef\csname orcid\x\endcsname{\noexpand\href{https://orcid.org/\csname orcidauthor\x\endcsname}{\noexpand\orcidicon}}
}

\begin{document}
% \linenumbers
%%%% Article title to be placed here
\title{Inertia-gravity waves in geophysical vortices}

\author{J\'er\'emie Vidal$^{1}$ \& Yves Colin de Verdi\`ere$^{2}$}
%#\history
%#\rec{24 October 2023}
%#\acc{12 July 2021}

%%%%%%%%% Insert author address here
\address{${}^{1}$ Universit\'e Grenoble Alpes, CNRS, ISTerre, 38000 Grenoble, France\\
${}^{2}$ Universit\'e Grenoble Alpes, CNRS, Institut Fourier, 38000 Grenoble, France\\
\orcidA{ JV, \href{https://orcid.org/0000-0002-3654-6633}{0000-0002-3654-6633}}}

%%%% Subject entries to be placed here %%%%
\subject{geophysics, fluid mechanics, mathematical physics}

%%%% Keyword entries to be placed here %%%%

%%%% Insert corresponding author and its email address}
\corres{J\'er\'emie Vidal\\\email{jeremie.vidal@univ-grenoble-alpes.fr}}

%%%% Abstract text to be placed here %%%%%%%%%%%%
\begin{abstract}
Pancake-like vortices are often generated by turbulence in geophysical flows. 
Here, we study the inertia-gravity oscillations that can exist within such geophysical vortices, due to the combined action of rotation and gravity. 
We consider a fluid enclosed within a triaxial ellipsoid, which is stratified in density with a constant Brunt-V\"ais\"al\"a frequency (using the Boussinesq approximation) and uniformly rotating along a (possibly) tilted axis with respect to gravity.
The wave problem is then governed by a mixed hyperbolic-elliptic equation for the velocity.
As in the rotating non-stratified case considered by Vantieghem (2014, \emph{Proc. R. Soc. A}, \textbf{470}, 20140093, \href{https://doi.org/10.1098/rspa.2014.0093}{doi:10.1098/rspa.2014.0093}), we find that the spectrum is pure point in ellipsoids (i.e. only consists of eigenvalues) with smooth polynomial eigenvectors.
Then, we characterise the spectrum using numerical computations (obtained with a bespoke Galerkin method) and asymptotic spectral theory. 
Finally, the results are discussed in light of natural applications (e.g. for Mediterranean eddies or Jupiter's vortices).
\end{abstract}
%%%%%%%%%%%%%%%%%%%%%%%%%%%

%%%%%%%%%% Insert the texts which can accomodate on firstpage in the tag "fmtext" %%%%%

\begin{fmtext}
%-----------------------------------------------------------------------
\section{Introduction}
%-----------------------------------------------------------------------
Geophysical flows are often influenced by the action of density stratification and rotation. 
In particular, buoyancy supports the propagation of internal (gravity) waves in stratified fluids  \cite{mowbray1967theoretical}.
Similarly, the Coriolis force sustains inertial waves in homogeneous rotating fluids \cite{greenspan1968theory}.
Combining density stratification and global rotation then gives birth to a new wave family \cite{leblond1981waves}, usually called inertia-gravity waves (IGWs). 
Such various waves, which are ubiquitous in rotating and stratified environments, are often believed to be key for understanding the (turbulent) dynamics of geophysical flows \cite{dauxois2018instabilities,le2023wave}.
\end{fmtext}

%%%%%%%%%%%%%%% End of first page %%%%%%%%%%%%%%%%%%%%%

\maketitle

\begin{figure}
\centering
\begin{tabular}{cc}
    \begin{tabular}{c}
    \begin{overpic}[height=0.27\textwidth]
{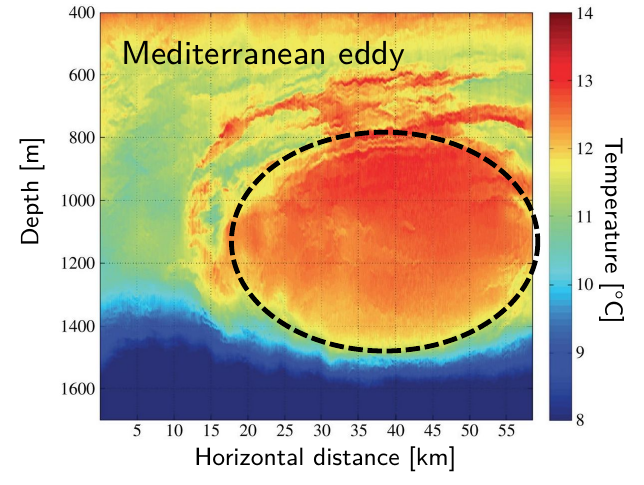}
    \put(-10,0){(a)}
    \end{overpic}
    \end{tabular} &
    \begin{tabular}{c}
    \begin{overpic}[height=0.27\textwidth]
{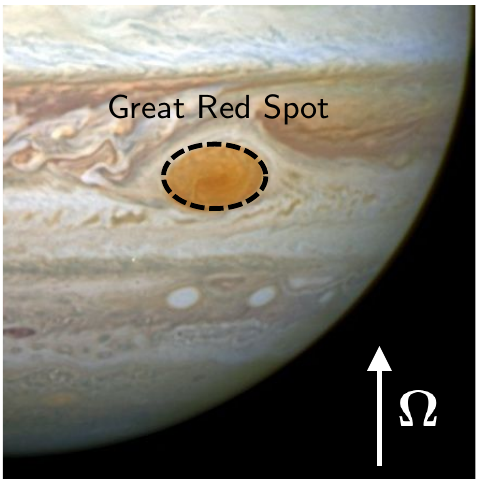}
    \put(-20,0){(b)}
    \end{overpic}
    \end{tabular} \\
\end{tabular}
\caption{(a) Cross-section of the temperature anomaly (acoustic tomography) in a Mediterranean eddy offshore in the Eastern Atlantic. Adapted from figure 3 in \cite{mcwilliams2016submesoscale}. (b) Picture of Jupiter's Great Red Spot (GRS) taken on 21 April 2014 with Hubble. Credits: NASA, ESA and A. Simon (Goddard Space Flight Center). Jupiter's axis of rotation is $\boldsymbol{\Omega}$.}
\label{fig:fig1}
\end{figure}

% \begin{figure}
% \centering
% \begin{subfigure}[t]{0.49\textwidth}
%     \centering
%     \caption{}
%     \includegraphics[height=0.6\textwidth]{./Fig1a.pdf}
% \end{subfigure}
% \hfill
% \begin{subfigure}[t]{0.49\textwidth}
%     \centering
%     \caption{}
%     \includegraphics[height=0.6\textwidth]{./Fig1b.pdf}
% \end{subfigure}
%     \caption{(a) Cross-section of the temperature anomaly (acoustic tomography) in a Mediterranean eddy offshore in the Eastern Atlantic. Adapted from figure 3 in \cite{mcwilliams2016submesoscale}. (b) Picture of Jupiter's Great Red Spot (GRS) taken on 21 April 2014 with Hubble. Credits: NASA, ESA and A. Simon (Goddard Space Flight Center). Jupiter's axis of rotation is $\boldsymbol{\Omega}$.}
%     \label{fig:fig1}
% \end{figure}

Another characteristic of rotating stratified fluids is that long-lived coherent vortices are often observed in geophysical conditions.
This is a direct manifestation of the high-Reynolds-number turbulence of geophysical flows. 
Rotating turbulence can sustain vortices with a large vertical extent without stratification (due to a nearly geostrophic balance \cite{hopfinger1993vortices}). 
However, such vortices become strongly flattened when density stratification is strong enough, leading to a pancake-like ellipsoidal shape \cite{billant2001self}. 
Notably, the elliptical shape strongly depends on the interplay between rotational effects, background shear, and the difference between the vortex's stratification and that of the ambient medium \cite{hassanzadeh2012universal,aubert2012universal,lemasquerier2020remote}. 
For instance, Mediterranean eddies are long-lived anticyclones of lenticular shape, with radii $10-100$~km and thicknesses of less than $1$~km  (figure \ref{fig:fig1}a), formed by the accumulation of warm salty water flowing from the Mediterranean Sea in the Atlantic Ocean  \cite{richardson2000census}.
Flattened vortices are also observed in planetary atmospheres. 
The most striking example is Jupiter's Great Red Spot (GRS, figure \ref{fig:fig1}b), which has persisted for more than a century  \cite{marcus1993jupiter}.

The origin of the long-term stability of such (nearly) isolated vortices has thus received much attention in geophysical fluid dynamics. 
For mathematical tractability, the quasi-geostrophic approximation is often employed to model a geophysical vortex by a fluid ellipsoid with spatially uniform potential vorticity \cite{mckiver2015ellipsoidal}. 
It has also been reported that such quasi-geostrophic vortices are often quite stable over time \cite{tsang2015ellipsoidal,facchini2016lifetime}.
However, various hydrodynamic instabilities could exist in geophysical vortices \cite{sipp1999vortices,yim2016stability} Ultimately, these instabilities may sustain (small-scale) bulk turbulence, which would likely affect the long-term vortex stability (because of additional dissipation).
In particular, the strong elliptical shape of geophysical vortices may trigger the so-called elliptical instability \cite{kerswell2002elliptical}. 
This is a parametric instability, due to nonlinear couplings between an elliptical flow and two waves (e.g. inertial or internal gravity waves).
Yet, it remains unclear whether the elliptical instability could be triggered for realistic parameters in rotating stratified fluids \cite{miyazaki1992three,kerswell1993elliptical,le2006thermo,vidal2019fossil}.
Before assessing the relevance of this mechanism for geophysical vortices, a preliminary step is to understand the wave properties in rotating stratified fluids. 

Although IGWs have been extensively studied in unbounded fluids \cite{leblond1981waves}, the properties of the global oscillations in bounded domains (e.g. the inertia-gravity modes, IGMs) are far from being fully understood. 
First insight into the mathematical problem can be obtained by considering the non-stratified regime, which has been extensively examined after Greenspan \cite{greenspan1968theory}. 
For homogeneous fluids viewed in a frame rotating at angular velocity $\boldsymbol{\Omega}$ (with respect to an inertial frame), inertial modes in a bounded domain $V$ are governed by the Poincar\'e problem
(called after Cartan \cite{cartan1922petites} who revisited Poincar\'e's pioneering paper
\cite{poincare1885equilibre}). 
The latter is given by
\begin{subequations}
\label{eq:poincareeq}
\begin{equation}
    \omega^2 \nabla^2 \Phi = (2 \boldsymbol{\Omega} \cdot \nabla)^2 \Phi, \quad \left . \nabla \Phi \boldsymbol{\cdot} \boldsymbol{n} \right |_{\partial V} = \left . \boldsymbol{n} \boldsymbol{\cdot} (\boldsymbol{u} \wedge 2 \boldsymbol{\Omega}) \right |_{\partial V},
    \tag{\theequation a--b}
\end{equation}
\end{subequations}
where $\boldsymbol{u}$ is the incompressible fluid velocity, $\Phi$ is the (reduced) hydrodynamic pressure and $\boldsymbol{n}$ is the unit vector normal to the fluid boundary $\partial V$.
Time dependence was assumed to be $\exp(\mathrm{i} \omega t)$, where $|\omega| < 2 |\boldsymbol{\Omega}|$ is the real-valued angular frequency \cite{greenspan1968theory}.
Equation (\ref{eq:poincareeq}a) is hyperbolic, but boundary condition (\ref{eq:poincareeq}b) is neither of Dirichlet-type or Neumann-type (it is a mixed-typed condition, sometimes called oblique condition). 
Hence, problem (\ref{eq:poincareeq}) is an ill-posed Cauchy problem \cite{rieutord2000wave}.
Nonetheless, it admits smooth polynomial eigenvectors when the geometry is a triaxial ellipsoid \cite{vantieghem2014inertial,ivers2017enumeration}.
Moreover, these eigenvectors form a complete set in the Hilbert subspace of complex-valued divergenceless fields tangent to the boundary with the $\mathrm{L}^2$ norm in ellipsoids \cite{backus2017completeness}.
Given the formal analogy between rotation and stratification \cite{veronis1970analogy}, internal (gravity) modes also obey an ill-posed problem in non-rotating stratified fluids \cite{rieutord1999analogy}. 
In the presence of stratification and rotation, the mathematical problem becomes even more complicated because IGMs obey a mixed hyperbolic-elliptic operator \cite{friedlander1982gafd}.
As for pure inertial \cite{brunet2019linear,favier2024super} and internal (gravity) waves \cite{maas1997observation}, IGWs can also converge after multiple reflections on solid boundaries to wave attractors \cite{pacary2023observation}. 
Attractors are interesting geometrical structures \cite{CdV2020attractors}, which are capable of focusing the wave energy at small length scales.
% (which could enhance dissipation in many systems).

This work builds upon and extends previous studies of IGMs in confined geometries.
Allen \cite{allen1971some} and Friedlander \& Siegmann \cite{friedlander1982jfm} paid attention to IGMs in spherical (and cylindrical) containers when gravity is constant and parallel to the rotation axis. 
Misaligned cases were later considered in \cite{friedlander1982gafd} with arbitrary gravity fields. 
Here, we aim to study IGMs using a simple model retaining the key ingredients to describe pancake-like geophysical vortices. 
Briefly, we consider a fluid ellipsoid subject to a constant (ambient) gravity field and uniformly rotating along an axis that is tilted with respect to gravity (to account for the full planetary rotation). 
This setup allows us to extend prior results about pure inertial modes in ellipsoids \cite{vantieghem2014inertial,ivers2017enumeration,backus2017completeness}, while preserving 
polynomial eigenvectors with stratification.
% an exact polynomial nature of the modal solutions with stratification. 
Finally, we characterise the wave spectrum using the mathematical theory recently presented in Colin de Verdi\`ere \& Vidal \cite{CdV2023spectrum} for inertial modes in ellipsoids. 
The manuscript is divided as follows. 
We first formulate the wave problem for general fluid volumes in \S\ref{sec:problem}.
Then, we assume that the domain is an ellipsoid in \S\ref{sec:finitedim} to model geophysical vortices. 
The wave spectrum for an ellipsoid is mathematically described and compared to numerical solutions in \S\ref{sec:results}.
The results are discussed in light of geophysical applications in \S\ref{sec:discussion}, and we end the paper in \S\ref{sec:ccl}.
Additional (technical) details are provided in the Electronic Supplementary Material (ESM).

%-----------------------------------------------------------------------
\section{Formulation of the general model}
\label{sec:problem}
\subsection{Primitive fluid dynamics equations}
%-----------------------------------------------------------------------
\begin{figure}
\centering
\begin{tabular}{cc}
    \begin{tabular}{c}
    \begin{overpic}[height=0.25\textwidth]
{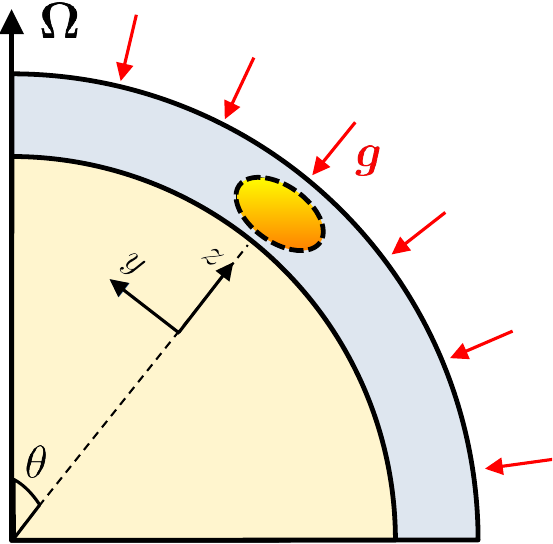}
    \put(-20,0){(a)}
    \end{overpic}
    \end{tabular} & 
    \begin{tabular}{c}
    \begin{overpic}[height=0.25\textwidth]
{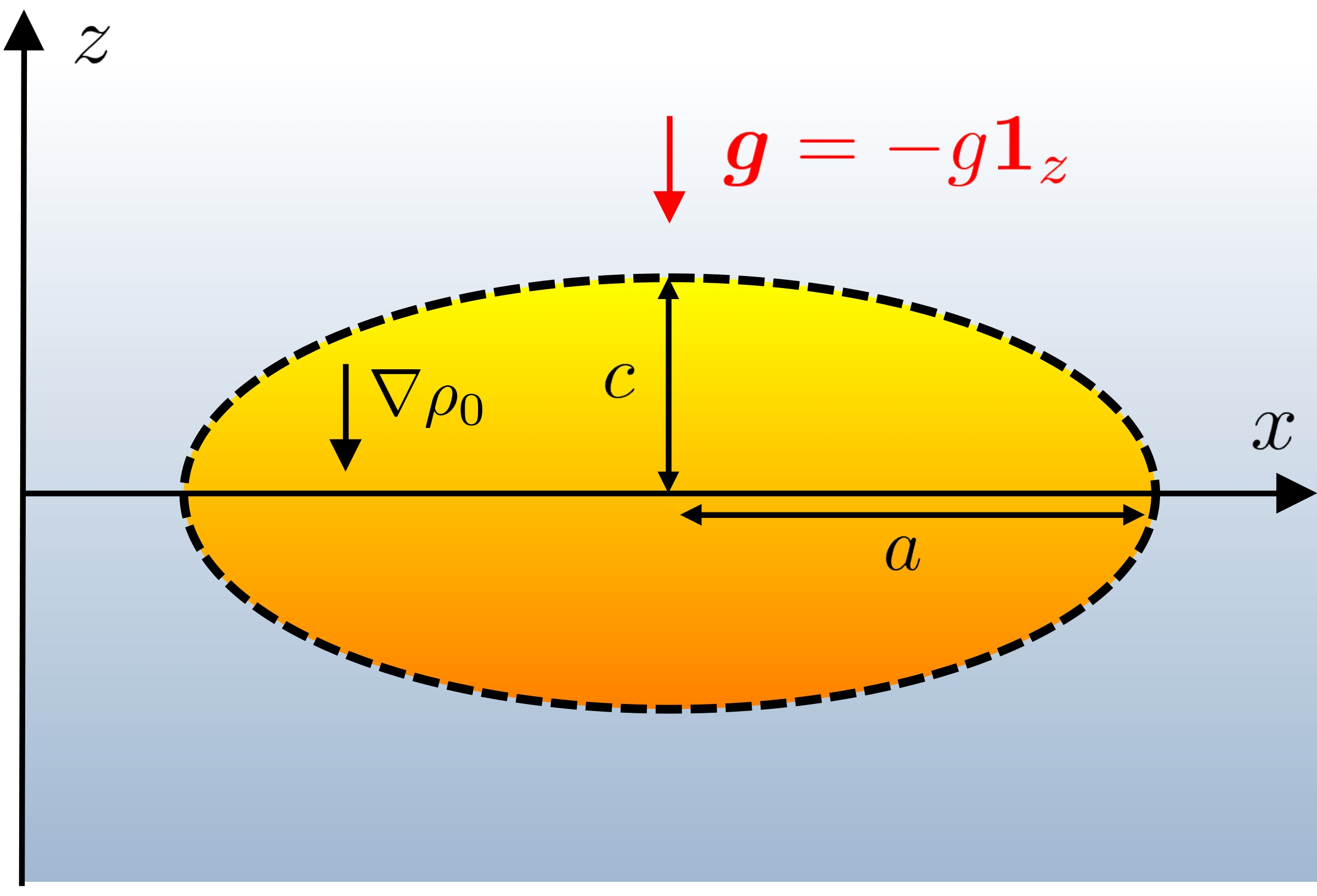}
    \put(-15,0){(b)}
    \end{overpic}
    \end{tabular}
\end{tabular}
\caption{Sketch of the mathematical model for geophysical vortices (not to scale). (a) Meridional cross-section of a vortex centred at colatitude $\theta$, which is embedded within an ocean or an atmosphere (blue region). (b) Front view of a flattened ellipsoidal vortex of semi-axes $a \gg b \gg c$, which is stratified under the uniform planet's gravity $\boldsymbol{g}$ (i.e. the background density $\rho_0$ decreases with increasing $z$).}
\label{fig:fig2}
\end{figure}

As illustrated in figure \ref{fig:fig2}(a), we model a vortex by a fluid domain of volume $V$ and (fictive) boundary $\partial V$, embedded within an ambient fluid (e.g. a planetary atmosphere or ocean) that is rotating at the angular velocity $\boldsymbol{\Omega} =\Omega_s \boldsymbol{1}_\Omega$ with respect to an inertial frame (where $\boldsymbol{1}_\Omega$ is a unit vector).
We work in the frame rotating at angular velocity $\boldsymbol{\Omega}$, and denote the position vector by $\boldsymbol{r} = (x,y,z)^\top$ using the Cartesian coordinates. 
The vortex is internally stratified in density under the action of gravity $\boldsymbol{g}$.
The vortex's stratification is quantified by the Brunt-V\"ais\"al\"a (BV) frequency $N$. 
Homogeneous fluids are such that $N^2=0$, whereas stably stratified fluids are characterised by $N^2 > 0$. 
As a starting point, we employ the Boussinesq approximation to account for density variations \cite{spiegel1960boussinesq}. 
The vortex is differentially rotating with respect to the ambient fluid outside $V$, which can be modelled by a background flow $\boldsymbol{U}_0$.  
Laboratory experiments \cite{aubert2012universal,lemasquerier2020remote} and numerical calculations \cite{hassanzadeh2012universal} show that, for incompressible fluids, $\boldsymbol{U}_0$ is often close to a uniform-vorticity flow. 
The latter can be sought as
$\boldsymbol{U}_0 = \boldsymbol{\omega} \times \boldsymbol{r} + \nabla \Psi$, where $\boldsymbol{\omega}$ is the vortex's rotation vector (either cyclonic or anticyclonic) and $\Psi$ is a scalar ensuring that the flow does not cross $\partial V$ (e.g. \cite{noir2013precession} in an ellipsoid). 
The amplitude of $\boldsymbol{U}_0$ is characterised by the Rossby number $Ro = |\boldsymbol{\omega}|/\Omega_s$, and we are interested in modelling stratified vortices in the regime $|Ro|\ll (N/\Omega_s)^2$.
Thus, we neglect the effects of $\boldsymbol{U}_0$ onto the background density and pressure. 
We further simplify the model by discarding baroclinic (and centrifugal effects), the very weak vortex's self-gravity, and the (weak) lateral variations of the planet's gravity at the size of a vortex.
Hence, the fluid is only subject to the ambient gravity $\boldsymbol{g} = -g(z) \boldsymbol{1}_z$ within $V$, where $\boldsymbol{1}_z$ is the unit vector along the vertical $z-$axis.
We denote by $\rho_*+\rho_0(z)$ the reference density within $V$ (with the mean density $\rho_*$), which is given by the hydrostatic equilibrium $\nabla P_0 =(\rho_*+\rho_0) \boldsymbol{g}$ where $P_0$ is the hydrostatic pressure. 

Now, we seek small-amplitude perturbations around the reference state $[\boldsymbol{U}_0, \rho_0, P_0]$. 
Without diffusion, the velocity $\boldsymbol{v}(\boldsymbol{r},t)$ obeys the linearised Euler equation given in the rotating frame by
\begin{equation}
    \partial_t \boldsymbol{v} + (\boldsymbol{U}_0 \boldsymbol{\cdot} \nabla ) \boldsymbol{v} + (\boldsymbol{v} \boldsymbol{\cdot} \nabla ) \boldsymbol{U}_0 + 2 \boldsymbol{\Omega} \times \boldsymbol{v} = - \nabla \pi + ({\rho}/{\rho_*}) \, \boldsymbol{g}
    \label{eq:Euler}
\end{equation}
together with the incompressible condition $\nabla \boldsymbol{\cdot} \boldsymbol{v} = 0$ for an incompressible fluid, and where $\pi = p/\rho_*$ is the reduced pressure. 
The density $\rho(\boldsymbol{r},t)$ is then governed by
\begin{equation}
 \partial_t \rho /\rho_* + \boldsymbol{U}_0 \cdot \nabla \rho/\rho_* = (\boldsymbol{v} \boldsymbol{\cdot} \boldsymbol{1}_z) \, N^2(z) / g(z),
 \label{eq:Mass}
\end{equation}
where $N^2(z) = \boldsymbol{g} \boldsymbol{\cdot} \nabla (\rho_0/\rho_*)$ is the BV frequency within the Boussinesq approximation.
The background flow has two main effects in the linear theory.
First, it is responsible for a shift in frequency of the wave motions (e.g. see in \cite{lam2023supply} for a discussion in unbounded fluids). 
Second, it can be responsible for the onset of hydrodynamic instabilities. 
In particular, when $\boldsymbol{U}_0$ is an elliptical flow \cite{aubert2012universal,lemasquerier2020remote}, the so-called elliptical instability could develop \cite{kerswell2002elliptical} and be responsible for the transition towards a wave turbulence regime \cite{le2017inertial}. 
For mathematical simplicity, we assume $\boldsymbol{U}_0 = \boldsymbol{0}$ in equations (\ref{eq:Euler})-(\ref{eq:Mass}) to focus on the free wave motions given by
\begin{subequations}
\label{eq:EulerMasswaves}
\begin{equation}
    \partial_t \boldsymbol{v} + 2 \boldsymbol{\Omega} \times \boldsymbol{v} = - \nabla \pi - N(z) \xi \, \boldsymbol{1}_z, \quad \nabla \boldsymbol{\cdot} \boldsymbol{v} = 0, \quad \partial_t \xi = (\boldsymbol{v} \boldsymbol{\cdot} \boldsymbol{1}_z) \, N(z),
    \tag{\theequation a--c}
\end{equation}
\end{subequations}
where we have introduced $\xi = (\rho/\rho_*) \,  g(z)/N(z)$.
The background flow will be reintroduced in future work (to investigate the outcome of the elliptical instability).

Finally, equations (\ref{eq:EulerMasswaves}a-c) are supplemented with boundary conditions (BCs).
We impose that the fields are regular at $\boldsymbol{r}=\boldsymbol{0}$.
Then, realistic conditions would be to couple the dynamics inside the vortex with the exterior flow on $\partial V$. 
This would amount to considering continuity conditions across $\partial V$, and decaying solutions at infinity (e.g for Rankine \cite{park2013instabilities} or Lamb-Oseen \cite{le2008inviscid} vortices).
For mathematical simplicity, we here neglect the coupling with the exterior fluid and assume that $\partial V$ is a rigid boundary.
Consequently, the velocity obeys the no-penetration BC $\left . \boldsymbol{v} \boldsymbol{\cdot} \boldsymbol{n} \right |_{\partial V} = 0$ on $\partial V$, where $\boldsymbol{n}$ is the unit (outward) vector normal to the boundary. 
Without diffusion, no other physical BCs have to be explicitly enforced in the model.
Indeed, the density and pressure BCs automatically follow from the velocity BC. 
In practice, the pressure satisfies a mixed BC that is given by the continuity of the normal component of equation (\ref{eq:EulerMasswaves}a) on $\partial V$ (see the ESM).

%-----------------------------------------------------------------------
\subsection{Mixed hyperbolic-elliptic problem}
%-----------------------------------------------------------------------
We seek modal solutions of equations (\ref{eq:EulerMasswaves}a-c) given by 
\begin{equation}
    \left [ \boldsymbol{v}, \pi, \xi \right ] (\boldsymbol{r},t) = \left [ \boldsymbol{u}, \Phi, \zeta \right ] (\boldsymbol{r}) \exp(\mathrm{i} \omega t),
\end{equation}
where $\left [ \boldsymbol{u}, \Phi, \zeta \right ]$ are complex-valued amplitudes depending on space, and with $\omega \in \mathbb{C}$. 
The wave properties are deeply tied to the nature of the mathematical problem. 
Instead of solving the primitive equations, it is worth considering a master equation that governs the evolution of either $\boldsymbol{u}$ or $\Phi$.  
If we multiply Euler equation (\ref{eq:EulerMasswaves}a) by $\mathrm{i} \omega$ and replace $\mathrm{i} \omega \zeta$ by buoyancy equation (\ref{eq:EulerMasswaves}c), we obtain for $\omega \neq 0$
\begin{subequations}
\label{eq:waveeqdt}
\begin{equation}
    - \omega^2 \boldsymbol{u} + 2 \boldsymbol{\Omega} \times (\mathrm{i} \omega \boldsymbol{u}) + N^2 (z) (\boldsymbol{u} \boldsymbol{\cdot} \boldsymbol{1}_z) \boldsymbol{1}_z = - \mathrm{i} \omega \nabla \Phi, \quad \nabla \boldsymbol{\cdot} \boldsymbol{u} = 0,
    \tag{\theequation a,b}
\end{equation}
\end{subequations}
Equations (\ref{eq:waveeqdt}a,b) and the no-penetration BC $\left . \boldsymbol{u} \boldsymbol{\cdot} \boldsymbol{n} \right |_{\partial V} = 0$ show that the velocity is an appropriate variable to explore the wave properties. 
Note that a similar equation can be obtained for the fluid particle displacement vector \cite{friedlander1985internal}.
On the contrary, non-oscillatory modes with $\omega=0$ require a specific consideration. 
For unstratified fluids with $N=0$, they are called geostrophic modes because they obey the geostrophic balance 
$2 \boldsymbol{\Omega} \times \boldsymbol{u} = - \nabla \Phi$.
When $N(z) > 0$, the steady modes are two-dimensional (2D) such that $\boldsymbol{u} \boldsymbol{\cdot}~\boldsymbol{1}_z=0$, and the density perturbation obeys the diagnostic equation $\rho_* \left (2 \boldsymbol{\Omega} \boldsymbol{\cdot} \nabla \boldsymbol{u}\right ) =~\nabla \zeta \times \boldsymbol{1}_z$ (resulting from the thermal wind balance).

Because of historical reasons (dating back to Poincar\'e \cite{poincare1885equilibre}), the problem is generally recast into a scalar equation for the pressure $\Phi$ \cite{friedlander1982jfm,friedlander1982gafd}. 
As shown in the ESM, the pressure is governed in $V$ by a second-order equation whose higher-order terms are given by
\begin{equation}
[ \omega^2 - N^2(z) ] \nabla^2 \Phi + N^2(z) \left ( \boldsymbol{1}_z \boldsymbol{\cdot} \nabla \right )^2 \Phi - \left ( 2 \boldsymbol{\Omega} \boldsymbol{\cdot} \nabla \right )^2 \Phi + \dots = 0
\label{eq:laplaP}
\end{equation}
when $|\omega| \neq~ \omega_\pm(z)$ with
\begin{equation}
2\omega_{\pm}^2(z) = [ N^2(z) + 4 |\boldsymbol{\Omega}|^2 ] \pm \sqrt{\left [ N^2(z) + 4|\boldsymbol{\Omega}|^2 \right ]^2 - 16 N^2(z) (\boldsymbol{\Omega} \boldsymbol{\cdot} \boldsymbol{1}_z)^2}.
\label{eq:omegapmleblonc}
\end{equation}
The lower-order terms in equation (\ref{eq:laplaP}), which vanish when $N$ is constant, are here omitted for concision but are given in the ESM. 
When $|\omega| = \omega_\pm(z)$, the pressure instead obeys a first-order equation (see the ESM). 
Finally, the pressure obeys a mixed BC on $\partial V$ (i.e. neither a pure Dirichlet nor Neumann BC), which is difficult to take into account for numerical computations.
However, the nature of the mathematical problem can be simply determined from the pressure equation by using microlocal analysis \cite{grubb2008distributions}.
This is a branch of mathematics, which studies the solutions of PDE using (notably) asymptotic geometric techniques. 
Microlocal tools have already proven useful in physics to study wave problems \cite{faure2023topo,venaille2023ray}.
In particular, the nature of equation (\ref{eq:laplaP}) is governed by a scalar quantity that is called the principal symbol.
The latter is obtained from the equation by freezing the variable coefficients and keeping only the highest-degree terms. 
In practice, it amounts to transforming the derivatives as
$\nabla \leftrightarrow~\mathrm{i} \boldsymbol{k}$, where $\boldsymbol{k}=~(k_1,k_2,k_3)^\top$ is a real-valued wave vector at the position $(x,y,z)$.
The principal symbol is then a homogeneous polynomial in $\boldsymbol{k}$, which encapsulates many properties of the spectral problem \cite{CdV2020attractors,colin2020spectral}. 
The problem is elliptic when the principal symbol is invertible, and hyperbolic when the principal symbol vanishes. 
The vanishing condition on the principal symbol gives
\begin{equation}
    |\boldsymbol{k}|^2 \omega^2 - \left [ N^2(z) |\boldsymbol{1}_z \times \boldsymbol{k}|^2 + (2\boldsymbol{\Omega} \boldsymbol{\cdot} \boldsymbol{k})^2 \right ] = 0,
    \label{eq:dispersionIGW}
\end{equation}
which is the dispersion relation of IGWs in unbounded fluids when $N$ is constant \cite{leblond1981waves}.
The wave-like equation is thus hyperbolic when equation (\ref{eq:dispersionIGW}) is satisfied for some non-zero wave vectors $\boldsymbol{k} \in \mathbb{R}^3$. 
Otherwise, it is elliptic (except when $\omega=0$). 
Next, the bounds for the hyperbolic domain can be obtained from equation (\ref{eq:dispersionIGW}).
The latter defines a quadric cone given by $(\boldsymbol{\mathfrak{g}}^\star)^{ij} k_i k_j = 0$ (using Einstein notation) when $\omega \neq 0$, where $\boldsymbol{\mathfrak{g}}^\star$ is the associated metric tensor.
The nature of this quadric cone depends on the eigenvalues of $\boldsymbol{\mathfrak{g}}^\star$, which are given by $\mu_0 = \omega^2 - N^2(z)$ and $\mu_\pm = \omega^2 - \omega_\pm^2(z)$. 
When the eigenvalues are non-zero and of the same sign, the surface is elliptic. 
On the contrary, the surface is hyperbolic when the eigenvalues are non-zero but of different signs. 
The above microlocal description agrees with the analysis presented in Section 3 of Friedlander \& Siegmann \cite{friedlander1982gafd}.
The wave-like equation is thus hyperbolic in $V$ when $\omega_-(z) <|\omega| < \omega_+(z)$, whereas it is elliptic in $V$ when $0 < |\omega| < \omega_-(z)$ or $|\omega| > \omega_+(z)$. 

%-----------------------------------------------------------------------
\subsection{Spectral problem for the velocity}
%-----------------------------------------------------------------------
Further wave properties can be obtained by considering the velocity equation. 
We introduce the Hilbert space $\boldsymbol{\mathcal{V}}$ spanned by complex-valued vector fields that are square-integrable in $V$. 
The complex-valued inner product in this Hilbert space is given by 
\begin{equation}
    \langle \boldsymbol{u},\boldsymbol{v}\rangle := \int_V \boldsymbol{u}^\dagger \boldsymbol{\cdot} \boldsymbol{v} \, \mathrm{d}V,
    \label{eq:innerproduct}
\end{equation}
where ${}^\dagger$ is the complex conjugate.
The associated norm is $||\boldsymbol{u}|| = \langle \boldsymbol{u}, \boldsymbol{u} \rangle ^{1/2}$. 
Then, we denote by $\boldsymbol{\mathcal{V}}^0 \subset \boldsymbol{\mathcal{V}}$ the closed subspace that is orthogonal, with respect to inner product (\ref{eq:innerproduct}), to the space spanned by vector fields made of gradients of smooth functions. 
A smooth element in $\boldsymbol{\mathcal{V}}^0$ is divergenceless and tangent to the boundary \cite{backus2017completeness}, that is
\begin{equation}
\boldsymbol{\mathcal{V}}^0 := \{\boldsymbol{u} \in \boldsymbol{\mathcal{V}} \ | \ \nabla \boldsymbol{\cdot} \boldsymbol{u} = 0 \ \text{in V}, \ \boldsymbol{u} \boldsymbol{\cdot} \boldsymbol{n} = 0 \ \text{on $\partial V$} \}.
\end{equation}
Finally, we introduce the orthogonal projector $\mathbb{L}$ from $\boldsymbol{\mathcal{V}}$ onto  $\boldsymbol{\mathcal{V}}^0$ (called Leray projector after \cite{leray1934mouvement}). 
For any vector $\boldsymbol{e} \in \boldsymbol{\mathcal{V}}$ written as $\boldsymbol{e} = \nabla \Psi + \nabla \times \boldsymbol{A}$ (using a Helmholtz decomposition) for some potentials $[\Psi, \boldsymbol{A}]$, the projected vector is defined by $\mathbb{L}(\boldsymbol{e}):= \boldsymbol{e} - \nabla \Psi$. 
This projector is used in the mathematical analysis of incompressible flows \cite{foias2001navier}, but also in some numerical studies \cite{teed2023solenoidal}.

Equipped with the above definitions, we can seek solutions of equation (\ref{eq:waveeqdt}) with $\boldsymbol{u} \in \boldsymbol{\mathcal{V}}^0$ in ellipsoids. 
We apply the orthogonal projector to the equation to remove the pressure term, and find that $(\omega, \boldsymbol{u})$ obey a quadratic eigenvalue problem (QEP). 
The latter is given by 
\begin{subequations}
\label{eq:QEP}
\begin{equation}
    \boldsymbol{\mathcal{Q}}_\omega (\boldsymbol{u}) = \boldsymbol{0}, \quad 
    \boldsymbol{\mathcal{Q}}_\omega := -\omega^2 \boldsymbol{\mathcal{I}} +  \omega \mathrm{i}\boldsymbol{\mathcal{C}}  + \boldsymbol{\mathcal{K}},
    \tag{\theequation a,b}
\end{equation}
\end{subequations}
where $\boldsymbol{\mathcal{I}}$ is the identity operator, and with the two operators
\begin{subequations}
\label{eq:CKop}
\begin{equation}
\mathrm{i} \boldsymbol{\mathcal{C}} (\boldsymbol{u}) := \mathrm{i} \mathbb{L} [2 \boldsymbol{\Omega} \times \boldsymbol{u}], \quad \boldsymbol{\mathcal{K}} (\boldsymbol{u}) := \mathbb{L} [N^2 (z) (\boldsymbol{u} \boldsymbol{\cdot} \boldsymbol{1}_z) \boldsymbol{1}_z].
\tag{\theequation a,b}
\end{equation}
\end{subequations}
The operator $\mathrm{i} \boldsymbol{\mathcal{C}}$, called the Poincar\'e operator \cite{CdV2023spectrum}, is a bounded self-adjoint operator in $\boldsymbol{\mathcal{V}}^0$ \cite{backus2017completeness}. 
The buoyancy operator $\boldsymbol{\mathcal{K}}$ has properties given by Proposition \ref{theo:operatorK}. 
\begin{proposition}
The operator $\boldsymbol{\mathcal{K}}$ is bounded, self-adjoint and positive in $\boldsymbol{\mathcal{V}}^0$.
\label{theo:operatorK}
\end{proposition}
\begin{proof}
We must show that $\boldsymbol{\mathcal{K}}$ is symmetric, bounded, and positive in $\boldsymbol{\mathcal{V}}^0$.
We verify that $\boldsymbol{\mathcal{K}}$ is symmetric since $\langle \boldsymbol{a}, \boldsymbol{\mathcal{K}}(\boldsymbol{b}) \rangle = \langle \boldsymbol{\mathcal{K}}(\boldsymbol{a}), \boldsymbol{b} \rangle$ for any $(\boldsymbol{a},\boldsymbol{b}) \in \boldsymbol{\mathcal{V}}^0 \times \boldsymbol{\mathcal{V}}^0$. 
Moreover, we have 
\begin{equation*}
||\boldsymbol{\mathcal{K}}(\boldsymbol{u})||^2 = \langle \boldsymbol{\mathcal{K}}(\boldsymbol{u}), N^2 (\boldsymbol{u} \boldsymbol{\cdot} \boldsymbol{1}_z) \boldsymbol{1}_z \rangle - \langle \boldsymbol{\mathcal{K}}(\boldsymbol{u}), \nabla \Phi \rangle = \langle \boldsymbol{\mathcal{K}}(\boldsymbol{u}), N^2 (\boldsymbol{u} \boldsymbol{\cdot} \boldsymbol{1}_z) \boldsymbol{1}_z \rangle
\end{equation*}
since $\langle \boldsymbol{\mathcal{K}}(\boldsymbol{u}), \nabla \Phi \rangle = 0$ using the definition of $\mathbb{L}$.
Using Cauchy-Schwarz inequality, we then obtain $||\boldsymbol{\mathcal{K}}(\boldsymbol{u})||^2 \leq ||\boldsymbol{\mathcal{K}}(\boldsymbol{u})|| \, ||N^2 (\boldsymbol{u} \boldsymbol{\cdot} \boldsymbol{1}_z) \boldsymbol{1}_z|| \leq \max(N^2) ||\boldsymbol{\mathcal{K}}(\boldsymbol{u})|| \, ||\boldsymbol{u}||$, showing that $\boldsymbol{\mathcal{K}}$ is bounded. 
Finally, $\boldsymbol{\mathcal{K}}$ is positive because $\langle \boldsymbol{u}, \boldsymbol{\mathcal{K}}(\boldsymbol{u}) \rangle \geq 0$ if $N^2 \geq 0$.
\end{proof}

We want to characterise the spectrum $\sigma$ of the quadratic spectral form $\boldsymbol{\mathcal{Q}}_\omega$ in ellipsoids, which is the set of complex numbers $\omega$ for which $\boldsymbol{\mathcal{Q}}_\omega$ is not continuously invertible. 
Actually, the spectral problem $\boldsymbol{\mathcal{Q}}_\omega(\boldsymbol{u})=\boldsymbol{0}$ can be converted into a standard eigenvalue problem (SEP) for a bounded self-adjoint operator in $\boldsymbol{\mathcal{V}}^0 \times \boldsymbol{\mathcal{V}}^0$.

\begin{proof}
\label{theo:QEP2SEP}
This is a general result for spectral families of the form $\boldsymbol{\mathcal{Q}}_\omega := -\omega^2 + \omega \boldsymbol{\mathcal{A}} + \boldsymbol{\mathcal{B}}$ acting on a Hilbert space $\boldsymbol{\mathcal{H}}$, where $[\boldsymbol{\mathcal{A}},\boldsymbol{\mathcal{B}}]$ are bounded self-adjoint operators and with $\boldsymbol{\mathcal{B}} \geq 0$ \cite{barston1967eigenvalueII,valette1989etude}.
When $\omega\neq 0$, the QEP $\boldsymbol{\mathcal{Q}}_\omega(\boldsymbol{u})=\boldsymbol{0}$ can be recast as a SEP for the operators $\boldsymbol{\mathcal{L}}$ or $\boldsymbol{\mathcal{T}}$  given by 
\begin{subequations}
\label{eq:selfadjointT}
\begin{equation}
    \boldsymbol{\mathcal{L}} := \begin{pmatrix}
0 & \boldsymbol{\mathcal{I}} \\
\boldsymbol{\mathcal{B}} & \boldsymbol{\mathcal{A}} \end{pmatrix}, \quad \boldsymbol{\mathcal{T}}:= \begin{pmatrix}
    \boldsymbol{\mathcal{A}} & \boldsymbol{\mathcal{B}}^{1/2} \\
    \boldsymbol{\mathcal{B}}^{1/2} & 0 \\
  \end{pmatrix},
    \tag{\theequation a,b}
\end{equation}
\end{subequations}
which are both acting on the Hilbert space $\boldsymbol{\mathcal{H}} \times \boldsymbol{\mathcal{H}}$ equipped with the inner product 
$\langle \langle \boldsymbol{z},\boldsymbol{\zeta} \rangle \rangle = \langle \boldsymbol{z}_1, \boldsymbol{\zeta}_1 \rangle + \langle \boldsymbol{z}_2, \boldsymbol{\zeta}_2 \rangle$ with $\boldsymbol{z}=(\boldsymbol{z}_1, \boldsymbol{z}_2)^\top$ and $\boldsymbol{\zeta}=(\boldsymbol{\zeta}_1, \boldsymbol{\zeta}_2)^\top$. 
Here, $\boldsymbol{\mathcal{T}}$ is a bounded self-adjoint operator (where $\boldsymbol{B}^{1/2}$ is the square root of $\boldsymbol{B}$, which is also self-adjoint \cite{wouk1966note}), whereas $\boldsymbol{\mathcal{L}}$ has no particular symmetries. 
The spectra of $\boldsymbol{\mathcal{Q}}_\omega$, $\boldsymbol{\mathcal{L}}$ and $\boldsymbol{\mathcal{T}}$ are identical outside $\omega=0$.
\end{proof}
\noindent The spectrum $\sigma(\boldsymbol{\mathcal{Q}}_\omega)$ is thus the disjoint union of the point spectrum given by
\begin{equation}
    \sigma_P (\boldsymbol{\mathcal{Q}}_\omega) := \{ \omega \in \mathbb{C} \ | \ \exists \, \boldsymbol{u} \neq \boldsymbol{0}  \ \, \boldsymbol{\mathcal{Q}}_\omega (\boldsymbol{u}) = \boldsymbol{0} \},
\end{equation}
which is the subset spanned by the eigenvalues of QEP (\ref{eq:QEP}), and the remaining continuous spectrum (i.e. the set of $\omega$ for which $\boldsymbol{\mathcal{T}}-\omega \boldsymbol{\mathcal{I}}$ is injective and has a dense range, but is not surjective).
This mathematical distinction is of direct physical interest to IGMs.
Indeed, the point spectrum is associated with (regular) square-integrable eigenvectors in diffusionless fluids.
On the contrary, a non-empty continuous spectrum is often characterised by (almost) periodic attracting trajectories obtained from ray tracing techniques \cite{CdV2020attractors}.  
These geometrical structures, called attractors \cite{CdV2020attractors}, are associated with nearly singular (i.e. non-square-integrable) velocity fields when diffusion is vanishingly small \cite{dintrans1999gravito}. 
Given the symmetries of the different operators, some general properties of the spectrum are given in Proposition \ref{theo:propQEP}.

\begin{proposition} 
\label{theo:propQEP}
The spectrum of QEP (\ref{eq:QEP}) is real-valued and even.
Moreover, the spectrum is bounded by $|\omega| \leq \omega_{\max}$ with $\omega_{\max} \leq \Omega_s + \sqrt{\Omega_s^2 + \max(N^2)}$. 
Finally, two eigen-pairs $(\omega_1, \boldsymbol{u}_1)$ and $(\omega_2, \boldsymbol{u}_2)$ with $\omega_1 \neq \omega_2$ are not orthogonal with respect to inner product (\ref{eq:innerproduct}), but satisfy the two integral properties \begin{subequations}
\label{eq:orthogonalityIGM}
\begin{equation}
    \omega_1 \omega_2 \langle \boldsymbol{u}_2, \boldsymbol{u}_1 \rangle = - \langle \boldsymbol{u}_2, \boldsymbol{\mathcal{K}}(\boldsymbol{u}_1) \rangle, \quad (\omega_2 + \omega_1) \langle \boldsymbol{u}_2, \boldsymbol{u}_1 \rangle = \langle \boldsymbol{u}_2, \mathrm{i} \boldsymbol{\mathcal{C}}(\boldsymbol{u}_1)  \rangle.
    \tag{\theequation a,b}
\end{equation}
\end{subequations}
\end{proposition}
\begin{proof}
$\omega \in \mathbb{R}$ result from the fact that the operator $\boldsymbol{\mathcal{T}}$ defined in equation (\ref{eq:selfadjointT}b) is self-adjoint when $\boldsymbol{\mathcal{K}} \geq 0$. 
Then, the complex-conjugate of QEP (\ref{eq:QEP}) gives
\begin{equation*}
    -\omega^2 \boldsymbol{u}^\dagger - \omega \mathrm{i} \boldsymbol{\mathcal{C}}(\boldsymbol{u}^\dagger) + \boldsymbol{\mathcal{K}}(\boldsymbol{u}^\dagger) = \boldsymbol{0},
\end{equation*}
which shows that $(-\omega, \boldsymbol{u}^\dagger)$ is an eigen-pair of the QEP. 
The upper bound of $\omega_{\max}$ can be obtained using a min-max principle \cite{barston1967eigenvalueII}. 
Finally, the orthogonality conditions can be found using vector manipulations of the QEP  \cite{barston1967eigenvalue}. 
\end{proof}

%-----------------------------------------------------------------------
\section{Ellipsoidal model for a constant BV frequency}
\label{sec:finitedim}
%-----------------------------------------------------------------------
We now simplify the general model to consider an idealised configuration of geophysical vortices. 
To get physical insight into the problem, it is useful to consider a shape that is amenable to mathematical analysis.
For simplicity, we assume that the boundary $\partial V$ is a triaxial ellipsoid of semi-axes $[a,b,c]$, where the $c-$axis is aligned with gravity (figure \ref{fig:fig2}b). 
The use of ellipsoids has a long-standing history in the study of geophysical vortices \cite{mckiver2015ellipsoidal}, and is supported by recent fluid dynamic experiments \cite{aubert2012universal,lemasquerier2020remote}.
In the rotating frame, the boundary is expressed using the Cartesian coordinates as $F(x,y,z) = 1$ with
$F(x,y,z) := (x/a)^2 + (y/b)^2 + (z/c)^2$, such that the unit normal is $\boldsymbol{n} :=~ \nabla F / ||\nabla F||$. 
Moreover, geophysical vortices with $N \gg 2 \Omega_s$ are usually strongly flattened with $a,b \gg c$ (e.g. Jupiter's GRS is expected to have a depth of only a few hundreds of km \cite{lemasquerier2020remote}). 
In such vortices, the ambient gravity is expected to weakly varies with depth. 
% (e.g. by less than a few percent in the outermost Jupiter's atmosphere).
Thus, we also assume that $\boldsymbol{g} = -g\boldsymbol{1}_z$ is uniform and consider a background density $\rho_* + \rho_0(z)$ that varies linearly with $z$ as
\begin{equation}
\rho_0 (z) /\rho_* = -(N^2/g) \, z,
\label{eq:rho0LN}
\end{equation}
where $N$ is a constant BV frequency. 
With these assumptions, the eigensolutions can be characterised using finite-dimensional arguments in ellipsoids. 

%-----------------------------------------------------------------------
\subsection{Polynomial description}
%-----------------------------------------------------------------------
Since an ellipsoid is a quadratic surface and expression (\ref{eq:rho0LN}) is a polynomial, we can seek eigenvectors using polynomial functions in the Cartesian coordinates. 
We denote by $\mathcal{P}_n: \mathbb{R}^3 \to \mathbb{C}$ the space spanned by the scalar Cartesian monomials $x^i y^j z^k$ of degree $i+j+k$ smaller or equal to $n$ with $n\geq 1$.
Its dimension is $N(\mathcal{P}_n) = (n+1)(n+2)(n+3)/6$. 
Then, we denote by $\boldsymbol{\mathcal{V}}_n$ the space of vectors in ellipsoids with coefficients in $\mathcal{P}_n$ and define the subspace $\boldsymbol{\mathcal{V}}_n^0 := \boldsymbol{\mathcal{V}}_n \cap \boldsymbol{\mathcal{V}}^0$ of dimension $N(\boldsymbol{\mathcal{V}}_n^0) = n(n+1)(2n+7)/6$, whose elements are divergenceless vector fields that are tangent to the boundary $\partial V$.
Finally, we introduce the subspaces $\boldsymbol{\mathcal{W}}_n^0 :=~ (\boldsymbol{
\mathcal{V}}_{n-1}^0)^\perp \cap \boldsymbol{\mathcal{V}}_n^0$ of dimension $N(\boldsymbol{\mathcal{W}}_n^0) = n(n+2)$, which are spanned by vectors in $\boldsymbol{\mathcal{V}}_n^0$ of degree $n$ that are orthogonal to polynomial vectors of degree $\leq n-1$. 
These polynomial spaces are key to characterise the spectrum because of Proposition \ref{theorem:invariancePcPk}.

\begin{proposition}
The spaces $\boldsymbol{\mathcal{V}}^0_n$ with $n \geq 1$ are invariant by the action of $\boldsymbol{\mathcal{C}}$ and $\boldsymbol{\mathcal{K}}$ in triaxial ellipsoids, that is $\boldsymbol{\mathcal{C}}_{|\boldsymbol{\mathcal{V}}^0_n} \subseteq \boldsymbol{\mathcal{V}}^0_n$ and $\boldsymbol{\mathcal{K}}_{|\boldsymbol{\mathcal{V}}^0_n} \subseteq \boldsymbol{\mathcal{V}}^0_n$.
% Consequently, we also have $\boldsymbol{\mathcal{C}}_{|\boldsymbol{\mathcal{W}}^0_n} \subseteq \boldsymbol{\mathcal{W}}^0_n$ and $\boldsymbol{\mathcal{K}}_{|\boldsymbol{\mathcal{W}}^0_n} \subseteq \boldsymbol{\mathcal{W}}^0_n$.
\label{theorem:invariancePcPk}
\end{proposition}
\begin{proof}
The invariance of $\boldsymbol{\mathcal{V}}^0_n$ under the action of 
$\boldsymbol{\mathcal{C}}$ has been demonstrated elsewhere \cite{backus2017completeness,CdV2023spectrum}. Here, we can use similar arguments to show that it is invariant by the action of $\boldsymbol{\mathcal{K}}$. 
Briefly, we must prove that the Leray projector $\mathbb{L}$ is well-defined on $\boldsymbol{\mathcal{V}}_n^0$.
This is ensured by Theorem 2.4 in \cite{axler2018neumann}, which shows that there is a unique scalar $\Psi \in \mathcal{P}_{n+1}$ in the definition of $\mathbb{L}$. 
% Since $\boldsymbol{\mathcal{V}}_n^0 :=~\oplus_{n\in \mathbb{N}^*} \boldsymbol{\mathcal{W}}_n^0$, we have $\boldsymbol{\mathcal{C}}_{|\boldsymbol{\mathcal{V}}^0_n} \subseteq \boldsymbol{\mathcal{V}}^0_n$ and $\boldsymbol{\mathcal{K}}_{|\boldsymbol{\mathcal{V}}^0_n} \subseteq \boldsymbol{\mathcal{V}}^0_n$. 
\end{proof}

We deduce from Proposition \ref{theorem:invariancePcPk} that $\boldsymbol{\mathcal{V}}^0_n$ is invariant by the action of $\boldsymbol{\mathcal{Q}}_\omega$.
Moreover, since $\oplus_{n\in \mathbb{N}^*} \boldsymbol{\mathcal{V}}_n^0$ is dense in $\boldsymbol{\mathcal{V}}^0$ \cite{CdV2023spectrum}, it shows that $\sigma(\boldsymbol{\mathcal{Q}}_\omega)$ is pure point in ellipsoids (i.e. $\sigma = \sigma_P$). 
The continuous spectrum is thus empty, and equation (\ref{eq:waveeqdt}) admits square-integrable eigenvectors in $\boldsymbol{\mathcal{V}}_n^0$. 
From a physical viewpoint, an empty continuous spectrum means that there are no singular velocity fields associated with wave attractors.
This is a difference with other geometries, in which wave attractors can exist even when $N$ is constant (e.g. in a trapezoid \cite{pacary2023observation}).
However, this does not preclude the existence of (almost) periodic wave trajectories using ray tracing techniques. 
The classical example is the pure inertial wave problem within a sphere.
Periodic trajectories have been found in this geometry using ray theory \cite{rabitti2014inertial} but, since the continuous spectrum is also empty in a full sphere\cite{ivers2015enumeration,backus2017completeness}, these structures are not associated with singular global modes for vanishingly small diffusion.  
Here, because the spectrum is pure point, we can solve the eigenvalue problem by restricting $\boldsymbol{u}$ to $\boldsymbol{\mathcal{V}}_n^0$ in ellipsoids.

%-----------------------------------------------------------------------
\subsection{Galerkin algorithm}
%-----------------------------------------------------------------------
An accurate numerical description can be developed by 
seeking $\boldsymbol{u} \in \boldsymbol{\mathcal{V}}_n^0$ as $\boldsymbol{u} = \sum_{j\geq 1} \alpha_j \boldsymbol{e}_j$, where $\{\alpha_j\}_{j\geq 1}$ are complex-valued coefficients, and $\{\boldsymbol{e}_j\}_{j\geq 1}$ are real-valued basis vector elements of $\boldsymbol{\mathcal{V}}_n^0$ \cite{lebovitz1989stability} normalised such that $||\boldsymbol{e}_j|| = 1$. 
Then, we substitute the above expansion into the QEP $\boldsymbol{\mathcal{Q}}_\omega(\boldsymbol{u})=\boldsymbol{0}$, and project the entire equation on every $\boldsymbol{e}_i$ using inner-product (\ref{eq:innerproduct}) to cancel out the spatial dependence (Galerkin method \cite{lebovitz1989stability}).
This gives the finite-dimensional QEP
\begin{equation}
\left ( -\omega^2 \boldsymbol{M}_n + \omega \mathrm{i} \boldsymbol{C}_n + \boldsymbol{K}_n \right ) \boldsymbol{\alpha} = \boldsymbol{0}
\label{eq:QEPVn0}
\end{equation}
where $\boldsymbol{\alpha}=(\alpha_1,\alpha_2,\dots)^\top$ is the eigenvector, and where $[\boldsymbol{M}_n,\boldsymbol{C}_n,\boldsymbol{K}_n]$ are three real-valued matrices of size $N({\boldsymbol{\mathcal{V}}_n^0}) \times N({\boldsymbol{\mathcal{V}}_n^0})$ and of elements 
$M_{n,ij} = \langle \boldsymbol{e}_i, \boldsymbol{e}_j \rangle$, $C_{n,ij} = \langle \boldsymbol{e}_i, 2 \boldsymbol{\Omega} \times \boldsymbol{e}_j \rangle$ and $K_{n,ij} =~ \langle \boldsymbol{e}_i, N^2 (\boldsymbol{e}_j \boldsymbol{\cdot} \boldsymbol{1}_z) \boldsymbol{1}_z \rangle$. 
The matrix $\boldsymbol{M}_n$ is positive definite but non-diagonal (because the chosen basis elements are not mutually orthogonal by construction in triaxial ellipsoids \cite{lebovitz1989stability}).
The matrix $\boldsymbol{C}_n$ is antisymmetric, and $\boldsymbol{K}_n$ is a real-valued positive matrix. 
Further properties of $[\boldsymbol{C}_n, \boldsymbol{K}_n]$ are given below in Proposition \ref{theo:dimkerCK}.

\begin{proposition}
\label{theo:dimkerCK}
$\dim (\ker \boldsymbol{C}_n) = \lceil n/2 \rceil$ and $\dim (\ker \boldsymbol{K}_n) = n(n+1)(n+2)/6$.
\end{proposition}
\begin{proof}
When $N=0$, $\ker (\boldsymbol{C}_n)$ is associated with  geostrophic modes of degree $n$. 
Its dimension was proved in \cite{CdV2023spectrum}. 
When $N\neq 0$, $\ker (\boldsymbol{K}_n)$ is associated with 2D motions such that $\boldsymbol{u} \boldsymbol{\cdot} \boldsymbol{1}_z = 0$.
The latter can be sought as $\boldsymbol{u} = \nabla \times (\Psi \boldsymbol{1}_z)$ where $\Psi = p [1-F(x,y,z)]$ is a stream function and with $p \in \mathcal{P}_{n-1}$. 
The dimension of $\ker(\boldsymbol{K}_n)$ is thus $\dim (\mathcal{P}_{n-1}) = n(n+1)(n+2)/6$. 
\end{proof}

\begin{figure}
\centering
\begin{tabular}{cc}
    \includegraphics[width=0.4\textwidth]{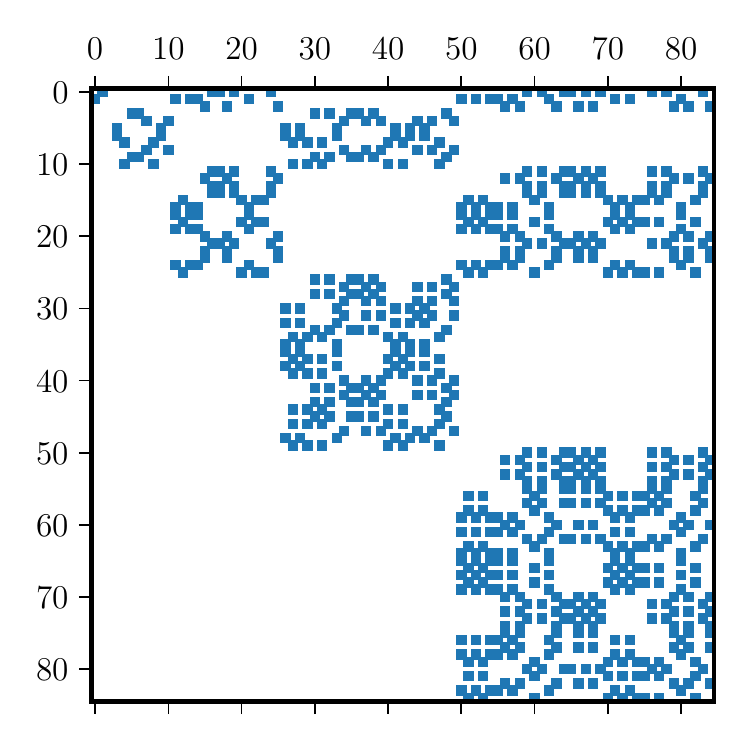} &
    \includegraphics[width=0.4\textwidth]
    {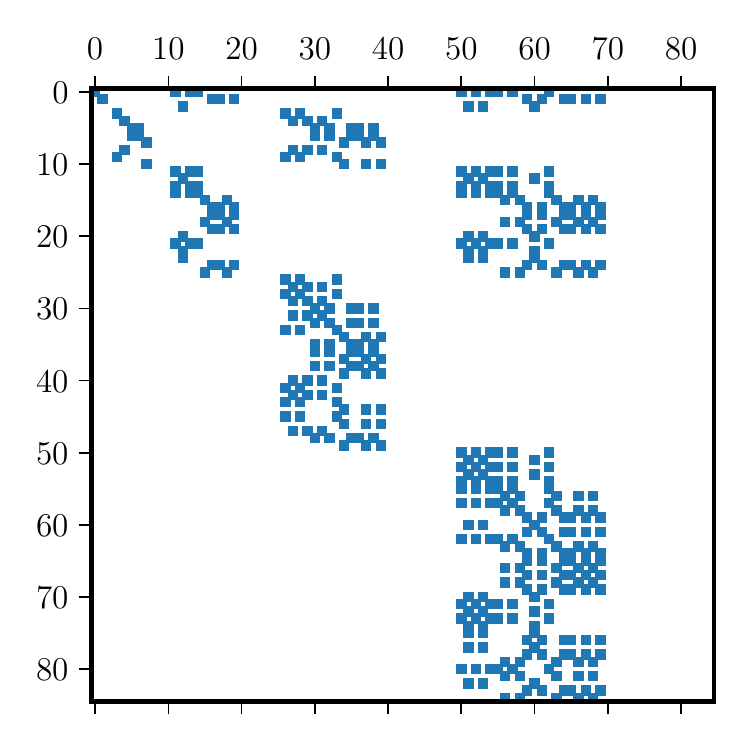} \\
\end{tabular}
\caption{Sparse upper triangular block matrices $\boldsymbol{A}_n$ (\emph{left}) and $\boldsymbol{B}_n$ (\emph{right}) for $n=5$ and $\boldsymbol{1}_\Omega = \boldsymbol{1}_z$.}
\label{fig:sparsity}
\end{figure}

For numerical convenience, we convert the finite-dimensional QEP into an equivalent finite-dimensional SEP when $\omega \neq 0$. 
The finite-dimensional counterpart of the self-adjoint operator defined in equation (\ref{eq:selfadjointT}b) is not well suited for numerical computations (because of the square root). 
Instead, we consider the SEP $\boldsymbol{L}_n \boldsymbol{z} = \omega \boldsymbol{z}$ with $\boldsymbol{z} = (\boldsymbol{\alpha}, \boldsymbol{\beta})^\top$ and the matrix $\boldsymbol{L}_n$ given by
\begin{subequations}
\label{eq:GEPV0n}
\begin{equation}
    \boldsymbol{L}_n := 
    \begin{pmatrix}
    0 & \boldsymbol{I} \\
    \boldsymbol{B}_n & \mathrm{i} \boldsymbol{A}_n \\
    \end{pmatrix}, \quad \boldsymbol{A}_n := \boldsymbol{M}_n^{-1} \boldsymbol{C}_n, \quad \boldsymbol{B}_n := \boldsymbol{M}_n^{-1} \boldsymbol{K}_n,
    \tag{\theequation a--c}
\end{equation}
\end{subequations}
where $\boldsymbol{I}$ is the identity matrix.
As illustrated in figure \ref{fig:sparsity}, the matrices $\boldsymbol{A}_n$ and $\boldsymbol{B}_n$ are sparse and have an upper triangular block structure (which directly results from Proposition \ref{theorem:invariancePcPk}). 
It has a practical consequence in the unstratified case $N=0$.
The eigenvalues of the pure inertial modes, given by $\boldsymbol{A}_n \boldsymbol{\alpha} = \omega \boldsymbol{\alpha}$, can indeed be computed separately for every space $\boldsymbol{\mathcal{W}}_n^0$ by seeking the eigenvalues of every diagonal block of $\boldsymbol{A}_n$.
Yet, since the matrix $\boldsymbol{L}_n$ has not an upper triangular block structure, we cannot compute separately the eigenvalues for every degree $n$ when $N \neq 0$. 

We have modified the bespoke numerical code \textsc{shine}, initially developed for compressible fluids \cite{vidal2020compressible,vidal2020acoustic}, to implement the above algorithm for Boussinesq fluids. 
The matrices $[\boldsymbol{M}_n,\boldsymbol{C}_n,\boldsymbol{K}_n]$ are first built analytically.
Since the basis vector elements are not mutually orthogonal in an ellipsoid \cite{lebovitz1989stability}, the matrices become numerically ill-conditioned when the polynomial degree is large. 
Hence, we perform the computations using extended-precision algorithms. 
We can illustrate the Galerkin algorithm by considering the six polynomial solutions $\boldsymbol{u} \in \boldsymbol{\mathcal{V}}_1^0$.
The latter are sought in the form of uniform-vorticity flows as \cite{noir2013precession}
\begin{subequations}
\label{eq:uniformvorticity}
\begin{equation}
    \boldsymbol{u} = \boldsymbol{\alpha} \times \boldsymbol{r} + \nabla \psi, \quad \psi = \alpha_x \frac{c^2-b^2}{b^2+c^2}yz + \alpha_y \frac{a^2-c^2}{a^2+c^2} xz + \alpha_z \frac{b^2-a^2}{a^2+b^2} xy,
    \tag{\theequation a--b}
\end{equation}
\end{subequations}
and with $\boldsymbol{\alpha} = (\alpha_x, \alpha_y, \alpha_z)^\top$.
Next, the block matrices of $\boldsymbol{L}_1$ are given by
\begin{subequations}
\allowdisplaybreaks
\begin{equation}
    \boldsymbol{B}_{1} = \begin{pmatrix}
        \dfrac{N^2 c^2}{b^2+c^2} & 0 & 0 \\
        0 & \dfrac{N^2 c^2}{a^2+c^2} & 0 \\
        0 & 0 & 0 \\
    \end{pmatrix}, \quad 
    \boldsymbol{A}_{1} = \begin{pmatrix}
        0 & - \dfrac{2 \Omega_z a^2}{a^2+c^2} & \dfrac{2 \Omega_y a^2}{a^2+b^2} \\
        \dfrac{2 \Omega_z b^2}{b^2+c^2} & 0 & - \dfrac{2 \Omega_x b^2}{a^2+b^2} \\
        - \frac{2 \Omega_y c^2}{b^2+c^2} & \dfrac{2 \Omega_x c^2}{a^2+c^2} & 0 \\ 
         \end{pmatrix},
\tag{\theequation a,b}
\end{equation}
\end{subequations}
and with $\boldsymbol{\Omega} = (\Omega_x, \Omega_y, \Omega_z)^\top$. 
Finally, the eigenvalue-eigenvector solutions of $\boldsymbol{L}_1$ can be solved analytically to extend the linear solutions obtained by Vantieghem \cite{vantieghem2014inertial} when $N=0$ and $\boldsymbol{1}_ \Omega = \boldsymbol{1}_z$.
However, we must rely on numerical computations to obtain $\boldsymbol{L}_n$ for higher degrees $n$. 

%-----------------------------------------------------------------------
\subsection{Enumeration of the eigenvalues}
\label{sec:VP0}
%-----------------------------------------------------------------------
We can explicitly count the number of non-zero and zero eigenvalues of the QEP when $\boldsymbol{u} \in \boldsymbol{\mathcal{V}}_n^0$. 
Since the spectra of the primitive equations, the one of $\boldsymbol{\mathcal{Q}}_\omega$ and the one of $\boldsymbol{\mathcal{L}}$ are identical when $\omega \neq 0$, their finite-dimensional counterparts have the same number of non-zero eigenvalues. 
\begin{proposition}
    The number of non-zero eigenvalues of QEP (\ref{eq:QEPVn0}) is $n(n+1)(n+5)/3$.
    \label{theo:nnzVn0}
\end{proposition}
\begin{proof}
    The number of non-zero eigenvalues can be estimated from the primitive equations
    \begin{equation*}
        \mathrm{i} \omega \boldsymbol{u} + 2 \boldsymbol{\Omega} \times \boldsymbol{u} = - \nabla \pi - \zeta N \boldsymbol{1}_z, \quad \nabla \boldsymbol{\cdot} \boldsymbol{u} = 0, \quad \mathrm{i} \omega \zeta = (\boldsymbol{u} \boldsymbol{\cdot} \boldsymbol{1}_z) \, N.
    \end{equation*}
    Solutions of the latter equations belong to the space spanned by $\boldsymbol{u} \in \boldsymbol{\mathcal{V}}_n^0$ and $\zeta \in \mathcal{P}_{n}$, whose dimension is $N(\boldsymbol{\mathcal{V}}_n^0) + N(\mathcal{P}_{n}) = (n+1)(n^2+4n+2)/2$.
    Next, we count the number of zero eigenvalues. 
    From the density equation, we deduce that $\omega=0$ is associated with 2D flows such that $\boldsymbol{u} \boldsymbol{\cdot} \boldsymbol{1}_z = 0$.
    Such 2D flows $\boldsymbol{u} \in \boldsymbol{\mathcal{V}}_n^0$ can be sought as $\boldsymbol{u} = \nabla \times (\Psi \boldsymbol{1}_z)$ where $\Psi = p [1-F(x,y,z)]$ is a stream function and with $p \in \mathcal{P}_{n-1}$. 
    Assuming that $\boldsymbol{\Omega} = (\Omega_x, \Omega_y, \Omega_z)^\top$, we obtain $ 2 \Omega_z \partial_x \Psi = -\partial_x \pi$ and $2 \Omega_z \partial_y \Psi = - \partial_y \pi$ from the horizontal components of the momentum equation. 
    Then, we deduce that $\pi = -2\Omega_z \Psi + h(z)$, where $h(z)$ is an arbitrary polynomial function of $z$ and of degree $n+1$ (with $n\geq 1$).
    Finally, the density perturbation is obtained from the vertical component of the momentum equation, which gives $\zeta N = 2(\Omega_y \partial_y \Psi + \Omega_x \partial_x \Psi) - \partial_z \pi$. 
    Thus, $\omega = 0$ is associated with an eigenspace of dimension $\dim (\mathcal{P}_{n-1}) + (n+1) = (n+1)(n^2+2n+6)/6$. 
    Therefore, the number of non-zero eigenvalues is given by $n(n+1)(n+5)/3$.
\end{proof}

Equipped with Proposition \ref{theo:nnzVn0}, we can enumerate the zero eigenvalues of the matrix $\boldsymbol{L}_n$, which has $2 N(\boldsymbol{\mathcal{V}}_n^0)$ eigenvalues.
Hence, the algebraic multiplicity of the zero eigenvalue $\omega=0$ is given by $n(n+1)(n+2)/3$. 
However, $\omega=0$ is not only associated with standard eigenvectors. 
The geometric multiplicity $m_L$ of $\omega=0$ is the dimension of $\ker \boldsymbol{L}_n$. 
The null space of the equation $\boldsymbol{L}_n (\boldsymbol{\alpha}, \boldsymbol{\beta})^\top =0$ is spanned by the vector elements such that $\boldsymbol{\beta}=\boldsymbol{0}$ and $\boldsymbol{K}_n \boldsymbol{\alpha}=~0$.
Therefore, we have $\ker\boldsymbol{L}_n=\ker\boldsymbol{K}_n \bigoplus \{0\}$ and $\omega=0$ has the geometric multiplicity $m_L = n(n+1)(n+2)/6$ according to Proposition \ref{theo:dimkerCK}.
However, since $\boldsymbol{L}_n$ is not Hermitian, the eigenvalue $\omega=0$ has an algebraic multiplicity $r_L > m_L$. 
The missing $r_L-m_L = n(n+1)(n+2)/6$ polynomial solutions are associated with generalized eigenvectors, which belong to  $\ker(\boldsymbol{L}_n^2)$. 
% This is summarised in Proposition \ref{theo:defectiveVP0}.
\begin{proposition}
\label{theo:defectiveVP0}
The eigenvalue $\omega=0$ of the matrix $\boldsymbol{L}_n$ has the geometric multiplicity $m_L =n(n+~1)(n+2)/6$, and the algebraic multiplicity  $r_L = 2m_L$ (because of generalized eigenvectors). 
\label{prop:kernelK}
\end{proposition}

The vector space associated with $\omega=0$ is thus partly spanned by generalized eigenvectors, which here satisfy $\boldsymbol{L}_n \boldsymbol{z} \neq \boldsymbol{0}$ and $\boldsymbol{L}_n^2 \boldsymbol{z} = \boldsymbol{0}$. 
Back in the temporal domain, they correspond to solutions that are not steady but linear in $t$.
However, these admissible solutions are mathematical artefacts due to the formulation of the quadratic problem in terms of the matrix $\boldsymbol{L}_n$.
Indeed, QEP (\ref{eq:QEPVn0}) can be converted into an equivalent self-ajoint SEP \cite{barston1967eigenvalue}, which is the finite-dimensional counterpart of the self-adjoint operator defined in equation (\ref{eq:selfadjointT}b).
This shows that, from a physical viewpoint, $\omega=0$ is only associated with standard eigenvectors. 
Therefore, we will not further elaborate on the (artificial) generalized eigenvectors of $\boldsymbol{L}_n$ below. 
Following our prior investigation of the non-stratified problem \cite{CdV2023spectrum}, we can combine polynomial computations for fixed values of $n$ and spectral theory for $n \to \infty$ to study the eigensolutions in ellipsoids. 

%-----------------------------------------------------------------------
\section{Point spectrum in ellipsoids}
\label{sec:results}
%-----------------------------------------------------------------------
\begin{figure}
\centering
\includegraphics[width=0.8\textwidth]{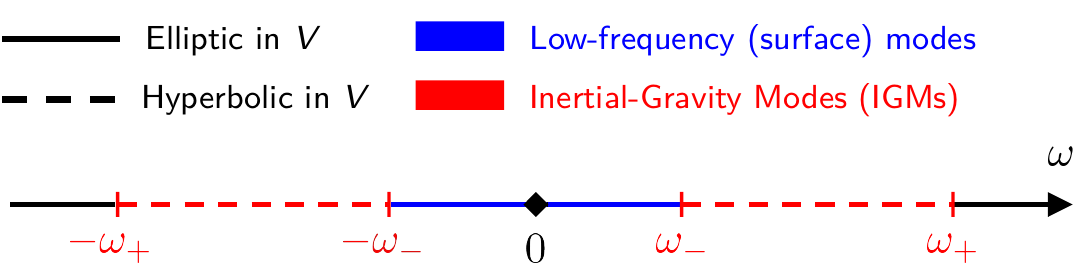}
\caption{Sketch of the bounded pure point spectrum in triaxial ellipsoids (not to scale) when $N$ is constant.  
The spectrum is symmetric with respect to $0$ and bounded by $|\omega| \leq \omega_{+}$ (Proposition \ref{theo:propsigma3}). 
The velocity equation is either hyperbolic or elliptic in volume. 
The spectrum is spanned by IGMs and low-frequency (surface) modes.} 
\label{fig:fig4}
\end{figure}

We have shown that $\sigma(\boldsymbol{\mathcal{Q}}_\omega)$ is pure point in ellipsoids when $N$ is constant.
For this configuration, we can further characterise the eigensolutions. 
In particular, the spectrum can be divided into two subsets.
If an eigenvalue is isolated in the spectrum and of finite multiplicity (when $n \to \infty$), it belongs to the discrete spectrum $\sigma_D$.
Otherwise, it belongs to the essential spectrum $\sigma_E := \sigma \backslash \sigma_D$ (i.e. the complement to the discrete spectrum).
Since $\oplus_{n\in \mathbb{N}^*} \boldsymbol{\mathcal{V}}_n^0$ is dense in $\boldsymbol{\mathcal{V}}^0$, we conclude from Proposition \ref{theo:defectiveVP0} that the eigenvalue $\omega=0$ has an infinite multiplicity when $n \to~\infty$ (i.e.  
$\omega=0$ belongs to $\sigma_E$).
Now, it remains to characterise the other eigensolutions.

When $N$ is constant, the frequencies $\omega_{\pm}$ defined in equation (\ref{eq:omegapmleblonc}) are constant in the volume. 
Moreover, it is possible to find a better estimate of $\omega_{\max}$ (compared to the one given in Proposition \ref{theo:propQEP}). 
In the aligned case, it was shown in \cite{friedlander1982jfm} that $\omega_{\max} = \omega_+$ with $\omega_+ =\max (N, 2 \Omega_s)$ .
Thus, $\sigma_3$ is empty in ellipsoids when $\boldsymbol{\Omega}$ and $ \boldsymbol{1}_z$ are collinear. 
We also show below in Proposition \ref{theo:propsigma3} that $\omega_{\max}=\omega_+$, even when $\boldsymbol{\Omega}$ and $\boldsymbol{1}_z$ are not collinear. 

\begin{proposition} 
\label{theo:propsigma3}
When $N$ is constant, the pure point spectrum is bounded by $\omega_+$ in an ellipsoid. 
\end{proposition}
\begin{proof}
We introduce the projector operator $\mathbb{P}$ acting on a pair $\boldsymbol{z}=~(\boldsymbol{u},\zeta)^\top$ by $\mathbb{P} (\boldsymbol{z}) :=(\mathbb{L} \boldsymbol{u}, \zeta )^\top$, where $\mathbb{L}$ is the Leray projector introduced in \S\ref{sec:problem}.
Then, we can reformulate the spectral problem as the spectral properties of the linear operator $\boldsymbol{\mathcal{P}}$ defined by $\boldsymbol{\mathcal{P}} (\boldsymbol{z}) := \mathbb{P}( \boldsymbol{\mathcal{E}} \mathbb{P} (\boldsymbol{z}))$ with
\begin{equation*}
\boldsymbol{\mathcal{E}} = \mathrm{i} \begin{pmatrix}
0 & -2\Omega _3 & 2\Omega _2 & 0 \\
2\Omega_3& 0 &- 2\Omega _1 &0 \\
-2\Omega _2 & 2\Omega _1 &0 & N \\
0 & 0 &-N & 0 \end{pmatrix}.
\end{equation*}
$\boldsymbol{\mathcal{P}}$ is a self-adjoint operator for square-integrable fields.
We have $||\mathbb{P}||=1$, and the eigenvalues of the Hermitian matrix $\boldsymbol{\mathcal{E}}$ are $\pm \omega_\pm$ such that $||\boldsymbol{\mathcal{E}}|| \leq \omega _+ $.
Thus, the spectrum is bounded by $\omega_+$.
\end{proof}

Therefore, we can divide the allowable frequency range into disjoint intervals with 
$I_1 = [-\omega_+, -\omega_-] \cup [\omega_-, \omega_+]$ and $I_2 = ]-\omega_-, \omega_-[ \, \setminus \, \{0\}$. 
The velocity equation is hyperbolic in $V$ when $\omega \in I_1$, and elliptic in $V$ when $\omega \in I_2$.
We also define the associated spectra $\sigma_i = \sigma_P \cap I_i$ for $i = \{1,2\}$ and, when $\boldsymbol{u}$ is restricted to $\boldsymbol{\mathcal{V}}_n^0$, the subsets are denoted by $[\sigma_{1,n}, \sigma_{2,n}]$ such that $\sigma_{i,n-1} \subseteq \sigma_{i,n}$. 
In unbounded fluids, a continuum of IGWs exists with angular frequencies $\omega \in I_1$ given by dispersion relation  (\ref{eq:dispersionIGW}). 
We will show below that the spectrum $\sigma_1 \subset\sigma_E$ is dense in $I_1$ within ellipsoids.
On the contrary, $\sigma_2$ is known to be empty in unbounded fluids (which agrees with prior experimental observations \cite{peacock2005effect}).
This is consistent with the fact that the velocity satisfies an elliptic-type equation when $\omega \in I_2$ (i.e. plane-wave solutions are evanescent in these intervals).
Yet, a countable number of eigenfrequencies can exist in some bounded geometries when $|\omega| < \omega_-$ \cite{allen1971some,friedlander1982jfm}.
The eigensolutions with $\omega \in \sigma_1$ are referred as class-I solutions, and those with $\omega \in \sigma_2$ are called class-II solutions in \cite{allen1971some,friedlander1982jfm}.
We describe the two spectra below.

%-----------------------------------------------------------------------
\subsection{Low-frequency (surface) modes}
\label{subsec:discrete}
\subsubsection{Illustrative examples}
%-----------------------------------------------------------------------
Contrary to the free-space case, it is known that $\sigma_2$ is not empty when $\boldsymbol{\Omega} \propto \boldsymbol{1}_z$ in some bounded geometries (e.g. in cylinders or spheres \cite{friedlander1982jfm,kerswell1993elliptical}).
In triaxial ellipsoids, we can also show that $\sigma_2 \neq \emptyset$ when $\boldsymbol{\Omega} \boldsymbol{\cdot} \boldsymbol{1}_z \neq 0$ (otherwise, $\omega_- = 0$ and $\sigma_1$ fills the entire interval $0 < |\omega|<\omega_+$).
To this end, we first consider the eigenvalues for $\boldsymbol{u} \in \boldsymbol{\mathcal{V}}_1^0$, which  are the roots of the characteristic equation $\det (\boldsymbol{L}_1 - \omega \boldsymbol{I}) = 0$.
The six roots $\omega_{j}$ of the characteristic polynomial are the two degenerate solutions $\omega_{1,2} = 0$ and the four solutions
\begin{subequations}
    \allowdisplaybreaks
    \label{eq:solw36V10}
    \begin{equation}
        \omega_{3,4} = \pm \sqrt{\frac{\beta_4}{2 \Delta} + \frac{\sqrt{(\beta_4/\Delta)^2 - 4 \beta_2/\Delta}}{2}}, \quad \omega_{5,6} = \pm \sqrt{\frac{\beta_4}{2 \Delta} - \frac{\sqrt{(\beta_4/\Delta)^2 - 4 \beta_2/\Delta}}{2}},
        \tag{\theequation a--b}
    \end{equation}
\end{subequations}
where we have introduced the positive coefficients
\begin{align*}
    \allowdisplaybreaks
    \small
    \beta_4 &= N^2 c^2 (a^2 + b^2) (a^2 + b^2 + 2 c^2) + 4 \Omega_x^2 b^2 c^2 (b^2 + c^2) + 4 \Omega_y^2 a^2 c^2 (a^2 + c^2) + 4 \Omega_z^2 a^2 b^2 (a^2 + b^2), \\
    \beta_2 &= c^4 N^2 [ 4 (\Omega_x^2 b^2 + \Omega_y^2 a^2) + N^2 (a^2 + b^2)], \\
    \Delta &= a^4 b^2 + a^4 c^2 + a^2 b^4 + 2a^2 b^2 c^2 + a^2 c^4 + b^4 c^2 + b^2 c^4. 
\end{align*}
The only non-zero solutions in the non-stratified case $N=0$ are $\omega_{3,4} = \pm \sqrt{\beta_4/\Delta}$, which reduce to the eigenvalues of the pure inertial modes in $\boldsymbol{\mathcal{V}}_1^0$ obtained by Vantieghem \cite{vantieghem2014inertial} when $\Omega_x=\Omega_y = 0$. 
Solutions (\ref{eq:solw36V10}) are illustrated in figure \ref{fig:discrete1}(a) for a sphere and a flattened (triaxial) ellipsoid. 
Several points are worth commenting on in the figure.
First, $\omega_3$ and $\omega_5$ smoothly vary as a function of $N$.
We observe that $\omega_3 \in \sigma_{1,1}$ for any value of $N$ but, on the contrary, $\omega_5$ can either belong to $\sigma_{1,1}$ or $\sigma_{2,1}$ when $N$ is varied.
For instance, the transition occurs here at $N/\Omega_s>4$ in a sphere and at a much larger value $N/\Omega_s \gtrsim 25 $ in the pancake-like ellipsoid. 

\begin{figure}
\centering
\begin{tabular}{cc}
    \begin{overpic}[width=0.49\textwidth]
{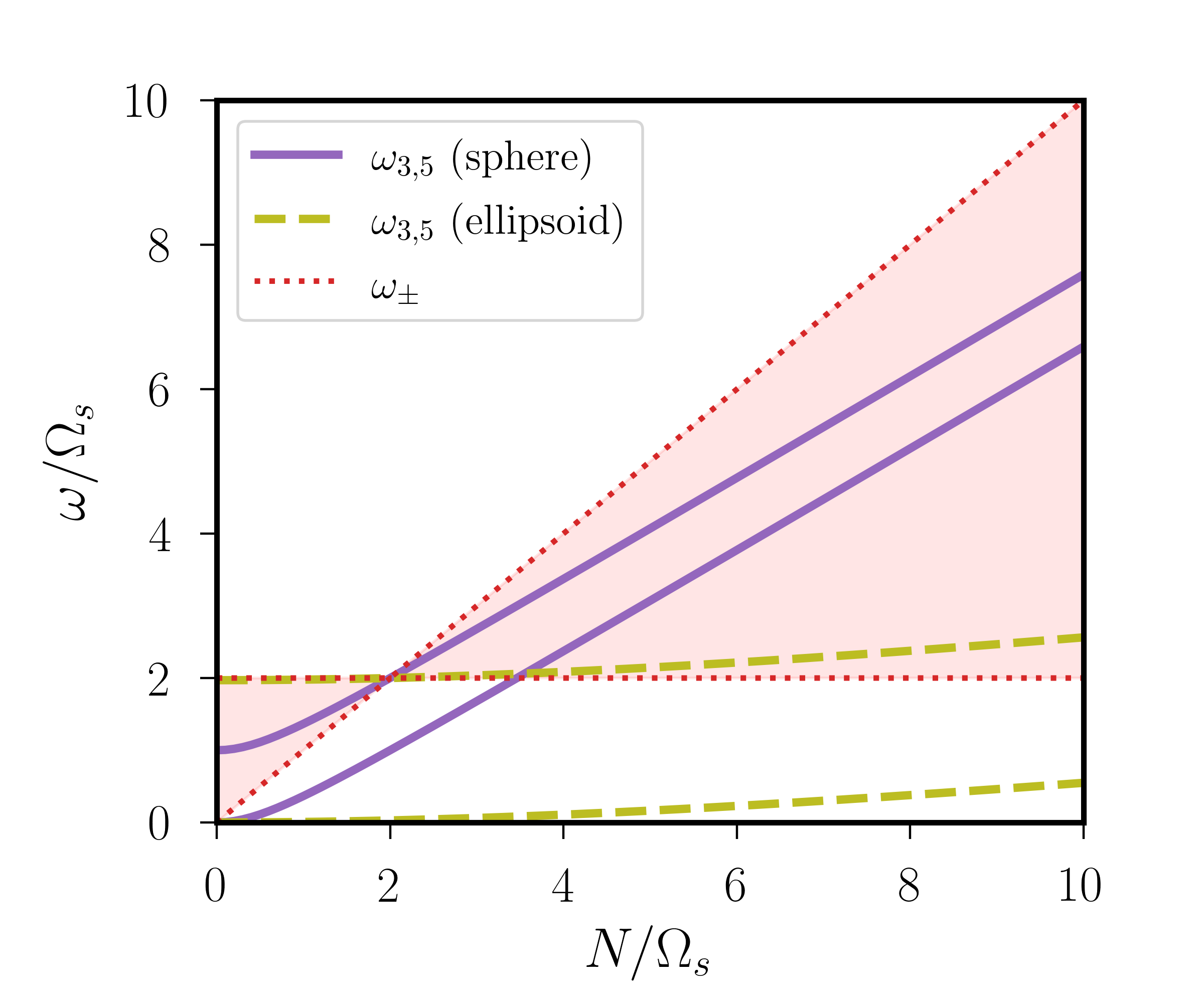}
    \put(-5,70){(a)}    
    \end{overpic} & 
    \begin{overpic}[width=0.49\textwidth]
{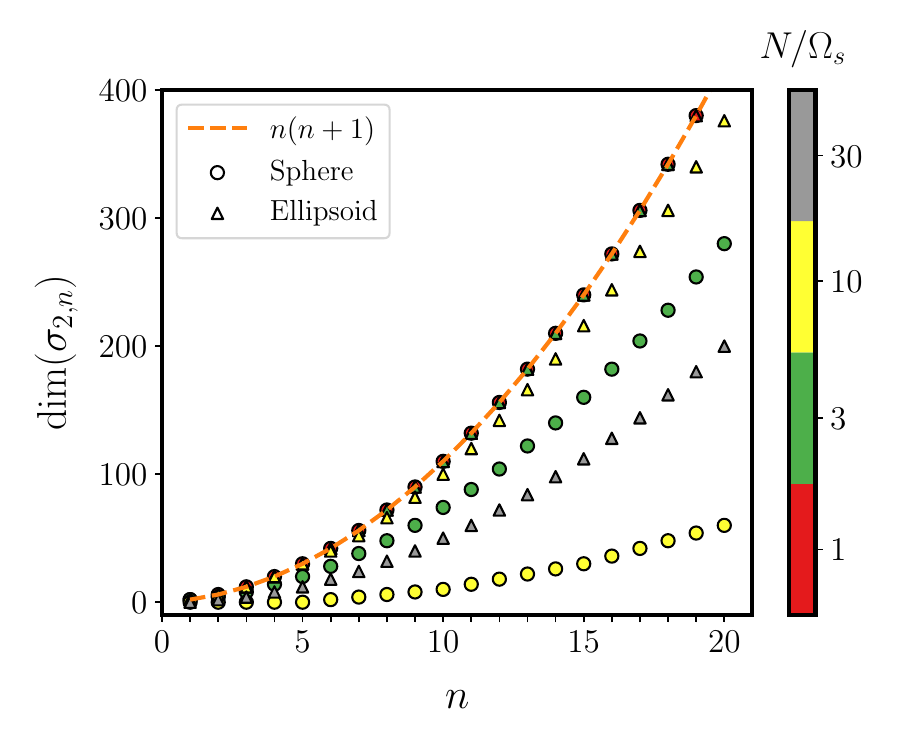}
    \put(-5,70){(b)}    
    \end{overpic}
\end{tabular}
\caption{Spectrum $\sigma_2$ in the aligned case $\boldsymbol{1}_\Omega = \boldsymbol{1}_z$. (a) Eigenvalues $\omega_{3,5} \geq 0 $ in $\sigma_{1,1}$ for a sphere with $b/a=c/a=1$ (solid purple curves) and an ellipsoid with $b/a=0.7$ and $c/a=0.1$ (dashed olive curves). Red region shows $\sigma_2$ in which the problem is hyperbolic. (b) Number of eigenvalues in $\sigma_{1,n}$ for a sphere ($b/a=c/a=1$) and a flattened ellipsoid ($b/a=0.7$ and $c/a=0.1$). Colour bar shows the value of $N/\Omega_s$.}
\label{fig:discrete1}
\end{figure}

Next, we show in figure \ref{fig:discrete1}(b) the evolution of $\dim (\sigma_{2,n})$ as a function of $n$ for some illustrative cases.  
We see that we have always $\sigma_{2,n} \neq \emptyset$ for a sufficiently large degree $n$, and that $\dim(\sigma_{2,n})$ increases when $n$ is increased.
This suggests that $\sigma_2$ has an infinite (but countable) number of eigenvalues when $n \to \infty$. 
However, $\dim(\sigma_{2,n})$ is found to strongly depend on the geometry and the value of $N/\Omega_s$. 
More specifically, $\dim(\sigma_{2,n})$ is reduced when $N$ is increased, but the variation is more pronounced in the sphere than in pancake-like ellipsoids with $c/a \ll 1$.
Finally, the number of eigenvalues seems to be bounded by $n(n+1)$, which corresponds to the number of spherical harmonics $Y_l^m$ of degree $2 \leq l \leq n+1$ and order $0<|m|<l$. 
This shows that $\sigma_2$ is related to solid (spherical) harmonics.

Further insight into the link between $\sigma_2$ and the theory of spherical harmonics can be gained by considering the aligned case $\boldsymbol{1}_\Omega = \boldsymbol{1}_z$ with $N=2\Omega_s$.
For this particular configuration, we have $\omega_-=\omega_+$ from equation (\ref{eq:omegapmleblonc}).
Then, pressure equation (\ref{eq:laplaP}) and its associated BC for eigenvalues $\omega \in \sigma_2$ reduce to
\begin{subequations}
\label{eq:laplaN=2W}
\begin{equation}
    \nabla^2 \Phi=0, \quad \boldsymbol{n} \boldsymbol{\cdot} \nabla \Phi = - \frac{1}{\mathrm{i} \omega} (2 \boldsymbol{\Omega} \times \boldsymbol{n}) \boldsymbol{\cdot} \nabla \Phi \ \, \text{on $\partial V$}
    \tag{\theequation a,b}.
\end{equation}
\end{subequations}
Laplace equation (\ref{eq:laplaN=2W}a) can then be solved using harmonic functions in spherical or ellipsoidal geometries, and the eigenvalues $\omega$ are the roots of a transcendental equation obtained from BC (\ref{eq:laplaN=2W}b). 
In a sphere, the pressure associated with $\boldsymbol{u} \in \boldsymbol{\mathcal{V}}_n^0$ is given in spherical coordinates $(r,\theta,\phi)$ by the solid harmonics $\Phi(\boldsymbol{r}) = r^l Y_l^m (\theta, \phi)$ with $\omega = \pm 2 \Omega_s m/l$, where $Y_l^m$ is the spherical harmonic of degree $l\leq n+1$ and azimuthal order $m<l$. 
The modes with $m>0$ have a negative frequency $\omega < 0$ when $\Omega_s > 0$ (respectively positive $\omega > 0$ if $m<0$).
Given our phase convention $\propto \exp[\mathrm{i} (\omega t + m \phi)]$, these modes have a positive azimuthal phase velocity $c_\phi = - \omega/m >0$ when $\Omega_s > 0$ (i.e. prograde direction). 
The explicit solutions show that the spectrum $\sigma_{2,n}$ is here isomorph to the set of rational numbers $\mathbb{Q}$, such that $\dim (\sigma_{2,n}) =~n(n+1)$ in this case.
Moreover, it can explicitly be shown that there is a one-to-one correspondence between the solutions in the sphere and those in the ellipsoid for this specific configuration \cite{friedlander1982jfm}.
Since $\mathbb{Q}$ is dense in $\mathbb{R}$, $\sigma_2$ belongs to the essential spectrum $\sigma_E$ when $\boldsymbol{1}_\Omega = \boldsymbol{1}_z$ and $N=2\Omega_s$.

%-----------------------------------------------------------------------
\subsubsection{General properties}
%-----------------------------------------------------------------------
The above examples strongly suggest that $\sigma_2 \subset \sigma_E$. 
This can be shown by further looking at the pressure problem. 
Assuming that $|\omega| \neq \omega_\pm$, the pressure is governed by a boundary-value problem given by equation (\ref{eq:laplaP}) in $V$ and a mixed-type BC on $\partial V$. 
This boundary-value scalar problem is elliptic if the equation is elliptic in $V$ (which occurs here when $0<|\omega| < \omega_-$), but also if some invertibility conditions are satisfied by the BC on $\partial V$ \cite{grubb2008distributions} (see the ESM). 
If both requirements were met in $V$ and everywhere on $\partial V$ in a given interval $[\omega_1,\omega_2]$, then the spectrum would be discrete in this interval (see the ESM). 
Otherwise, the spectrum would be essential in such an interval. 
Here, we can prove that $\sigma_2$ belongs to the essential spectrum in an ellipsoid when $\boldsymbol{\Omega}$ and $\boldsymbol{1}_z$ are aligned (see the ESM).
We also conjecture that $\sigma_2 \in \sigma_E$ in the misaligned case. 
Because the spectrum is pure point in ellipsoids, an infinite number of modes would thus exist with $0 < |\omega| < \omega_-$ when $n \to \infty$. 

% \begin{proposition}
%     When $0 < |\omega| < \omega_-$, the boundary-value pressure problem is not elliptic in an ellipsoid when $N$ is constant. 
%     % for any orientation of $\boldsymbol{\Omega}$ with respect to $\boldsymbol{1}_z$
%     Hence, $\sigma_2 \in \sigma_E$ is dense in $I_2$.
%     \label{theo:essentialealigned}
% \end{proposition}
% \begin{proof}
% The proof, which requires advanced tools of microlocal analysis, is omitted here for concision. 
% It is detailed in the ESM for the interested readers.
% % Note that we also expect $\sigma_2$ to be essential and dense in $I_2$ for misaligned global rotation and gravity (see the ESM).
% \end{proof}

\begin{figure}
\centering
\includegraphics[width=0.99\textwidth]{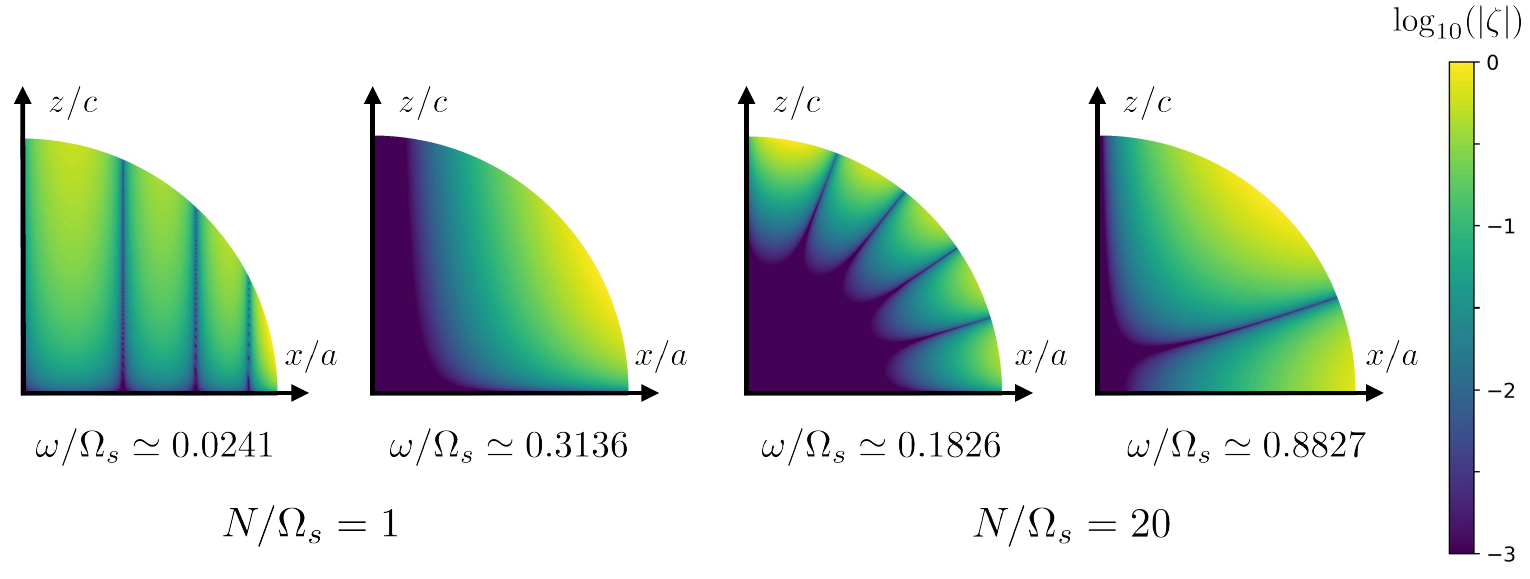} 
\caption{Spatial structure of the low-frequency (surface) modes for $\boldsymbol{u} \in \boldsymbol{\mathcal{V}}_{10}^0$ in the aligned case $\boldsymbol{1}_\Omega = \boldsymbol{1}_z$ as a function of $N/\Omega_s$, in a flattened ellipsoid with $b/a=1$ and $c/a=0.1$. Meridional view of the absolute value of the normalised density perturbation $|\zeta|$ (in logarithmic scale) in the $(Oxz)$ plane (with rescaled axes to have a 1:1 ratio).}
\label{fig:kelvin3d}
\end{figure}

Finally, we explore the properties of the corresponding eigenvectors. 
We remind the reader that these modes only exist when $\omega_- \neq 0$. 
Otherwise, pure inertial modes (when $N=0$), pure internal (gravity) modes (when $\boldsymbol{\Omega}=\boldsymbol{0}$) or pure IGMs (when $\boldsymbol{\Omega} \boldsymbol{\cdot} \boldsymbol{1}_z = 0$) fill the entire interval $0< |\omega| < \omega_+$. 
Two modal families are shown in figure \ref{fig:kelvin3d}.
When $N/\Omega_s \leq 2$, the energy of the lowest-frequency modes with $|\omega| \to 0$ is nearly aligned with the rotation axis.
This structure is reminiscent of geostrophic modes without stratification, and fully agrees with asymptotic analysis when $N/\Omega_s \ll 1$.
As investigated by Allen \cite{allen1971some} when $\boldsymbol{\Omega} \propto \boldsymbol{1}_z$, these modes appear because a weak stratification imposes some spatial restrictions on the steady geostrophic modes without stratification. 
The modes are nearly $z-$independent throughout the fluid (at the leading order), and have a phase travelling in the prograde direction.
Note that this origin seems similar to the appearance of low-frequency Rossby modes in a volume $V$ without closed geostrophic contours \cite{greenspan1968theory}. 
On the contrary, the higher-frequency modes with $|\omega| \lesssim \omega_-$ have an energy that is maximum near the boundary and decays more or less significantly away from the boundary.
These low-frequency (surface) modes seem similar to the (trapped) Kelvin waves in shallow-water models \cite{thomson18801}. 
Indeed, as for the Kelvin waves, the wave energy of these low-frequency (surface) modes is mainly trapped near the boundary and decays away from $\partial V$ (see figure \ref{fig:kelvin3d}). 
Therefore, the presence of a boundary introduces a new subset of low-frequency (surface) modes in a rotating stratified fluid (which differ from the higher-frequency IGMs). 

%-----------------------------------------------------------------------
\subsection{High-frequency IGMs}
\label{subsec:essential}
%-----------------------------------------------------------------------
We investigate IGMs with angular frequencies belonging to $\sigma_1$.
A few illustrative modes are shown in figure \ref{fig:IGM3d}.
IGMs are mainly inertial modes modified by gravity when $N \ll 2 \Omega_s$, whereas they are gravity modes more or less affected by rotation when $N \gg 2 \Omega_s$. 
In between these two limits, rotation and stratification can have competing effects.
Moreover, some of their characteristics are directly controlled by the dispersion relation of IGWs.
For example, the modes exhibit a pancake-like structure when $\omega/\Omega_s \to~2$ and $\boldsymbol{1}_\Omega=\boldsymbol{1}_z$.
From dispersion relation (\ref{eq:dispersionIGW}), the angle $\vartheta$ between $\boldsymbol{1}_z$ and $\boldsymbol{k}$ is such that $\vartheta \to 0$ when $\omega/\Omega_s \to 2$.
Since the group velocity is orthogonal to $\boldsymbol{k}$ for IGWs \cite{leblond1981waves}, the energy propagates in the horizontal direction (in agreement with the observed pancake-like structure).
On the contrary, we have $\theta \to \pi/2$ for high-frequency IGWs with $\omega \to N$, such that the energy varies in the vertical direction.
IGMs with $\omega \in \sigma_1$ are thus the direct counterparts in ellipsoids of IGWs in unbounded fluids. 

\begin{figure}
\centering
\includegraphics[width=0.99\textwidth]{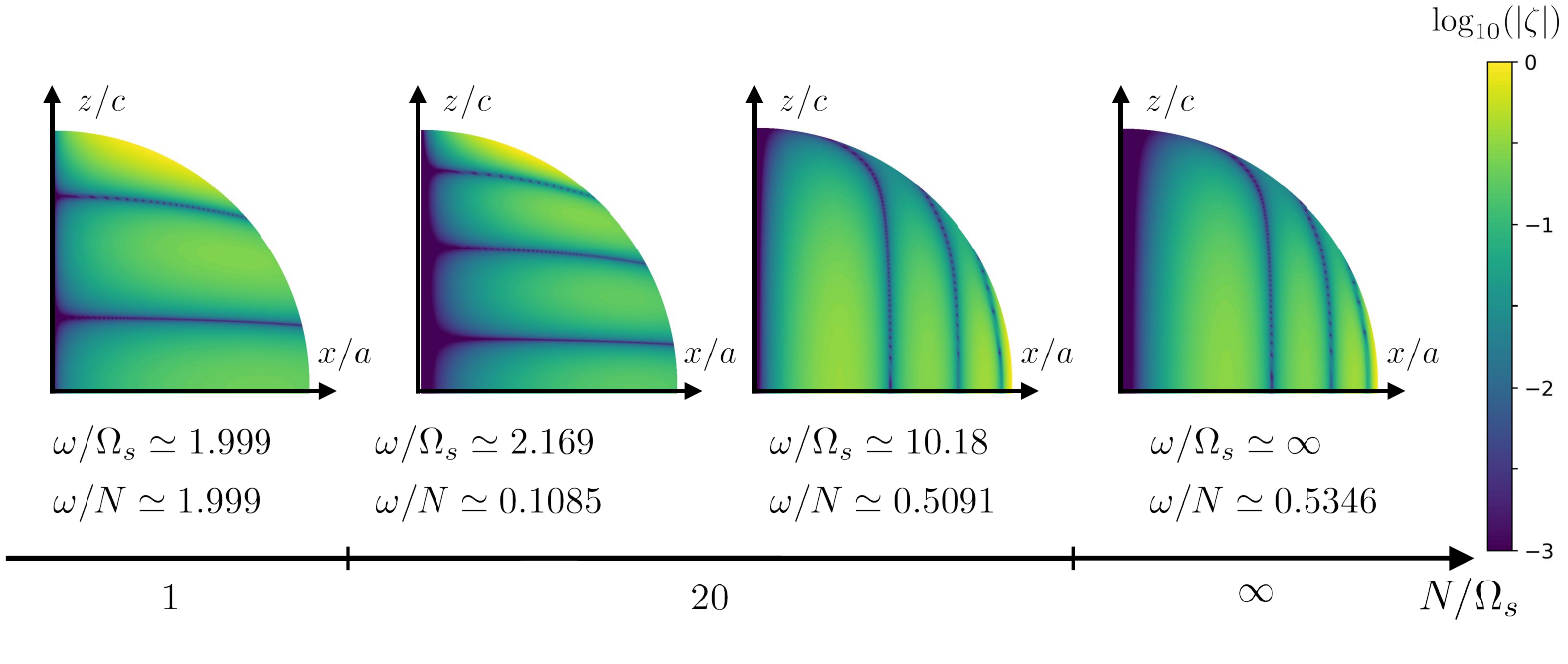} 
\caption{Spatial structure of IGMs for $\boldsymbol{u} \in \boldsymbol{\mathcal{V}}_{10}^0$ in the aligned case $\boldsymbol{1}_\Omega = \boldsymbol{1}_z$ as a function of $N/\Omega_s$, in a flattened ellipsoid with $b/a=1$ and $c/a=0.1$. Meridional view of the absolute value of the normalised density perturbation $|\zeta|$ (in logarithmic scale) in the $(Oxz)$ plane (with rescaled axes to have a 1:1 ratio).}
\label{fig:IGM3d}
\end{figure}

Such as IGWs without boundaries, IGMs are dense in this frequency interval when $n \to \infty$.
Yet, the eigenfrequencies are likely not equiprobable within this interval in ellipsoids. 
The frequency density is worth obtaining for physical applications, since it will strongly constrain the likelihood of obtaining (resonant) IGMs at a given frequency in geophysical vortices. 
The asymptotic density of the eigenvalues in $\sigma_1$ can be obtained using mathematical analysis when $n \to \infty$, and then compared to the density obtained from the polynomial computations at fixed values of $n$. 
We define the probability measure $\pi_n$ of the eigenvalues $\{\omega_j\}$ of $\boldsymbol{\mathcal{Q}}_{\omega} |_{\boldsymbol{\mathcal{W}}_n^0}$ given by 
$\pi_n:= (1/d_n) \sum_{j=1}^{d_n} \delta (\omega_j)$ with $d_n = 2 N(\boldsymbol{\mathcal{W}}_n^0)$, where $\delta$ is the Dirac function. 
As explained above, the eigenvalues of SEP (\ref{eq:GEPV0n}) cannot be easily separated in different subsets for every degree $n$.
Therefore, we rather consider the joint repartition of the eigenvalues of $\boldsymbol{\mathcal{Q}}_\omega |_{\boldsymbol{\mathcal{V}}_n^0}$. 
Naturally, the two repartitions have the same asymptotic distribution when $n \to \infty$.
Asymptotic properties of the essential spectrum of PDEs can generally be investigated using microlocal analysis.
Such an approach usually neglects the boundary effects but, here, it is possible to account for them while keeping a microlocal description.   
The asymptotic behaviour is given in Proposition \ref{theo:dens}. 

\begin{figure}
\centering
\begin{tabular}{cc}
    \includegraphics[width=0.48\textwidth]{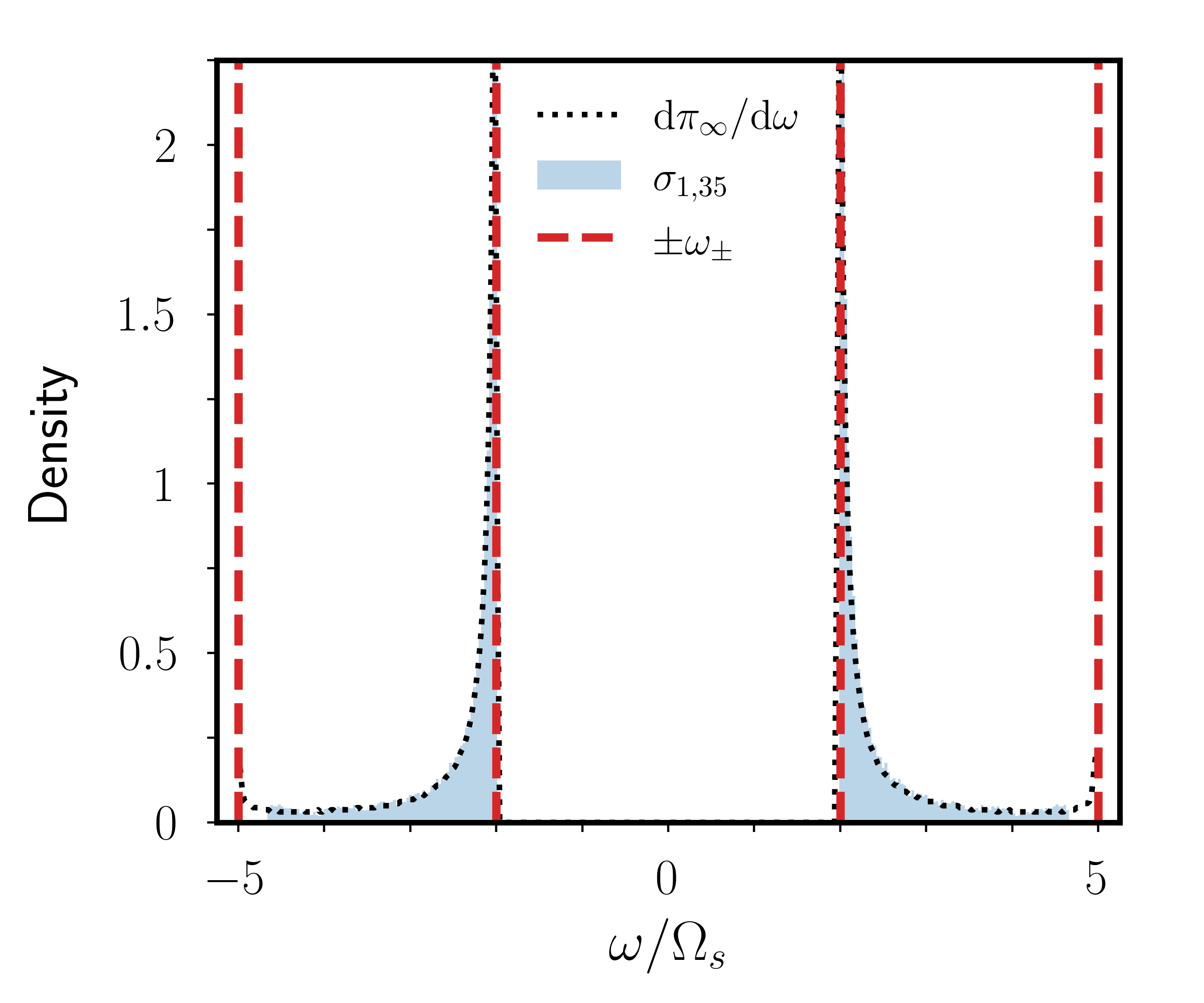} &
    \includegraphics[width=0.48\textwidth]{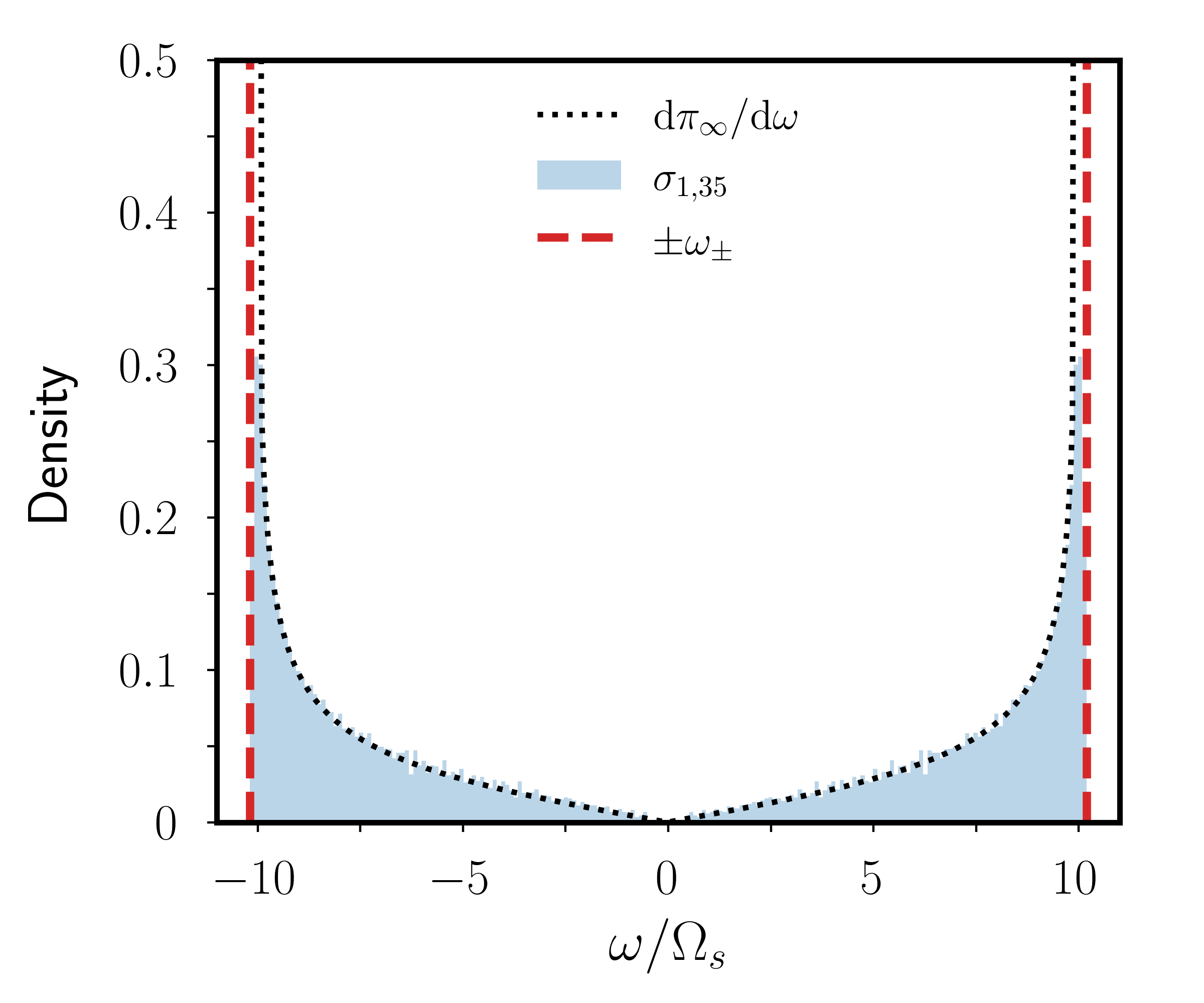} \\
\end{tabular}
\caption{Density of IGMs in $\sigma_1$. Comparison between  asymptotic density $\mathrm{d} \pi_\infty/\mathrm{d}\omega$ given in Proposition \ref{theo:dens} and $\sigma_{2,35}$ obtained from polynomial computations. \emph{Left}: Flattened ellipsoid $b/a=1$ and $ c/a=0.1$ in the aligned case $\boldsymbol{1}_\Omega = \boldsymbol{1}_z$ with $N/\Omega_s = 5$. \emph{Right}: Sphere $b/a=c/a=1$ in the perpendicular case $\boldsymbol{1}_\Omega \boldsymbol{\cdot} \boldsymbol{1}_z = 0$ with $N/\Omega_s = 10$.}
\label{fig:bench}
\end{figure}

\begin{proposition} 
\label{theo:dens} 
The measure $\pi_n$ converges weakly to the asymptotic measure $\pi_\infty$ when $n \to \infty$. 
The asymptotic measure is characterised by a cumulative distribution function given by
\begin{equation*}
\int_{-\infty}^\lambda \mathrm{d} \pi_\infty = \frac{1}{8\pi }\mathrm{Area}(\mathfrak{S}_\lambda \cap S^2 )
\end{equation*}
where $S^2$ is the surface of the unit sphere, $\mathfrak{S}_\lambda$ is the surface in $\mathbb{R}^3$ defined by 
\begin{equation*}
\mathfrak{S}_\lambda:=\{ \widetilde{\boldsymbol{k}} \in \mathbb{R}^3 \ | \ \omega(\widetilde{\boldsymbol{k}}) \leq \lambda \}, \quad \widetilde{\boldsymbol{k}} = ( \sqrt{A_1}k_x,\sqrt{A_2}k_y, \sqrt{A_3}k_z )^\top,
\end{equation*}
with $A_1=1/a^2$, $A_2=1/b^2$ and $A_3=1/c^2$, and where $\omega(\widetilde{\boldsymbol{k}})$ is the plane-wave dispersion relation given by equation (\ref{eq:dispersionIGW}) for the rescaled wave vector $\widetilde{\boldsymbol{k}}$. 
Hence, the asymptotic measure has a probability density function $f_\infty:=~\mathrm{d}\pi_\infty/ \mathrm{d}\omega$ that is even and non-vanishing when $\omega_-^2 \leq \omega^2 \leq \omega_+^2$.
\end{proposition}
\begin{proof}
This is the most technical part of the paper. 
The proof closely follows the strategy presented in Colin de Verdi\`ere \& Vidal \cite{CdV2023spectrum} for the non-stratified problem, so we only outline the key steps below.
The proof can be extended to matrix operators that only involve pseudo-differential operators of degree $0$.
This condition is not satisfied by the operator $\boldsymbol{\mathcal{T}}$ defined in equation (\ref{eq:selfadjointT}b).
Thus, the starting point is to consider the matrix operator $\boldsymbol{\mathcal{L}}$ defined in equation (\ref{eq:selfadjointT}a).
The fact that this operator is not self-adjoint is not a problem in the analysis (i.e. in the trace formula). 
Then, we apply the tools of microlocal analysis in the presence of boundary effects. 
They allow us to get the asymptotics of $\sum_{j=1}^{d_n} f(\omega_j)$ when $n \to \infty$ for any function $f\in C_0 (\mathbb{R}\setminus 0)$ (i.e. the set of real-valued continuous functions vanishing near zero).
The analysis heavily relies on the principal symbol of the operator $\boldsymbol{\mathcal{L}}$ when $\omega \neq 0$.  
This gives the asymptotic measure in the interval $I_1$ (where the velocity equation is hyperbolic). 
Note that we avoid the value $\omega=0$ because $\ker \boldsymbol{\mathcal{L}}$ is not clearly related to the steady solutions of the original physical problem (due to the existence of generalized eigenvectors). 
\end{proof}

Proposition \ref{theo:dens} shows that $\sigma_1$ is dense in $I_1$ when $n\to\infty$. 
Next, we compare in figure \ref{fig:bench} the asymptotic density with the density obtained from polynomial computations at a fixed degree $n=35$.
The asymptotic density is normalised such that $\int_{-\infty}^\infty \mathrm{d} \pi_\infty = 1$ but, since Proposition \ref{theo:dens} heavily relies on the dispersion relation of IGWs, it is assumed that $\mathrm{d} \pi_\infty/\mathrm{d} \omega = 0$ when $\omega \in I_2$. 
To have a correct comparison, we have thus removed the spectra $\sigma_{2,n}$ before normalising the density.
An excellent agreement is then obtained with the asymptotic density (even for the moderate polynomial degree $n=35$ considered here). 
We have also represented a particular configuration for which $\boldsymbol{1}_\Omega \boldsymbol{\cdot} \boldsymbol{1}_z = 0$, such that $\omega_- = 0$ from expression (\ref{eq:omegapmleblonc}). 
In this case, IGMs are dense in the entire interval $]0,\omega_+[$. 
Note that a similar agreement is found for other configurations (not shown).
Hence, Proposition \ref{theo:dens} can be used to obtain the density of $\sigma_1$. 

The asymptotic behaviour is further illustrated in figure \ref{fig:essential2} for different parameters and geometries.  
The density of pure inertial modes is known to be uniform within the frequency interval $]-2\Omega_s,2\Omega_s[$ for a sphere \cite{CdV2023spectrum}.
Yet, the density of pure internal gravity modes without rotation is not uniform in the sphere (it is larger for the high-frequency modes than for the low-frequency modes). 
Between these two limits, the density is very close to the non-rotating case when $N/\Omega_s \gg 1$.
On the contrary, the density of the high-frequency IGMs with $|\omega|/\omega_{\max}\to~ 1$ approaches the density of the non-stratified modes when $N/\Omega_s \leq 2$. 
We also observe that the geometry does modify the spectral density. 
When $N/\Omega_s < 2$, the density is maximal near $|\omega| \to \omega_{-}$ in elongated (i.e. prolate) ellipsoids with $c/a \gg 1$, whereas it is maximal near $|\omega| \to \omega_+$ in flattened (i.e. oblate) ellipsoids with $c/a \ll 1$. 
Instead, low-frequency IGMs with $|\omega| \to \omega_-$ are favoured in flattened ellipsoids when $N/\Omega_s > 2$. 

\begin{figure}
\centering
\begin{tabular}{cc}
    \includegraphics[width=0.49\textwidth]{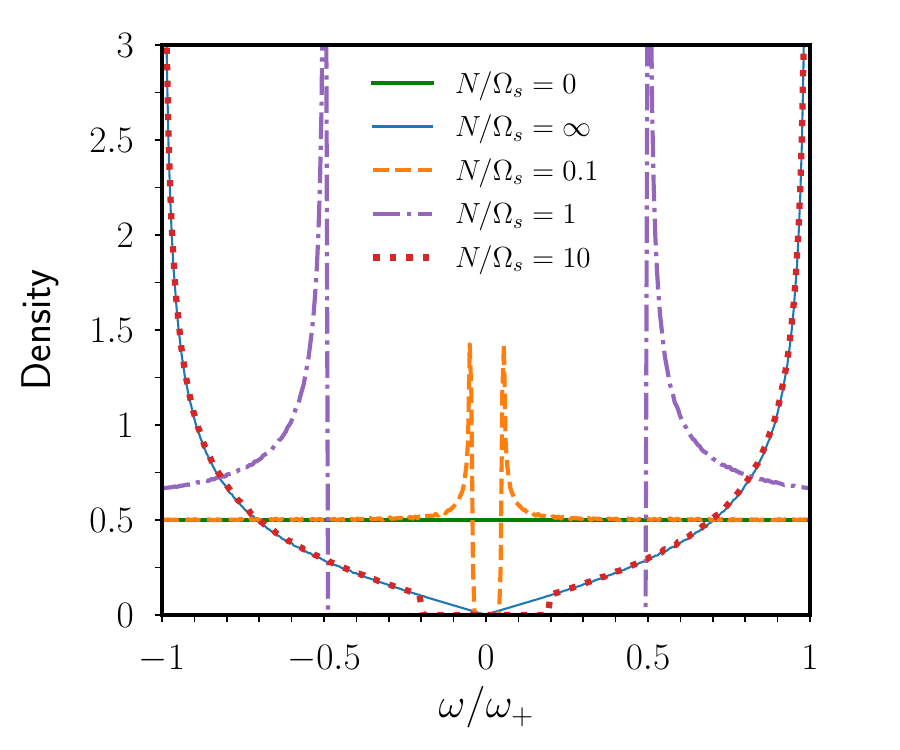} &
    \includegraphics[width=0.49\textwidth]
    {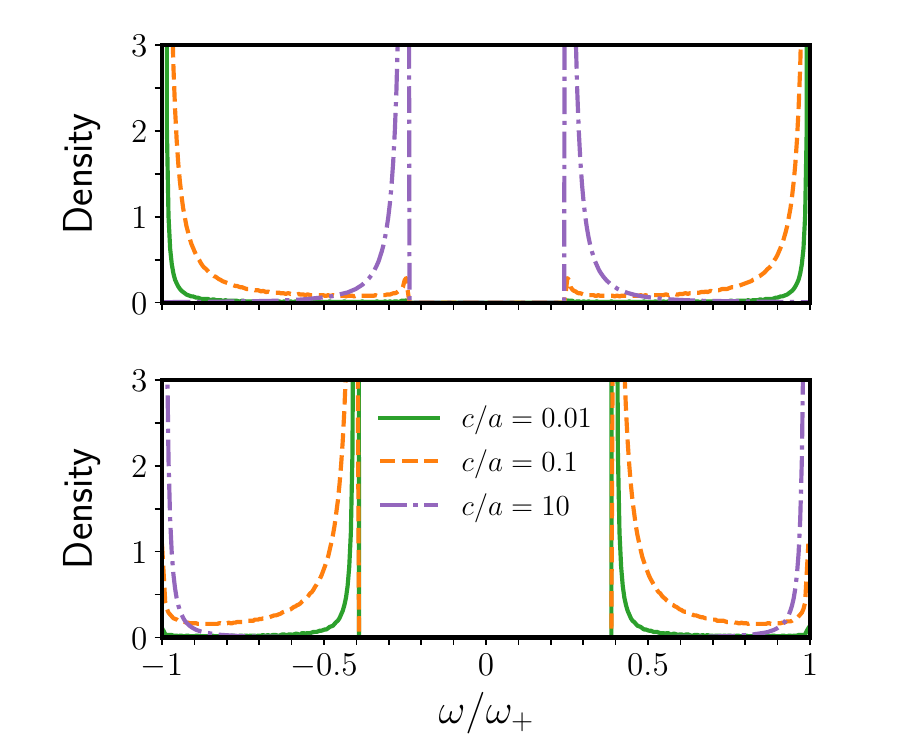} \\
\end{tabular}
\caption{Asymptotic density $\mathrm{d} \pi_\infty/\mathrm{d}\widetilde{\omega}$ of IGMs with normalised eigenvalues $\widetilde{\omega}:=\omega/\omega_{+}$, given by Proposition \ref{theo:dens}. \emph{Left}: Sphere $b/a=c/a=1$ in the aligned case $\boldsymbol{1}_\Omega = \boldsymbol{1}_z$.  \emph{Right}: Axisymmetric ellipsoids with $b/a=1$ and varying $c/a$ ratios. Top panel computed for $N/\Omega_s=0.5$, and bottom one for $N/\Omega_s=5$.}
\label{fig:essential2}
\end{figure}

Finally, the situation $\omega_+=\omega_-$ occurring when $\boldsymbol{1}_\Omega=\boldsymbol{1}_z$ and $N/\Omega_s = 2$ deserves a specific consideration. 
The spectrum $\sigma_1$ is then reduced to the degenerate eigenvalue $\omega/\Omega_s=\pm2$, whose multiplicity can be determined.
Indeed, we know that $\dim(\sigma_{2,n}) = n(n+1)$ from the explicit solutions of equations (\ref{eq:laplaN=2W}a,b).
Because the total number of non-zero eigenvalues is fixed for a degree $n$ (Proposition \ref{theo:nnzVn0}), we have $\dim(\sigma_{1,n})=n(n+1)(n+~2)/3$ when $\omega_+=\omega_-$.

%-----------------------------------------------------------------------
\section{Physical discussion}
\label{sec:discussion}
%-----------------------------------------------------------------------
The idealised model with a constant $N$ has allowed us to identify the key properties of the bounded oscillations in rotating stratified ellipsoids. 
Next, physical insight into geophysical vortices could be obtained by plugging realistic parameters into the model.
Typical geophysical estimates for some large-scale vortices are given in table \ref{table:vortex}. 
These vortices are strongly stratified with $N/\Omega_s \gg 1$ and, consequently, they usually have a very small aspect ratio $c/a \leq 10^{-2}$. 
However, the value of $N$ is still subject to strong uncertainties.
Indeed, the remote determination of the ellipsoidal geometry is only sensitive to the difference between the BV frequency within the vortex and that of the ambient fluid \cite{hassanzadeh2012universal,aubert2012universal,lemasquerier2020remote}.
The detection of IGWs (or IGMs) might be another way to estimate the internal stratification of such vortices.
Various mechanisms could indeed excite such wave motions within a vortex (e.g. as interactions with the surrounding fluid, or bulk turbulence \cite{lanchon2023internal}). 
The bounds given by formula (\ref{eq:omegapmleblonc}), which depend on 
% are a function of 
the orientation of $\boldsymbol{\Omega}$ and on the ratio $N/\Omega_s$, can be simplified in the relevant regime $N/\Omega_s \gg1$. 
It gives
\begin{subequations}
\begin{equation}
    \omega_+^2 \simeq N^2 + 4 \left [ |\boldsymbol{\Omega}|^2 - (\boldsymbol{\Omega} \boldsymbol{\cdot} \boldsymbol{1}_z)^2 \right ], \quad
    \omega_-^2 \simeq 4 |\boldsymbol{\Omega} \boldsymbol{\cdot} \boldsymbol{1}_z|^2.
    \tag{\theequation a,b}
\end{equation}
\end{subequations}
We see that the upper bound approaches $\omega_+^2 \to N^2$ for realistic values $N/\Omega_s \gg 1$. 
Detecting IGMs within a vortex may thus yield an estimate of the vortex's stratification.

A striking property of the global modes is that they admit, under our assumptions, polynomial solutions. 
This behaviour directly results from Proposition \ref{theorem:invariancePcPk}, and IGMs in other configurations can be very different. 
To illustrate this point, we consider a sphere of radius $R$ that is stratified under the radial gravity $\boldsymbol{g} = -(g/R) \, \boldsymbol{r}$. 
The background hydrostatic density is then $\nabla \rho_0 = -(\beta/R) \, \boldsymbol{r}$ (where $\beta>0$ is a constant).
The BV frequency is thus given by $N^2 = N_0^2 (r/R)^2$ using the Boussinesq approximation, where $N_0 = \sqrt{\beta g/\rho_*}$ is the constant value of $N$ at $r=R$. 
Then, the IGMs obey a mixed hyperbolic-elliptic equation that changes of type when \cite{friedlander1982gafd,dintrans1999gravito}
\begin{equation}
    s^2/(\omega^2-4\Omega_s^2) + z^2/\omega^2 = 1/N_0^2,
    \label{eq:turningSPHERE}
\end{equation}
where $s = \sqrt{x^2+y^2}$ is the cylindrical radius measured from the rotation axis $\boldsymbol{\Omega}=\Omega_s \boldsymbol{1}_z$. 
For a given angular frequency $\omega$, equation (\ref{eq:turningSPHERE}) defines a turning surface that separates regions where the equation is hyperbolic and elliptic in $V$. 
Such turning surfaces strongly impact the spatial structure of IGMs, which are mainly trapped in the hyperbolic region \cite{dintrans1999gravito}.
Moreover, IGMs with $|\omega| \to 2 \Omega_s < N$ become trapped in the equatorial waveguide when $N/\Omega_s \gg 1$ \cite{stewartson1976waves}. 
Finally, such turning surfaces are also likely responsible for a non-empty continuous spectrum without diffusion \cite{dintrans1999gravito}. 
On the contrary, the eigenmodes in a stratified ellipsoid with a constant $N$ are always smooth and can penetrate deep inside the volume.

\begin{table}
\caption{Typical estimates for some Jovian vortices and Mediterranean eddies (obtained from \cite{aubert2012universal,lemasquerier2020remote}). Rotation is usually described using the $f-$plane approximation with the Coriolis parameter $f = 2 \boldsymbol{\Omega} \boldsymbol{\cdot} \boldsymbol{1}_z = 2 \Omega_z \cos (\theta)$, where $\theta$ the mean colatitude of the vortex (figure \ref{fig:fig2}).}
\centering
\begin{tabular}{lccccc}
\hline\vspace{1em}
{} & $f$ [rad.s$^{-1}$] & $N/f$ & $a$ [km] & $b$ [km] & $c$ [km] \\
Jupiter's GRS & $\sim 1.4 \times 10^{-4}$ & $\mathcal{O}(10^2)$ & $\sim 9000$ & $\sim 5000$ & $\sim 100$ \\
Jupiter's Oval BA & $\sim 1.9 \times 10^{-4}$ & $\mathcal{O}(10^2)$ & $\sim 3000$ & $\sim 3000$ & $\sim 50$\\
Mediterranean eddies & $\sim 8.5 \times 10^{-5}$ & $10-50$ & $10-10^2$ & $10-10^2$ & $0.5 - 1$ \\
\hline
\end{tabular}
\label{table:vortex}
\end{table}

We have also shown the existence of low-frequency (surface) modes when $0 < |\omega| < \omega_-$.
Their properties are very similar to those of the (trapped) Kelvin waves in rotating shallow-water models \cite{thomson18801}, which owe their existence to topological features \cite{tauber2019bulk,venaille2021wave}. 
Our study brings here a complementary viewpoint.
Indeed, we show that such low-frequency modes can exist in a bounded geometry with rigid walls if the pressure problem fails to be elliptic (due to the BC, see the ESM).
Since this situation is expected to be rather generic, such low-frequency (surface) modes could also exist in more complicated models of stratified vortices.

Finally, we turn our attention to the lifetime of geophysical vortices. 
Without bulk motions, the slow temporal evolution of a stratified vortex is governed by diffusion processes that operate on the slow time scale \cite{facchini2016lifetime,le2021numerical} $\tau_\nu = (a^2/\nu) (2 \Omega_s/N)^2$, where $\nu$ is the kinematic viscosity of the fluid within the vortex. 
Yet, it is known that bulk dissipation can be significantly increased in the presence of IGMs shaped by wave attractors \cite{rieutord1999analogy,dintrans1999gravito}. 
Hence, we could wonder how much energy could be dissipated if some IGMs were forced inside an elliptical vortex.
Without stratification, a pure inertial mode $(\omega, \boldsymbol{u})$ satisfies in an ellipsoid the intriguing integral property  \cite{vantieghem2014inertial,ivers2017enumeration}
\begin{equation}
    \int_V \boldsymbol{u}^\dagger \boldsymbol{\cdot} \nabla^2 \boldsymbol{u} \, \mathrm{d} V = \int_V \boldsymbol{u}^\dagger \boldsymbol{\cdot} \mathbb{L}(\nabla^2 \boldsymbol{u}) \, \mathrm{d} V = 0.
    \label{eq:viscousintegral}
\end{equation}
This could be interpreted as a vanishing viscous dissipation of the inertial modes in an ellipsoid if there were no boundary layers. 
However, when viscous BCs are properly taken into account on the boundary (either with no-slip or stress-free BCs), the inertial modes have a non-zero viscous dissipation \cite{greenspan1968theory,vidal2023precession}.
From a mathematical viewpoint, formula (\ref{eq:viscousintegral}) is valid in an ellipsoid because (i) $\mathbb{L}(\nabla^2 \boldsymbol{u}) \in \boldsymbol{\mathcal{V}}_{n-2}^0$ when $\boldsymbol{u} \in \boldsymbol{\mathcal{V}}_n^0$ and (ii) two inertial eigenvectors $(\boldsymbol{u}_i, \boldsymbol{u}_j)$ are orthogonal with respect to inner product (\ref{eq:innerproduct}) such that $\langle \boldsymbol{u}_i, \boldsymbol{u}_j \rangle = \delta_{ij}$ (when properly normalised). 
Here, the IGMs do not satisfy property (\ref{eq:viscousintegral}) because they are not mutually orthogonal with respect to inner product (\ref{eq:innerproduct}) as explained in Proposition \ref{theo:propQEP}. 
This shows that the IGMs are (at least) subject to viscous dissipation in the bulk on the slow time scale $1/E$ in a stratified vortex. 
In addition, they will also be affected by thermal diffusion (if the stratification has a thermal origin, otherwise by molecular diffusion).
Since the IGMs are here smooth polynomial solutions, it is possible to account for the leading-order dissipative effects using boundary-layer theory (see the ESM). 
This also shows that the IGMs would be dissipated on a typical time scale that is of the same order as $\tau_\nu$ in the absence of motions.
The presence of IGMs may thus affect the slow temporal evolution of stratified vortices. 
Moreover, small-scale turbulence is often present inside stratified vortices, and IGMs could be involved in nonlinear interactions to generate (weakly) turbulent motions \cite{kerswell2002elliptical,lanchon2023internal,boury2023triadic}. 
Therefore, IGMs may play a role in the transition towards turbulence and the long-term evolution of stratified vortices.

%-----------------------------------------------------------------------
\section{Concluding remarks}
\label{sec:ccl}
%-----------------------------------------------------------------------
We have investigated the normal modes of a diffusionless fluid enclosed in a uniformly rotating ellipsoidal vortex, which is stratified in density under a background gravity.
Indeed, they could be important for the dynamics of pancake-like stratified vortices (which are ubiquitous in geophysical flows). 
We have employed the Boussinesq approximation for simplicity, and shown that the problem can be recast as a mixed hyperbolic-elliptic equation for the fluid velocity. 
Next, we have simplified the model by considering a constant BV frequency.  
This idealised model retains the key ingredients to account for the dynamics of geophysical vortices. 
Under these assumptions, we have proved that the problem admits smooth polynomial eigenvectors in triaxial ellipsoids. 
Given the polynomial form of the eigenvectors, we have characterised the spectrum by using a bespoke numerical algorithm and asymptotic spectral theory following our prior study without stratification \cite{CdV2023spectrum}. 
The normal modes consist of high-frequency IGMs and low-frequency (surface) modes.
Using microlocal analysis, we have shown that the boundary plays a key role in sustaining such low-frequency (surface) modes.
Finally, we have discussed our results in the light of geophysical applications, arguing that accounting for IGMs in the models could be useful to (better) understand the dynamics of stratified vortices. 

Several questions have remained unanswered by our work (even in the linear theory). 
A more in-depth investigation of the properties of the low-frequency (surface) modes will be presented in a mathematical work \cite{CdV2024spectrum}. 
From a physical viewpoint, the model could also be extended to account for additional ingredients (while keeping polynomial solutions).
% the exact polynomial nature of the solutions). 
For instance, a uniform-vorticity background flow $\boldsymbol{U}_0$ will be taken into account to study the effects of the vortex's differential rotation (as measured by the Rossby number $Ro$). 
Large-scale vortices (e.g. Jupiter's vortices or Mediterranean eddies) are expected to be in the low-Rossby regime $|Ro|~\lesssim~\mathcal{O}(10^{-1})$, whereas smaller-scale vortices are in the opposite regime $|Ro| \gtrsim \mathcal{O}(1)$ (e.g. vortices at sub-mesoscales $\lesssim 1-10$~km in the Earth's oceans \cite{mcwilliams2016submesoscale}). 
We have also entirely discarded magnetic effects.
This is a reasonable assumption since, for realistic parameters, magnetic fields are expected to only weakly modify the high-frequency IGMs at the size of a vortex. 
However, magnetic fields are known to affect the low-frequency spectrum in rotating, stratified and electrically conducting fluids \cite{friedlander1987hydromagnetic,kerswell1993elliptical}.
Moreover, there is a class of magnetic fields leaving invariant the three-dimensional polynomial vector spaces in an ellipsoid \cite{malkus1967hydromagnetic,gerick2020pressure}.
Therefore, our ellipsoidal model could be used as a toy model to investigate the properties of these Magneto-Archimedean-Coriolis (MAC) modes, by combining microlocal analysis and polynomial computations. 
Other physical ingredients (e.g. a variable BV frequency) would break the exact polynomial nature of the eigenvectors.
Hence, the spectrum should be investigated using another numerical method.

Finally, future studies could build upon our work to understand the onset of turbulence within geophysical vortices. 
Turbulent mixing and dissipation may indeed (partly) explain why some stratified vortices are remarkably stable over long time scales, while some others have much shorter lifetimes.
Several routes exist for the excitation of turbulence in stratified fluids, but direct applications of our work include the study of elliptical instabilities. 
So far, the flows driven by elliptical instabilities have only received scant attention in stratified environments \cite{cebron2010tidal,le2018parametric,vidal2018magnetic,onuki2023breaking}.
The polynomial Galerkin algorithm presented in this work can be extended in the linear theory to determine the onset of the instabilities upon an elliptical flow. 
Therefore, elliptical instabilities will be considered in future work to explore, as a function of the Rossby number, the transition towards turbulence in stratified vortices.

%-----------------------------------------------------------------------
% Acknowledgements
%-----------------------------------------------------------------------
\enlargethispage{20pt}

\dataccess{The paper has an Electronic Supplementary Material (ESM). The source code is available at \url{https://bitbucket.org/vidalje/}.}

\aucontribute{This work is an idea of JV to pursue the interdisciplinary collaboration with YCdV after their joint work \cite{CdV2023spectrum}. JV designed the study and performed the Galerkin computations, whereas YCdV conducted the microlocal analysis. Both authors discussed the results presented in the paper and drafted the manuscript before giving final approval for submission.}

\competing{The authors declare that they have no competing interests.}

\ai{The authors did not use AI-assisted technologies in creating this article.}

\funding{JV received funding from the European Research Council (ERC) under the European Union's Horizon 2020 research and innovation programme (grant agreement No 847433, \textsc{theia} project). The numerical computations were performed using both laboratory resources and using the \textsc{gricad} infrastructure (\url{https://gricad.univ-grenoble-alpes.fr}), which is supported by Grenoble research communities.
A CC-BY public copyright license has been applied by the authors to the present document and will be applied to all subsequent versions up to the Author Accepted Manuscript arising from this submission, in accordance with the grant's open access conditions.}

\ack{JV would like to dedicate this work to the memory of Norman R. Lebovitz (1935-2022) and Stjin Vantieghem (1983-2023), who both revivified the topic and made pioneering contributions. JV acknowledges Bernard Valette and David C\'ebron for fruitful discussions over the years.
YCdV acknowledges Gerd Grubb for her clarification of the definition of a boundary-value elliptic problem.
The authors acknowledge the three anonymous referees, whose valuable comments helped them to
improve the clarity and quality of the manuscript.}

%%%%%%%%%% Insert bibliography here %%%%%%%%%%%%%%

{
\bibliography{./bibwaves}

\begin{thebibliography}{99}

\bibitem{mowbray1967theoretical}
Mowbray DE, Rarity BSH. 1967  A theoretical and experimental investigation of
  the phase configuration of internal waves of small amplitude in a density
  stratified liquid. {\em J. Fluid Mech.} \textbf{28}, 1--16.
(\href{https://doi.org/10.1017/S0022112067001867}{doi:10.1017/S0022112067001867}).

\bibitem{greenspan1968theory}
Greenspan HP. 1968 {\em The theory of rotating fluids}.
Cambridge, UK: Cambridge University Press.

\bibitem{leblond1981waves}
LeBlond PH, Mysak LA. 1981 {\em Waves in the ocean}.
Amsterdam: Elsevier.

\bibitem{dauxois2018instabilities}
Dauxois T, Joubaud S, Odier P, Venaille A. 2018  Instabilities of internal
  gravity wave beams. {\em Annu. Rev. Fluid Mech.} \textbf{50}, 131--156.
(\href{https://doi.org/10.1146/annurev-fluid-122316-044539}{doi:10.1146/annurev-fluid-122316-044539}).

\bibitem{le2023wave}
Le~Bars M. 2023  Wave turbulence in geophysical flows. {\em J. Fluid Mech.}
  \textbf{962}, F1.
(\href{https://doi.org/10.1017/jfm.2023.219}{doi:10.1017/jfm.2023.219}).

\bibitem{mcwilliams2016submesoscale}
McWilliams JC. 2016  Submesoscale currents in the ocean. {\em Proc. R. Soc. A}
  \textbf{472}, 20160117.
(\href{https://doi.org/10.1098/rspa.2016.0117}{doi:10.1098/rspa.2016.0117}).

\bibitem{hopfinger1993vortices}
Hopfinger EJ, Van~Heijst GJF. 1993  Vortices in rotating fluids. {\em Annu.
  Rev. Fluid Mech.} \textbf{25}, 241--289.
(\href{https://doi.org/10.1146/annurev.fl.25.010193.001325}{doi:10.1146/annurev.fl.25.010193.001325}).

\bibitem{billant2001self}
Billant P, Chomaz JM. 2001  Self-similarity of strongly stratified inviscid
  flows. {\em Phys. Fluids} \textbf{13}, 1645--1651.
(\href{https://doi.org/10.1063/1.1369125}{doi:10.1063/1.1369125}).

\bibitem{hassanzadeh2012universal}
Hassanzadeh P, Marcus PS, Le~Gal P. 2012  The universal aspect ratio of
  vortices in rotating stratified flows: theory and simulation. {\em J. Fluid
  Mech.} \textbf{706}, 46--57.
(\href{https://doi.org/10.1017/jfm.2012.180}{doi:10.1017/jfm.2012.180}).

\bibitem{aubert2012universal}
Aubert O, Le~Bars M, Le~Gal P, Marcus PS. 2012  The universal aspect ratio of
  vortices in rotating stratified flows: experiments and observations. {\em J.
  Fluid Mech.} \textbf{706}, 34--45.
(\href{https://doi.org/10.1017/jfm.2012.176}{doi:10.1017/jfm.2012.176}).

\bibitem{lemasquerier2020remote}
Lemasquerier D, Facchini G, Favier B, Le~Bars M. 2020  Remote determination of
  the shape of {Jupiter's} vortices from laboratory experiments. {\em Nat.
  Phys.} \textbf{16}, 695--700.
(\href{https://doi.org/10.1038/s41567-020-0833-9}{doi:10.1038/s41567-020-0833-9}).

\bibitem{richardson2000census}
Richardson PL, Bower AS, Zenk W. 2000  A census of {Meddies} tracked by floats.
  {\em Prog. Oceanogr.} \textbf{45}, 209--250.
(\href{https://doi.org/10.1016/S0079-6611(99)00053-1}{doi:10.1016/S0079-6611(99)00053-1}).

\bibitem{marcus1993jupiter}
Marcus PS. 1993  {Jupiter's} {Great} {Red} {Spot} and other vortices. {\em
  Annu. Rev. Astron. Astrophys.} \textbf{31}, 523--569.
(\href{https://doi.org/10.1146/annurev.aa.31.090193.002515}{doi:10.1146/annurev.aa.31.090193.002515}).

\bibitem{mckiver2015ellipsoidal}
McKiver WJ. 2015  The ellipsoidal vortex: {A} novel approach to geophysical
  turbulence. {\em Adv. Math. Phys.} \textbf{2015}.
(\href{https://doi.org/10.1155/2015/613683}{doi:10.1155/2015/613683}).

\bibitem{tsang2015ellipsoidal}
Tsang YK, Dritschel DG. 2015  Ellipsoidal vortices in rotating stratified
  fluids: beyond the quasi-geostrophic approximation. {\em J. Fluid Mech.}
  \textbf{762}, 196--231.
(\href{https://doi.org/10.1017/jfm.2014.630}{doi:10.1017/jfm.2014.630}).

\bibitem{facchini2016lifetime}
Facchini G, Le~Bars M. 2016  On the lifetime of a pancake anticyclone in a
  rotating stratified flow. {\em J. Fluid Mech.} \textbf{804}, 688--711.
(\href{https://doi.org/10.1017/jfm.2016.549}{doi:10.1017/jfm.2016.549}).

\bibitem{sipp1999vortices}
Sipp D, Lauga E, Jacquin L. 1999  Vortices in rotating systems: centrifugal,
  elliptic and hyperbolic type instabilities. {\em Phys. Fluids} \textbf{11},
  3716--3728.
(\href{https://doi.org/10.1063/1.870180}{doi:10.1063/1.870180}).

\bibitem{yim2016stability}
Yim E, Billant P, M{\'e}nesguen C. 2016  Stability of an isolated pancake
  vortex in continuously stratified-rotating fluids. {\em J. Fluid Mech.}
  \textbf{801}, 508--553.
(\href{https://doi.org/10.1017/jfm.2016.402}{doi:10.1017/jfm.2016.402}).

\bibitem{kerswell2002elliptical}
Kerswell RR. 2002  Elliptical instability. {\em Annu. Rev. Fluid Mech.}
  \textbf{34}, 83--113.
(\href{https://doi.org/10.1146/annurev.fluid.34.081701.171829}{doi:10.1146/annurev.fluid.34.081701.171829}).

\bibitem{miyazaki1992three}
Miyazaki T, Fukumoto Y. 1992  Three-dimensional instability of strained
  vortices in a stably stratified fluid. {\em Phys. Fluids} \textbf{4},
  2515--2522.
(\href{https://doi.org/10.1063/1.858438}{doi:10.1063/1.858438}).

\bibitem{kerswell1993elliptical}
Kerswell RR. 1993  Elliptical instabilities of stratified, hydromagnetic waves.
  {\em Geophys. Astrophys. Fluid Dyn.} \textbf{71}, 105--143.
(\href{https://doi.org/10.1080/03091929308203599}{doi:10.1080/03091929308203599}).

\bibitem{le2006thermo}
Le~Bars M, Le~Diz{\`e}s S. 2006  Thermo-elliptical instability in a rotating
  cylindrical shell. {\em J. Fluid Mech.} \textbf{563}, 189--198.
(\href{https://doi.org/10.1017/S0022112006001674}{doi:10.1017/S0022112006001674}).

\bibitem{vidal2019fossil}
Vidal J, C{\'e}bron D, Ud-Doula A, Alecian E. 2019  Fossil field decay due to
  nonlinear tides in massive binaries. {\em Astron. Astrophys.} \textbf{629},
  A142.
(\href{https://doi.org/10.1051/0004-6361/201935658}{doi:10.1051/0004-6361/201935658}).

\bibitem{cartan1922petites}
Cartan ME. 1922  Sur les petites oscillations d'une masse fluide. {\em Bull.
  Sci. Math.} \textbf{46}, 317--369.

\bibitem{poincare1885equilibre}
Poincar{\'e} H. 1885  Sur l'{\'e}quilibre d'une masse fluide anim{\'e}e d'un
  mouvement de rotation. {\em Acta Math.} \textbf{7}, 259--380.

\bibitem{rieutord2000wave}
Rieutord M, Georgeot B, Valdettaro L. 2000  Wave attractors in rotating fluids:
  a paradigm for ill-posed {Cauchy} problems. {\em Phys. Rev. lett.}
  \textbf{85}, 4277.
(\href{https://doi.org/10.1103/PhysRevLett.85.4277}{doi:10.1103/PhysRevLett.85.4277}).

\bibitem{vantieghem2014inertial}
Vantieghem S. 2014  Inertial modes in a rotating triaxial ellipsoid. {\em Proc.
  R. Soc. A} \textbf{470}, 20140093.
(\href{https://doi.org/10.1098/rspa.2014.0093}{doi:10.1098/rspa.2014.0093}).

\bibitem{ivers2017enumeration}
Ivers DJ. 2017  Enumeration, orthogonality and completeness of the
  incompressible {Coriolis} modes in a tri-axial ellipsoid. {\em Geophys.
  Astrophys. Fluid Dyn.} \textbf{111}, 333--354.
(\href{https://doi.org/10.1080/03091929.2017.1330412}{doi:10.1080/03091929.2017.1330412}).

\bibitem{backus2017completeness}
Backus G, Rieutord M. 2017  Completeness of inertial modes of an incompressible
  inviscid fluid in a corotating ellipsoid. {\em Phys. Rev. E} \textbf{95},
  053116.
(\href{https://doi.org/10.1103/PhysRevE.95.053116}{doi:10.1103/PhysRevE.95.053116}).

\bibitem{veronis1970analogy}
Veronis G. 1970  The analogy between rotating and stratified fluids. {\em Annu.
  Rev. Fluid Mech.} \textbf{2}, 37--66.
(\href{https://doi.org/10.1146/annurev.fl.02.010170.000345}{doi:10.1146/annurev.fl.02.010170.000345}).

\bibitem{rieutord1999analogy}
Rieutord M, Noui K. 1999  On the analogy between gravity modes and inertial
  modes in spherical geometry. {\em Eur. Phys. J. B} \textbf{9}, 731--738.
(\href{https://doi.org/10.1007/s100510050818}{doi:10.1007/s100510050818}).

\bibitem{friedlander1982gafd}
Friedlander S, Siegmann WL. 1982  Internal waves in a rotating stratified fluid
  in an arbitrary gravitational field. {\em Geophys. Astrophys. Fluid Dyn.}
  \textbf{19}, 267--291.
(\href{https://doi.org/10.1080/03091928208208959}{doi:10.1080/03091928208208959}).

\bibitem{brunet2019linear}
Brunet M, Dauxois T, Cortet PP. 2019  Linear and nonlinear regimes of an
  inertial wave attractor. {\em Phys. Rev. Fluids} \textbf{4}, 034801.
(\href{https://doi.org/10.1103/PhysRevFluids.4.034801}{doi:10.1103/PhysRevFluids.4.034801}).

\bibitem{favier2024super}
Favier B, Le~Dizès S. 2024  Inertial wave super-attractor in a truncated
  elliptic cone. {\em J. Fluid Mech.} \textbf{980}, A6.
(\href{https://doi.org/10.1017/jfm.2024.5}{doi:10.1017/jfm.2024.5}).

\bibitem{maas1997observation}
Maas LRM, Benielli D, Sommeria J, Lam FPA. 1997  Observation of an internal
  wave attractor in a confined, stably stratified fluid. {\em Nature}
  \textbf{388}, 557--561.
(\href{https://doi.org/10.1038/41509}{doi:10.1038/41509}).

\bibitem{pacary2023observation}
Pacary C, Dauxois T, Ermanyuk E, Metz P, Moulin M, Joubaud S. 2023  Observation
  of inertia-gravity wave attractors in an axisymmetric enclosed basin. {\em
  Phys. Rev. Fluids} \textbf{8}, 104802.
(\href{https://doi.org/10.1103/PhysRevFluids.8.104802}{doi:10.1103/PhysRevFluids.8.104802}).

\bibitem{CdV2020attractors}
Colin De~Verdi\`ere Y, Saint-Raymond L. 2020  Attractors for two-dimensional
  waves with homogeneous {Hamiltonians} of degree 0. {\em Commun. Pure Appl.
  Math.} \textbf{73}, 421--462.
(\href{https://doi.org/10.1002/cpa.21845}{doi:10.1002/cpa.21845}).

\bibitem{allen1971some}
Allen JS. 1971  Some aspects of the initial-value problem for the inviscid
  motion of a contained, rotating, weakly-stratified fluid. {\em J. Fluid
  Mech.} \textbf{46}, 1--23.
(\href{https://doi.org/10.1017/S0022112071000363}{doi:10.1017/S0022112071000363}).

\bibitem{friedlander1982jfm}
Friedlander S, Siegmann WL. 1982  Internal waves in a contained rotating
  stratified fluid. {\em J. Fluid Mech.} \textbf{114}, 123--156.
(\href{https://doi.org/10.1017/S002211208200007X}{doi:10.1017/S002211208200007X}).

\bibitem{CdV2023spectrum}
Colin~de Verdi{\`e}re Y, Vidal J. 2023  The spectrum of the {Poincar{\'e}}
  operator in an ellipsoid. arxiv:2305.01369.
(\href{https://doi.org/10.48550/arXiv.2305.01369}{doi:10.48550/arXiv.2305.01369}).

\bibitem{spiegel1960boussinesq}
Spiegel EA, Veronis G. 1960  On the {Boussinesq} approximation for a
  compressible fluid. {\em Astrophys. J.} \textbf{131}, 442--447.
(\href{https://ui.adsabs.harvard.edu/link_gateway/1960ApJ...131..442S/doi:10.1086/146849}{doi:10.1086/146849}).

\bibitem{noir2013precession}
Noir J, C{\'e}bron D. 2013  Precession-driven flows in non-axisymmetric
  ellipsoids. {\em J. Fluid Mec.} \textbf{737}, 412--439.
(\href{https://doi.org/10.1017/jfm.2013.524}{doi:10.1017/jfm.2013.524}).

\bibitem{lam2023supply}
Lam H, Delache A, Godeferd FS. 2023  Supply mechanisms of the geostrophic mode
  in rotating turbulence: interactions with self, waves and eddies. {\em J.
  Fluid Mech.} \textbf{971}, A10.
(\href{https://doi.org/10.1017/jfm.2023.644}{doi:10.1017/jfm.2023.644}).

\bibitem{le2017inertial}
Le~Reun T, Favier B, Barker AJ, Le~Bars M. 2017  Inertial wave turbulence
  driven by elliptical instability. {\em Phys. Rev. Lett.} \textbf{119},
  034502.
(\href{https://doi.org/10.1103/PhysRevLett.119.034502}{doi:10.1103/PhysRevLett.119.034502}).

\bibitem{park2013instabilities}
Park J, Billant P. 2013  Instabilities and waves on a columnar vortex in a
  strongly stratified and rotating fluid. {\em Phys. Fluids} \textbf{25}.
(\href{https://doi.org/10.1063/1.4816512}{doi:10.1063/1.4816512}).

\bibitem{le2008inviscid}
Le~Diz\`es S. 2008  Inviscid waves on a {Lamb--Oseen} vortex in a rotating
  stratified fluid: consequences for the elliptic instability. {\em J. Fluid
  Mech.} \textbf{597}, 283--303.
(\href{https://doi.org/10.1017/S0022112007009780}{doi:10.1017/S0022112007009780}).

\bibitem{friedlander1985internal}
Friedlander S. 1985  Internal oscillations in the {Earth's} fluid core. {\em
  Geophys. J. Int.} \textbf{80}, 345--361.
(\href{https://doi.org/10.1111/j.1365-246X.1985.tb05099.x}{doi:10.1111/j.1365-246X.1985.tb05099.x}).

\bibitem{grubb2008distributions}
Grubb G. 2008 {\em Distributions and operators}.
Graduate Texts in Mathematics. Springer, New York (USA): Springer.

\bibitem{faure2023topo}
Faure F. 2023  Manifestation of the topological index formula in quantum waves
  and geophysical waves. {\em Ann. Henri Lebesgue} \textbf{6}, 449--492.
(\href{https://doi.org/10.5802/ahl.169}{doi:10.5802/ahl.169}).

\bibitem{venaille2023ray}
Venaille A, Onuki Y, Perez N, Leclerc A. 2023  From ray tracing to waves of
  topological origin in continuous media. {\em SciPost Phys.} \textbf{14}, 062.
(\href{https://scipost.org/SciPostPhys.14.4.062}{doi:SciPostPhys.14.4.062}).

\bibitem{colin2020spectral}
Colin~de Verdi{\`e}re Y. 2020  Spectral theory of pseudodifferential operators
  of degree 0 and an application to forced linear waves. {\em Anal. PDE}
  \textbf{13}, 1521--1537.
(\href{https://doi.org/10.2140/apde.2020.13.1521}{doi:10.2140/apde.2020.13.1521}).

\bibitem{leray1934mouvement}
Leray J. 1934  Sur le mouvement d'un liquide visqueux emplissant l'espace. {\em
  Acta Math.} \textbf{63}, 193--248.

\bibitem{foias2001navier}
Foias C, Manley O, Rosa R, Temam R. 2001 {\em {Navier-Stokes} equations and
  turbulence}.
Cambridge, UK: Cambridge University Press.

\bibitem{teed2023solenoidal}
Teed RJ, Dormy E. 2023  Solenoidal force balances in numerical dynamos. {\em J.
  Fluid Mech.} \textbf{964}, A26.
(\href{https://doi.org/10.1017/jfm.2023.332}{doi:10.1017/jfm.2023.332}).

\bibitem{barston1967eigenvalueII}
Barston EM. 1967  Eigenvalue Problem for {Lagrangian} Systems. {II}. {\em J.
  Math. Phys.} \textbf{8}, 1886--1892.
(\href{https://doi.org/10.1063/1.1705433}{doi:10.1063/1.1705433}).

\bibitem{valette1989etude}
Valette B. 1989  \'Etude d'une classe de probl{\`e}mes spectraux. {\em C. R.
  Acad. Sci. Paris} \textbf{309}, 785--788.

\bibitem{wouk1966note}
Wouk A. 1966  A note on square roots of positive operators. {\em SIAM Rev.}
  \textbf{8}, 100--102.
(\href{https://doi.org/10.1137/1008008}{doi:10.1137/1008008}).

\bibitem{dintrans1999gravito}
Dintrans B, Rieutord M, Valdettaro L. 1999  Gravito-inertial waves in a
  rotating stratified sphere or spherical shell. {\em J. Fluid Mech.}
  \textbf{398}, 271--297.
(\href{https://doi.org/10.1017/S0022112099006308}{doi:10.1017/S0022112099006308}).

\bibitem{barston1967eigenvalue}
Barston EM. 1967  Eigenvalue problem for {Lagrangian} systems. {\em J. Math.
  Phys.} \textbf{8}, 523--532.
(\href{https://doi.org/10.1063/1.1705227}{doi:10.1063/1.1705227}).

\bibitem{axler2018neumann}
Axler S, Shin PJ. 2018  The {Neumann} problem on ellipsoids. {\em J. Appl.
  Math. Comput.} \textbf{57}, 261--278.
(\href{https://doi.org/10.1007/s12190-017-1105-4}{doi:10.1007/s12190-017-1105-4}).

\bibitem{rabitti2014inertial}
Rabitti A, Maas LRM. 2014  Inertial wave rays in rotating spherical fluid
  domains. {\em J. Fluid Mech.} \textbf{758}, 621--654.
(\href{https://doi.org/10.1017/jfm.2014.551}{doi:10.1017/jfm.2014.551}).

\bibitem{ivers2015enumeration}
Ivers DJ, Jackson A, Winch D. 2015  Enumeration, orthogonality and completeness
  of the incompressible {Coriolis} modes in a sphere. {\em J. Fluid Mech.}
  \textbf{766}, 468--498.
(\href{https://doi.org/10.1017/jfm.2015.27}{doi:10.1017/jfm.2015.27}).

\bibitem{lebovitz1989stability}
Lebovitz NR. 1989  The stability equations for rotating, inviscid fluids:
  {Galerkin} methods and orthogonal bases. {\em Geophys. Astrophys. Fluid Dyn.}
  \textbf{46}, 221--243.
(\href{https://doi.org/10.1080/03091928908208913}{doi:10.1080/03091928908208913}).

\bibitem{vidal2020compressible}
Vidal J, Su S, C{\'e}bron D. 2020  Compressible fluid modes in rigid
  ellipsoids: towards modal acoustic velocimetry. {\em J. Fluid Mech.}
  \textbf{885}, A39.
(\href{https://doi.org/10.1017/jfm.2019.1004}{doi:10.1017/jfm.2019.1004}).

\bibitem{vidal2020acoustic}
Vidal J, C{\'e}bron D. 2020  Acoustic and inertial modes in planetary-like
  rotating ellipsoids. {\em Proc. R. Soc. A} \textbf{476}, 20200131.
(\href{https://doi.org/10.1098/rspa.2020.0131}{doi:10.1098/rspa.2020.0131}).

\bibitem{peacock2005effect}
Peacock T, Weidman P. 2005  The effect of rotation on conical wave beams in a
  stratified fluid. {\em Exp. Fluids} \textbf{39}, 32--37.
(\href{https://doi.org/10.1007/s00348-005-0955-y}{doi:10.1007/s00348-005-0955-y}).

\bibitem{thomson18801}
Thomson W. 1880  On gravitational oscillations of rotating water. {\em Proc. R.
  Soc. Edinburgh} \textbf{10}, 92--100.
(\href{https://doi.org/10.1017/S0370164600043467}{doi:10.1017/S0370164600043467}).

\bibitem{lanchon2023internal}
Lanchon N, Mora DO, Monsalve E, Cortet PP. 2023  Internal wave turbulence in a
  stratified fluid with and without eigenmodes of the experimental domain. {\em
  Phys. Rev. Fluids} \textbf{8}, 054802.
(\href{https://doi.org/10.1103/PhysRevFluids.8.054802}{doi:10.1103/PhysRevFluids.8.054802}).

\bibitem{stewartson1976waves}
Stewartson K, Walton IC. 1976  On waves in a thin shell of stratified rotating
  fluid. {\em Proc. R. Soc. Lond. A} \textbf{349}, 141--156.
(\href{https://doi.org/10.1098/rspa.1976.0064}{doi:10.1098/rspa.1976.0064}).

\bibitem{tauber2019bulk}
Tauber C, Delplace P, Venaille A. 2019  A bulk-interface correspondence for
  equatorial waves. {\em J. Fluid Mech.} \textbf{868}, R2.
(\href{https://doi.org/10.1017/jfm.2019.233}{doi:10.1017/jfm.2019.233}).

\bibitem{venaille2021wave}
Venaille A, Delplace P. 2021  Wave topology brought to the coast. {\em Phys.
  Rev. Res.} \textbf{3}, 043002.
(\href{https://doi.org/10.1103/PhysRevResearch.3.043002}{doi:10.1103/PhysRevResearch.3.043002}).

\bibitem{le2021numerical}
Le~Bars M. 2021  Numerical study of the {McIntyre} instability around
  {Gaussian} floating vortices in thermal wind balance. {\em Phys. Rev. Fluids}
  \textbf{6}, 093801.
(\href{https://doi.org/10.1103/PhysRevFluids.6.093801}{doi:10.1103/PhysRevFluids.6.093801}).

\bibitem{vidal2023precession}
Vidal J, C{\'e}bron D. 2023  Precession-driven flows in stress-free ellipsoids.
  {\em J. Fluid Mech.} \textbf{954}, A5.
(\href{https://doi.org/10.1017/jfm.2022.976}{doi:10.1017/jfm.2022.976}).

\bibitem{boury2023triadic}
Boury S, Maurer P, Joubaud S, Peacock T, Odier P. 2023  Triadic resonant
  instability in confined and unconfined axisymmetric geometries. {\em J. Fluid
  Mech.} \textbf{957}, A20.
(\href{https://doi.org/10.1017/jfm.2023.58}{doi:10.1017/jfm.2023.58}).

\bibitem{CdV2024spectrum}
Colin~de Verdi{\`e}re Y, Vidal J. 2024  On gravito-inertial surface waves.
  Preprint.
(\href{https://hal.science/hal-04461197}{doi:hal-04461197}).

\bibitem{friedlander1987hydromagnetic}
Friedlander S. 1987  Hydromagnetic waves in the {Earth's} fluid core. {\em
  Geophys. Astrophys. Fluid Dyn.} \textbf{39}, 315--333.
(\href{https://doi.org/10.1080/03091928708208816}{doi:10.1080/03091928708208816}).

\bibitem{malkus1967hydromagnetic}
Malkus WVR. 1967  Hydromagnetic planetary waves. {\em J. Fluid Mech.}
  \textbf{28}, 793--802.
(\href{https://doi.org/10.1017/S0022112067002447}{doi:10.1017/S0022112067002447}).

\bibitem{gerick2020pressure}
Gerick F, Jault D, Noir J, Vidal J. 2020  Pressure torque of torsional
  {Alfv{\'e}n} modes acting on an ellipsoidal mantle. {\em Geophys. J. Int.}
  \textbf{222}, 338--351.
(\href{https://doi.org/10.1093/gji/ggaa166}{doi:10.1093/gji/ggaa166}).

\bibitem{cebron2010tidal}
C{\'e}bron D, Maubert P, Le~Bars M. 2010  Tidal instability in a rotating and
  differentially heated ellipsoidal shell. {\em Geophys. J. Int.} \textbf{182},
  1311--1318.
(\href{https://doi.org/10.1111/j.1365-246X.2010.04712.x}{doi:10.1111/j.1365-246X.2010.04712.x}).

\bibitem{le2018parametric}
Le~Reun T, Favier B, Le~Bars M. 2018  Parametric instability and wave
  turbulence driven by tidal excitation of internal waves. {\em J. Fluid Mech.}
  \textbf{840}, 498--529.
(\href{https://doi.org/10.1017/jfm.2018.18}{doi:10.1017/jfm.2018.18}).

\bibitem{vidal2018magnetic}
Vidal J, C{\'e}bron D, Schaeffer N, Hollerbach R. 2018  Magnetic fields driven
  by tidal mixing in radiative stars. {\em Mon. Not. R. Astron. Soc.}
  \textbf{475}, 4579--4594.
(\href{https://doi.org/10.1093/mnras/sty080}{doi:10.1093/mnras/sty080}).

\bibitem{onuki2023breaking}
Onuki Y, Joubaud S, Dauxois T. 2023  Breaking of internal waves parametrically
  excited by ageostrophic anticyclonic instability. {\em J. Phys. Oceanogr.}
  \textbf{53}, 1591--1613.
(\href{https://doi.org/10.1175/JPO-D-22-0152.1}{doi:10.1175/JPO-D-22-0152.1}).

\end{thebibliography}


\begin{thebibliography}{99}

\bibitem{greenspan1968theory}
Greenspan HP. 1968 {\em The theory of rotating fluids}.
Cambridge (UK): Cambridge University Press.

\bibitem{friedlander1982gafd}
Friedlander S, Siegmann WL. 1982a  Internal waves in a rotating stratified
  fluid in an arbitrary gravitational field. {\em Geophys. Astrophys. Fluid
  Dyn.} \textbf{19}, 267--291.
(\href{https://doi.org/10.1080/03091928208208959}{doi:10.1080/03091928208208959}).

\bibitem{friedlander1982jfm}
Friedlander S, Siegmann WL. 1982b  Internal waves in a contained rotating
  stratified fluid. {\em J. Fluid Mech.} \textbf{114}, 123--156.
(\href{https://doi.org/10.1017/S002211208200007X}{doi:10.1017/S002211208200007X}).

\bibitem{schutz1980geometrical}
Schutz BF. 1980 {\em Geometrical methods of {Mathematical} {Physics}}.
Cambridge (UK): Cambridge University Press.

\bibitem{CdV2024spectrum}
Colin~de Verdi{\`e}re Y, Vidal J. 2024  On gravito-inertial surface waves.
  Preprint.
(\href{https://hal.science/hal-04461197}{doi:hal-04461197}).

\bibitem{taylor2013partial}
Taylor M. 2013 {\em Partial differential equations {II}: {Qualitative} studies
  of linear equations}.
New York (USA): Springer.

\bibitem{grubb2008distributions}
Grubb G. 2008 {\em Distributions and operators}.
Graduate Texts in Mathematics. Springer, New York (USA): Springer.

\bibitem{miranda1970partial}
Miranda C. 1970 {\em Partial differential equations of elliptic type}.
Springer-Verlag.

\bibitem{reed1972mathsphys}
Reed M, Simon B. 1972 {\em Methods of modern {Mathematical} {Physics}. {I}.
  {Functional} analysis}.
San Diego (USA): Academic Press.

\bibitem{colin2020spectral}
Colin~de Verdi{\`e}re Y. 2020  Spectral theory of pseudodifferential operators
  of degree $0$ and an application to forced linear waves. {\em Anal. PDE}
  \textbf{13}, 1521--1537.
(\href{https://doi.org/10.2140/apde.2020.13.1521}{doi:10.2140/apde.2020.13.1521}).

\bibitem{dintrans1999gravito}
Dintrans B, Rieutord M, Valdettaro L. 1999  Gravito-inertial waves in a
  rotating stratified sphere or spherical shell. {\em J. Fluid Mech.}
  \textbf{398}, 271--297.
(\href{https://doi.org/10.1017/S0022112099006308}{doi:10.1017/S0022112099006308}).

\bibitem{hassanzadeh2012universal}
Hassanzadeh P, Marcus PS, Le~Gal P. 2012  The universal aspect ratio of
  vortices in rotating stratified flows: theory and simulation. {\em J. Fluid
  Mech.} \textbf{706}, 46--57.
(\href{https://doi.org/10.1017/jfm.2012.180}{doi:10.1017/jfm.2012.180}).

\bibitem{aubert2012universal}
Aubert O, Le~Bars M, Le~Gal P, Marcus PS. 2012  The universal aspect ratio of
  vortices in rotating stratified flows: experiments and observations. {\em J.
  Fluid Mech.} \textbf{706}, 34--45.
(\href{https://doi.org/10.1017/jfm.2012.176}{doi:10.1017/jfm.2012.176}).

\bibitem{lemasquerier2020remote}
Lemasquerier D, Facchini G, Favier B, Le~Bars M. 2020  Remote determination of
  the shape of {Jupiter's} vortices from laboratory experiments. {\em Nat.
  Phys.} \textbf{16}, 695--700.
(\href{https://doi.org/10.1038/s41567-020-0833-9}{doi:10.1038/s41567-020-0833-9}).

\bibitem{friedlander1989asymptotic}
Friedlander S. 1989  Asymptotic behaviour of decay rates of internal waves in a
  rotating stratified spherical shell. {\em Geophys. J. Int.} \textbf{96},
  245--252.
(\href{https://doi.org/10.1111/j.1365-246X.1989.tb04448.x}{doi:10.1111/j.1365-246X.1989.tb04448.x}).

\bibitem{rieutord1997ekman}
Rieutord M, Zahn JP. 1997  Ekman pumping and tidal dissipation in close
  binaries: a refutation of {Tassoul's} mechanism. {\em Astrophys. J.}
  \textbf{474}, 760.
(\href{https://dx.doi.org/10.1086/303494}{doi:10.1086/303494}).

\bibitem{vidal2023precession}
Vidal J, C{\'e}bron D. 2023  Precession-driven flows in stress-free ellipsoids.
  {\em J. Fluid Mech.} \textbf{954}, A5.
(\href{https://doi.org/10.1017/jfm.2022.976}{doi:10.1017/jfm.2022.976}).

\bibitem{french2012ab}
French M, Becker A, Lorenzen W, Nettelmann N, Bethkenhagen M, Wicht J, Redmer
  R. 2012  Ab initio simulations for material properties along the {Jupiter}
  adiabat. {\em Astrophys. J. Suppl. Ser.} \textbf{202}, 5.
(\href{https://dx.doi.org/10.1088/0067-0049/202/1/5}{doi:10.1088/0067-0049/202/1/5}).

\end{thebibliography}
\bibliographystyle{RS.bst}
}

%#\includepdf[pages={1-10}]{./SuppMat.pdf}

\end{document}

\title{Electronic Supplementary Material of "Inertia-gravity waves in geophysical vortices"}

\author{J\'er\'emie Vidal\footnote{Universit\'e Grenoble Alpes, CNRS,  ISTerre, 38000 Grenoble, France} \href{https://orcid.org/0000-0002-3654-6633}{\orcidicon} \& Yves Colin de Verdi\`ere\footnote{Universit\'e Grenoble Alpes, CNRS, Institut Fourier, 38000 Grenoble, France}}

\maketitle

\begin{abstract}
We present the derivation of the boundary-value pressure problem governing the normal modes of an incompressible fluid, which is rotating and stratified in density under the Bousinesq approximation. 
Next, we discuss for which conditions this is an elliptic boundary-value problem (using microlocal analysis). 
Applications to ellipsoidal vortices with a constant Brunt-V\"ais\"al\"a frequency are then considered. 
Finally, we introduce a boundary-layer analysis to determine the (leading-order) diffusive decay rates of the normal modes in an ellipsoid.
\end{abstract}

%-----------------------------------------------------------------------
\section{Boundary-value pressure problem}
%-----------------------------------------------------------------------
We consider a diffusionless and incompressible fluid, which is enclosed in a smooth bounded domain $V$ (whose boundary is denoted by $\partial V$). 
The fluid is rotating at the angular velocity $\boldsymbol{\Omega} = (\Omega_x, \Omega_y, \Omega_z)^\top$ with respect to an inertial frame, and is stratified in density under the imposed gravity field $\boldsymbol{g} = -g(z) \boldsymbol{1}_z$ within the Boussinesq approximation. 
At rest, the pressure $P_0$ and the density satisfy the hydrostatic balance $\nabla P_0 = [\rho_*+\rho_0(z)] \boldsymbol{g}$, where $\rho_*$ is the typical (mean) density of the fluid within the Boussinesq approximation. 
Then, we seek small-amplitude oscillatory perturbations (where $\omega \in \mathbb{R}$ is the angular frequency) for the velocity $\boldsymbol{u}\exp(\mathrm{i} \omega t)$, the pressure $\Phi\exp(\mathrm{i} \omega t)$, and the density $\zeta \exp(\mathrm{i} \omega t)$ upon the hydrostatic basic state.  
These perturbations are governed in the rotating frame by the incompressible Euler equations
\begin{subequations}
\label{eq:euler}
\begin{equation}
\mathrm{i} \omega \boldsymbol{u} + 2 \left ( \boldsymbol{\Omega} \times \boldsymbol{u} \right ) = - \nabla \Phi - (\zeta/\rho^*) g(z) \boldsymbol{1}_z, \quad \nabla \boldsymbol{\cdot} \boldsymbol{u} = 0,
\tag{\theequation a,b}
\end{equation}
\end{subequations}
and the buoyancy equation
\begin{subequations}
\label{eq:mass}
\begin{equation}
    \mathrm{i} \omega \zeta = \frac{\rho_* N^2(z)}{g(z)} \left (\boldsymbol{u} \boldsymbol{\cdot} \boldsymbol{1}_z \right ), \quad N^2(z) := - \frac{g(z)}{\rho_*} \frac{\mathrm{d} \rho_0}{\mathrm{d}z},
    \tag{\theequation a,b}
\end{equation}
\end{subequations}
where $N(z)$ is the Brunt-V\"ais\"al\"a (BV) frequency. 
The latter equations can be reduced, as long as $\omega\neq0$, to a quadratic equation in $\omega$ for the velocity given by
\begin{equation}
    -\omega^2 \boldsymbol{u} + \mathrm{i} \omega  \left (2 \boldsymbol{\Omega} \times \boldsymbol{u} \right ) + N^2 (z) \left ( \boldsymbol{u} \boldsymbol{\cdot} \boldsymbol{1}_z \right ) \boldsymbol{1}_z = -\mathrm{i} \omega \nabla \Phi.
    \label{eq:waveeqU}
\end{equation}
Finally, equation (\ref{eq:waveeqU}) is supplemented with the no-penetration boundary condition (BC) given by $\left . \boldsymbol{u} \boldsymbol{\cdot} \boldsymbol{n} \right |_{\partial V}= 0$ (i.e. $\boldsymbol{u}$ is tangent to $\partial V$), where $\boldsymbol{n}$ is the unit vector normal to the boundary.\vspace{0.5em}

Our goal is to find an alternative formulation of the problem in terms of the pressure $\Phi$. 
An equation for the pressure can be obtained by taking the divergence of equation (\ref{eq:waveeqU}), which gives
\begin{equation}
    -\mathrm{i} \omega \nabla^2 \Phi = \mathrm{i} \omega  \nabla \boldsymbol{\cdot} \left (2 \boldsymbol{\Omega} \times \boldsymbol{u} \right ) + ( \boldsymbol{u} \boldsymbol{\cdot} \boldsymbol{1}_z ) \partial_z N^2 + N^2(z) \partial_z ( \boldsymbol{u} \boldsymbol{\cdot} \boldsymbol{1}_z ).
    \label{eq:laplaPu}
\end{equation}
The associated boundary condition (BC) is then obtained by taking the scalar product of equation (\ref{eq:waveeqU}) with $\boldsymbol{n}$, which gives
\begin{equation}
    -\mathrm{i} \omega (\nabla \Phi \boldsymbol{\cdot} \boldsymbol{n}) = \mathrm{i} \omega  (2 \boldsymbol{\Omega} \times \boldsymbol{u} ) \boldsymbol{\cdot} \boldsymbol{n} + N^2 (z) ( \boldsymbol{u} \boldsymbol{\cdot} \boldsymbol{1}_z ) (\boldsymbol{1}_z \boldsymbol{\cdot} \boldsymbol{n}) \quad \text{on $\partial V$}.
    \label{eq:BCPu}
\end{equation}
However, equations (\ref{eq:laplaPu})-(\ref{eq:BCPu}) still depend on $\boldsymbol{u}$. 
Therefore, an expression of $\boldsymbol{u}$ in terms of $\Phi$ only is required before obtaining the equation for the pressure (e.g. the Poincar\'e problem, see Box \ref{box:Poincare}). 
We will give below the explicit formula(s) accounting for misaligned rotation and gravity with a varying BV frequency $N(z)$.

\begin{encart}[float=t,label={box:Poincare}]{Poincar\'e problem}
\noindent When $N=0$, the velocity is given by (see formula 2.7.2 in Greenspan \cite{greenspan1968theory})
\begin{equation}
(4 |\boldsymbol{\Omega}|^2 - \omega^2) \boldsymbol{u} = -\mathrm{i} \omega \nabla \Phi - \frac{4}{\mathrm{i}\omega} (\boldsymbol{\Omega} \boldsymbol{\cdot} \nabla \Phi) \boldsymbol{\Omega} + 2 \boldsymbol{\Omega} \times \nabla \Phi.
\end{equation}
The divergence of the latter equation yields the so-called Poincar\'e equation \cite{greenspan1968theory}
\begin{equation}
\omega^2 \nabla^2 \Phi = (2 \boldsymbol{\Omega} \cdot \nabla)^2 \Phi.
\end{equation}
The associated BC, obtained from equation (\ref{eq:BCP2}) below, is given by
\begin{equation}
\omega^2 \nabla \Phi \boldsymbol{\cdot} \boldsymbol{n} =  \mathrm{i} \omega (2\boldsymbol{\Omega} \times \boldsymbol{n}) \boldsymbol{\cdot} \nabla \Phi + 4 (\boldsymbol{\Omega} \boldsymbol{\cdot} \nabla \Phi) (\boldsymbol{\Omega} \boldsymbol{\cdot} \boldsymbol{n}) \quad \text{on $\partial V$},
\end{equation}
which is identical to equation (2.7.4) in Greenspan \cite{greenspan1968theory}.
\end{encart}

%-----------------------------------------------------------------------
\subsection{Second-order equation in the volume}
%-----------------------------------------------------------------------
To find the equation satisfied by the pressure $\Phi$ in $V$, we first find the equation relating $\boldsymbol{u}$ and $\Phi$ (Proposition \ref{prop:uP}). Then, we enforce the incompressibility condition to obtain the pressure equation (Proposition \ref{prop:pressureEq}). 
Finally, we give in Proposition \ref{eq:laplaPNcst} the pressure equation when $N$ is constant. 
It will be of prime importance for the microlocal analysis presented in Section \ref{sec:microlocal}.

\begin{proposition}
The velocity $\boldsymbol{u}$ is related to the pressure $\Phi$ by the equation $\mathcal{D}_\omega (z) \boldsymbol{u} = \boldsymbol{f}$ when $\mathcal{D}_\omega (z) = (\omega^2 - \omega_+^2)(\omega^2 - \omega_-^2) \neq 0$ with 
\begin{equation}
    \omega_\pm^2 = \frac{1}{2} (N^2(z) + 4 |\boldsymbol{\Omega}|^2) \pm \sqrt{\frac{1}{4}(N^2(z) + 4 |\boldsymbol{\Omega}|^2)^2 - 4 N^2(z) (\boldsymbol{\Omega} \boldsymbol{\cdot} \boldsymbol{1}_z)^2},
    \label{eq:wpmNz}
\end{equation}
and with 
\begin{multline}
\boldsymbol{f} = \mathrm{i} \omega^3 \nabla \Phi - 2 \omega^2 (\boldsymbol{\Omega} \times \nabla \Phi) - 4 \mathrm{i} \omega (\boldsymbol{\Omega} \boldsymbol{\cdot} \nabla \Phi) \boldsymbol{\Omega} \\ 
+ N^2 (z) \left [ (2\boldsymbol{\Omega} \boldsymbol{\cdot} \boldsymbol{1}_z) (\boldsymbol{1}_z \times \nabla \Phi) + \mathrm{i} \omega \boldsymbol{1}_z \times (\boldsymbol{1}_z \times \nabla \Phi) \right ].
\end{multline}
\label{prop:uP}
\end{proposition}
\vspace{-0.75em}
\begin{proof}
We recast Euler equation (\ref{eq:euler}a) and buoyancy equation (\ref{eq:mass}) as a linear system given by
\begin{subequations}
\begin{equation}
\boldsymbol{M} \begin{pmatrix}
u_x \\
u_y \\
u_z \\
\zeta/\rho_* \\
\end{pmatrix}
 = - \begin{pmatrix}
\partial_x \Phi \\
\partial_y \Phi \\
\partial_z \Phi \\
0 \\
\end{pmatrix}, \quad 
\boldsymbol{M} := \begin{pmatrix}
    \mathrm{i} \omega & -2 \Omega_z & 2 \Omega_y & 0 \\
    2 \Omega_z & \mathrm{i} \omega & -2\Omega_x & 0 \\
    -2 \Omega_y & 2\Omega_x & \mathrm{i} \omega & g(z) \\
    0 & 0 & -N^2(z)/g(z) & \mathrm{i} \omega \\
\end{pmatrix},
\tag{\theequation a,b}
\end{equation}
\end{subequations}
and with $\boldsymbol{u} = (u_x, u_y, u_z)^\top$.
The system admits non-zero solutions when $\det( \boldsymbol{M}) := \mathcal{D}_\omega (z) \neq 0$.
In this case, we can invert the linear system and the result follows. Note that it reduces to equation (2.7.2) in Greenspan \cite{greenspan1968theory} when $N=0$.
\end{proof}

\begin{proposition}
When $\mathcal{D}_\omega (z) \neq 0$, the pressure obeys the second-order equation given by
\begin{equation}
    \nabla \boldsymbol{\cdot} (\boldsymbol{f}/\mathcal{D}_\omega) = 0 \iff \nabla \boldsymbol{\cdot} \boldsymbol{f} - (1/\mathcal{D}_\omega) \, f_z \partial_z {D}_\omega = 0,
    \label{eq:Pvargen}
\end{equation}
with $\partial_z \mathcal{D}_\omega = -(\omega^2-\omega_-^2) \partial_z \omega_+^2 - (\omega^2-\omega_+^2) \partial_z \omega_-^2$. 
\label{prop:pressureEq}
\end{proposition}
\vspace{0.5em}
\noindent When $N(z)$ is variable and $\boldsymbol{\Omega} = \Omega_z \boldsymbol{1}_z$, equation (\ref{eq:Pvargen}) reduces to 
\begin{equation}
    \nabla^2 \Phi + \frac{\partial}{\partial z} \left ( \frac{N^2 - 4 \Omega_z^2}{\omega^2 - N^2} \frac{\partial \Phi}{\partial z} \right ) = 0,
    \label{eq:pressEqalignedNz}
\end{equation}
which is identical to equation (5.1) in Friedlander \& Siegmann \cite{friedlander1982gafd}. 
For oscillatory motions with $\omega \neq 0$, equation (\ref{eq:laplaP}) is thus equivalent to equation (\ref{eq:waveeqU}) when $|\omega| \neq \omega_\pm(z)$.
However, the pressure must obey another equation $\mathcal{D}_\omega (z) =0$ (see below). 

\begin{proposition}
When $N$ is constant, equation (\ref{eq:Pvargen}) reduces to $\mathcal{P}_\omega (\Phi) = 0$ with
\begin{equation}
\mathcal{P}_\omega := - \nabla^2 - \frac{N^2}{\omega^2 - N^2} \left ( \boldsymbol{1}_z \boldsymbol{\cdot} \nabla \right )^2 + \frac{4}{\omega^2 - N^2} \left ( \boldsymbol{\Omega} \boldsymbol{\cdot} \nabla \right )^2.
\label{eq:laplaP}
\end{equation}
The operator $\mathcal{P}_\omega$ is elliptic when $0<|\omega| < |\omega_-$, and its principal symbol $\mathfrak{p}$ is given by
\begin{equation}
    \mathfrak{p} := A k_x^2 + B k_y^2 + C k_z^2 + 2D k_y k_z + 2E k_x k_z + 2F k_x k_y
    \label{eq:symbolPw}
\end{equation}
with 
% $A = 1 - {4 \Omega_x^2}/(\omega^2 - N^2)$, $B = 1 - {4 \Omega_y^2}/(\omega^2 - N^2)$, $C = 1 + (N^2 - 4 \Omega_z^2)/(\omega^2 - N^2) = (\omega^2 - 4 \Omega_z^2)(\omega^2-N^2)$, $D = - {4 \Omega_y \Omega_z}/(\omega^2 - N^2)$, $E = - {4 \Omega_x \Omega_z}/(\omega^2 - N^2)$ and $F = - {4 \Omega_x \Omega_y}/(\omega^2 - N^2)$.
\begin{equation*}
    A=\frac{\omega^2 - N^2 - 4 \Omega_x^2}{\omega^2 - N^2}, \ B=\frac{\omega^2 - N^2 - 4 \Omega_y^2}{\omega^2 - N^2}, \ C=\frac{\omega^2 - 4 \Omega_z^2}{\omega^2 - N^2}, 
\end{equation*}
\begin{equation*}
D=\frac{-4 \Omega_y \Omega_z}{\omega^2 - N^2}, \ E=\frac{-4\Omega_x \Omega_z}{\omega^2 - N^2}, \ F=\frac{-4\Omega_x \Omega_y}{\omega^2 - N^2}.
\end{equation*}
\label{eq:laplaPNcst}
\end{proposition}

\noindent Finally, it is worth explicitly writing the Cartesian components of $\boldsymbol{f}=(f_x,f_y,f_z)^\top$ when $N$ is constant (as we will reuse them below).
They are given by
\begin{subequations}
\label{eq:fxfyfz}
\begin{align}
\dfrac{f_x}{\omega^2-N^2} &= \mathrm{i} \omega  [A \partial_x \Phi + F \partial_y \Phi + E \partial_z \Phi] &-& \quad \frac{2 \omega^2 \Omega_y}{\omega^2 - N^2} \partial_z \Phi &+& \quad 2 \Omega_z \partial_y \Phi, \\
\dfrac{f_y}{\omega^2-N^2} &= \mathrm{i} \omega [F \partial_x \Phi + B \partial_y \Phi + D \partial_z \Phi] &+& \quad \frac{2\omega^2 \Omega_x}{\omega^2 - N^2} \partial_z \Phi &-& \quad 2 \Omega_z \partial_x \Phi, \\
\dfrac{f_z}{\omega^2-N^2} &= \mathrm{i} \omega [E \partial_x \Phi + D \partial_y \Phi + C \partial_z \Phi] &+& \quad \frac{2 \omega^2}{\omega^2 - N^2} (\Omega_y \partial_x \Phi - \Omega_x \partial_y \Phi).
\end{align}
where the coefficients $(A,B,C,D,E,F)$ appear in the principal symbol of $\mathcal{P}_\omega$ (Proposition \ref{eq:laplaPNcst}).
\end{subequations}

%-----------------------------------------------------------------------
\subsection{Boundary condition}
%-----------------------------------------------------------------------
\begin{encart}[float=t,label={box:aligned}]{Aligned case}
\noindent
When $\boldsymbol{\Omega} = \Omega_z \boldsymbol{1}_z$ and $N$ is constant, pressure equation (\ref{eq:pressEqalignedNz}) simplifies into
\begin{equation}
   - \nabla^2 \Phi - \frac{N^2 - 4 \Omega_z^2}{ \omega^2 - N^2 } \left ( \boldsymbol{1}_z \boldsymbol{\cdot} \nabla \right )^2 \Phi = 0,
    \label{eq:pressurealigned}
\end{equation}
which is identical to equation (2.21) in Friedlander \& Siegmann \cite{friedlander1982jfm}. 
If required, the velocity is related to the pressure as
\begin{equation}
[4 \Omega_z^2 - \omega^2] \boldsymbol{u} = -\mathrm{i} \omega \nabla \Phi + \frac{\mathrm{i} \omega}{\omega^2 - N^2} [ 4 \Omega_z^2 - N^2] (\boldsymbol{1}_z \boldsymbol{\cdot} \nabla \Phi) \boldsymbol{1}_z + 2 \Omega_z (\boldsymbol{1}_z \times \nabla \Phi).
\end{equation}
To find the pressure BC, we rearrange equation (\ref{eq:BCP2}) as
\begin{equation*}
\mathrm{i} \omega \left [ \omega^2 - N^2 \right ] \nabla \Phi \boldsymbol{\cdot} \boldsymbol{n} =  2 \Omega_z \left [ \omega^2 - N^2 \right ] (\boldsymbol{1}_z \times \nabla \Phi) \boldsymbol{\cdot} \boldsymbol{n} + \mathrm{i} \omega \left [ 4 \Omega_z^2 - N^2 \right ] (\boldsymbol{1}_z \boldsymbol{\cdot} \boldsymbol{n}) \frac{\partial \Phi}{\partial z},
\end{equation*}
which simplifies into
\begin{equation}
 \nabla \Phi \boldsymbol{\cdot} \boldsymbol{n} =  \frac{2 \Omega_z}{\mathrm{i} \omega} (\boldsymbol{n} \times \boldsymbol{1}_z ) \boldsymbol{\cdot} \nabla \Phi - \left [ \frac{N^2 -4 \Omega_z^2}{\omega^2 - N^2} \right ] (\boldsymbol{1}_z \boldsymbol{\cdot} \boldsymbol{n}) \frac{\partial \Phi}{\partial z} \quad \text{on $\partial V$}.
 \label{eq:BCP3}
\end{equation}
BC (\ref{eq:BCP3}) is identical to formula (2.22) in Friedlander \& Siegmann \cite{friedlander1982jfm}.
\end{encart}

\begin{proposition}
The pressure BC is given by $\left . \boldsymbol{f} \boldsymbol{\cdot} \boldsymbol{n} \right |_{\partial V} = 0$, where the components of $\boldsymbol{f}$ are given in equations (\ref{eq:fxfyfz}a-c).
\label{prop:BCp}
\end{proposition}
\vspace{0.5em}
\noindent For the sake of comparison with prior studies, the pressure BC can be rewritten as
\begin{multline}
\mathrm{i} \omega^3 \nabla \Phi \boldsymbol{\cdot} \boldsymbol{n} =  2 \omega^2 (\boldsymbol{\Omega} \times \nabla \Phi) \boldsymbol{\cdot} \boldsymbol{n} + 4 \mathrm{i} \omega (\boldsymbol{\Omega} \boldsymbol{\cdot} \nabla \Phi) (\boldsymbol{\Omega} \boldsymbol{\cdot} \boldsymbol{n}) \\
-  N^2(z) \left [ (2\boldsymbol{\Omega} \boldsymbol{\cdot} \boldsymbol{1}_z) (\boldsymbol{1}_z \times \nabla \Phi) + \mathrm{i} \omega  \boldsymbol{1}_z \times (\boldsymbol{1}_z \times \nabla \Phi) \right ] \boldsymbol{\cdot} \boldsymbol{n}.
\label{eq:BCP1}
\end{multline}
Since we have $\boldsymbol{1}_z \times (\boldsymbol{1}_z \times \nabla \Phi) = (\boldsymbol{1}_z \boldsymbol{\cdot} \nabla \Phi) \boldsymbol{1}_z - \nabla \Phi$, we can rewrite the pressure BC as
\begin{multline}
\mathrm{i} \omega \left [ \omega^2 - N^2(z) \right ] \nabla \Phi \boldsymbol{\cdot} \boldsymbol{n} =  2 \left [ \omega^2 (\boldsymbol{\Omega} \times \nabla \Phi) - N^2(z) (\boldsymbol{\Omega} \boldsymbol{\cdot} \boldsymbol{1}_z) (\boldsymbol{1}_z \times \nabla \Phi) \right ] \boldsymbol{\cdot} \boldsymbol{n} \\
+ \mathrm{i} \omega \left [ 4 (\boldsymbol{\Omega} \boldsymbol{\cdot} \nabla \Phi) \boldsymbol{\Omega} - N^2(z) (\boldsymbol{1}_z \boldsymbol{\cdot} \nabla \Phi) \boldsymbol{1}_z \right ] \boldsymbol{\cdot} \boldsymbol{n}.
\label{eq:BCP2}
\end{multline}
BC (\ref{eq:BCP2}) can be strongly simplified in some cases (see Box \ref{box:aligned}). 

%-----------------------------------------------------------------------
\subsection{First-order equation when $\mathcal{D}_\omega (z) =0$}
%-----------------------------------------------------------------------
The pressure is no longer given by equation (\ref{eq:Pvargen}) when $\mathcal{D}_\omega (z) =0$.
The latter defines the equation of turning surfaces, on which the mathematical problem change of type. 
Onto such surfaces, the pressure obeys another partial differential equation. 
Assuming that the field unknowns remain non-singular across such surfaces, this equation is found by using vector manipulations of the primitive equations (Proposition \ref{prop:eqPfirstorder}).

\begin{proposition}
When $\mathcal{D}_\omega (z) = 0$ (i.e. on turning surfaces), the pressure is given by the first-order equation
\begin{equation}
   (2 \boldsymbol{\Omega} \boldsymbol{\cdot} \boldsymbol{1}_z ) (2\boldsymbol{\Omega} \boldsymbol{\cdot} \nabla \Phi) + 2\mathrm{i} \omega \, \boldsymbol{\Omega} \boldsymbol{\cdot} (\boldsymbol{1}_z \times \nabla \Phi) - \omega^2 (\boldsymbol{1}_z \boldsymbol{\cdot} \nabla \Phi) = 0,
    \label{eq:pressureDw0}
\end{equation}
which is a particular case of equation (3.5) in Friedlander \& Siegmann \cite{friedlander1982gafd}. 
\label{prop:eqPfirstorder}
\end{proposition}
\vspace{-0.75em}
\begin{proof}
We apply $2\boldsymbol{\Omega} \times$ to equation (\ref{eq:euler}a), which gives
\begin{align*}
\mathrm{i} \omega (2\boldsymbol{\Omega} \times \boldsymbol{u}) &= -2\boldsymbol{\Omega} \times (2\boldsymbol{\Omega} \times \boldsymbol{u}) - 2\boldsymbol{\Omega} \times \nabla \Phi - g(z)(\zeta/\rho^*) (2\boldsymbol{\Omega} \times \boldsymbol{1}_z), \\
{} &= -4  \left [ (\boldsymbol{\Omega} \boldsymbol{\cdot} \boldsymbol{u}) \boldsymbol{\Omega} - |\boldsymbol{\Omega}|^2 \boldsymbol{u} \right ] - 2\boldsymbol{\Omega} \times \nabla \Phi - g(z)(\zeta/\rho^*) (2\boldsymbol{\Omega} \times \boldsymbol{1}_z).
\end{align*}
The latter equation can be rearranged to give
\begin{equation}
    \left ( 4 |\boldsymbol{\Omega}|^2 - \omega^2 \right ) \boldsymbol{u} = -\mathrm{i} \omega \nabla \Phi + 4 (\boldsymbol{\Omega} \boldsymbol{\cdot} \boldsymbol{u}) \boldsymbol{\Omega} + 2 \boldsymbol{\Omega} \times \nabla \Phi + \frac{N^2(z)}{\mathrm{i} \omega} (\boldsymbol{1}_z \boldsymbol{\cdot} \boldsymbol{u} ) \left [ 2 \boldsymbol{\Omega} \times \boldsymbol{1}_z - \mathrm{i} \omega \boldsymbol{1}_z \right ], 
    \label{eq:ufctofphi}
\end{equation}
where we have used equation (\ref{eq:euler}a) to simplify $2\boldsymbol{\Omega} \times \boldsymbol{u}$ and used equation (\ref{eq:mass}a) to simplify the buoyancy term. 
Thus, it only remains to express $\boldsymbol{\Omega} \boldsymbol{\cdot} \boldsymbol{u}$ and $\boldsymbol{1}_z \boldsymbol{\cdot} \boldsymbol{u}$ as a function of $\Phi$ only.
We take the scalar product of equation (\ref{eq:euler}a) with $\mathrm{i} \omega \boldsymbol{\Omega}$, which gives
\begin{equation}
    \omega^2 (\boldsymbol{\Omega} \boldsymbol{\cdot} \boldsymbol{u}) = \mathrm{i} \omega (\boldsymbol{\Omega} \boldsymbol{\cdot} \nabla \Phi) + N^2(z) (\boldsymbol{\Omega} \boldsymbol{\cdot} \boldsymbol{1}_z) (\boldsymbol{1}_z \boldsymbol{\cdot} \boldsymbol{u}).
    \label{eq:kdotu}
\end{equation}
Then, we take the scalar product of equation (\ref{eq:euler}a) with $\mathrm{i} \omega \boldsymbol{1}_z$ to obtain
\begin{equation}
    \left [ N^2(z) - \omega^2 \right ] (\boldsymbol{1}_z \boldsymbol{\cdot} \boldsymbol{u}) = \mathrm{i} \omega \, 2 \boldsymbol{\Omega} \boldsymbol{\cdot} (\boldsymbol{1}_z \times \boldsymbol{u}) - \mathrm{i} \omega (\boldsymbol{1}_z \boldsymbol{\cdot} \nabla \Phi).
\end{equation}
Since we have also from equation (\ref{eq:euler}a)
\begin{equation}
\mathrm{i} \omega (\boldsymbol{1}_z \times \boldsymbol{u}) = -  \boldsymbol{1}_z \times (2\boldsymbol{\Omega} \times \boldsymbol{u}) - \boldsymbol{1}_z \times \nabla \Phi = -2 \left [ (\boldsymbol{1}_z \boldsymbol{\cdot} \boldsymbol{u}) \boldsymbol{\Omega} - (\boldsymbol{1}_z \boldsymbol{\cdot} \boldsymbol{\Omega}) \boldsymbol{u} \right ] - \boldsymbol{1}_z \times \nabla \Phi,
    \label{eq:zxu}
\end{equation}
we combine equations (\ref{eq:kdotu}) to (\ref{eq:zxu}) to obtain
\begin{equation}
    -\mathcal{D}_\omega (z) (\boldsymbol{1}_z \boldsymbol{\cdot} \boldsymbol{u}) = \mathrm{i} \omega (2 \boldsymbol{\Omega} \boldsymbol{\cdot} \boldsymbol{1}_z ) (2\boldsymbol{\Omega} \boldsymbol{\cdot} \nabla \Phi) - 2 \omega^2 \, \boldsymbol{\Omega} \boldsymbol{\cdot} (\boldsymbol{1}_z \times \nabla \Phi) - \mathrm{i} \omega^3 (\boldsymbol{1}_z \boldsymbol{\cdot} \nabla \Phi).
    \label{eq:zdotu}
\end{equation}
From equation (\ref{eq:zdotu}), we see that the pressure must satisfy equation (\ref{eq:pressureDw0}) when $\mathcal{D}_\omega (z) =0$.
\end{proof}

%-----------------------------------------------------------------------
\section{Microlocal analysis of boundary-value problems}
\label{sec:microlocal}
%-----------------------------------------------------------------------
We investigate, using microlocal analysis, under which conditions a quite general boundary-value scalar problem is elliptic in \S\ref{subsec:maths}.
Then, we revisit the ellipsoidal model with a constant $N$ in \S\ref{sec:IGMs} to prove that the low-frequency spectrum $\sigma_2$ is essential.

%-----------------------------------------------------------------------
\subsection{Elliptic boundary-value problem}
\label{subsec:maths}
%-----------------------------------------------------------------------
\begin{remark}
In this subsection, the Cartesian coordinates are denoted by $(x_1,x_2,x_3)$ to make use of index notations.
\end{remark}

We consider a smooth bounded domain $V$ in $\mathbb{R}^d$, and a symmetric elliptic differential operator $\mathcal{P}_\omega$ of order 2 with smooth coefficients and of principal symbol $\mathfrak{p}$.
We also assume that the Dirichlet extension of $\mathcal{P}_\omega$ is $>0$.
The scalar $\Phi$ is governed by the boundary-value problem given by
\begin{subequations}
\label{eq:pressureelliptic}
\begin{equation}
    \mathcal{P}_\omega(\Phi) = 0 \ \,  \text{in $V$}, \quad \mathcal{B}(\Phi) := \boldsymbol{n}^\star \boldsymbol{\cdot} \nabla \Phi + \mathcal{V}(\Phi) = 0 \ \,  \text{on $\partial V$},
    \tag{\theequation a,b}
\end{equation}
\end{subequations}
where $\boldsymbol{n}^\star$ is the unit vector normal computed with respect to the metric  $\boldsymbol{g}$ induced on $\partial V$ by the principal symbol of $\mathcal{P}_\omega$ (see Box \ref{box:normalmetric}), and $\mathcal{V}$ is a differential operator of degree $1$ tangent\footnote{\emph{tangent} is here understood in the language of differential geometry \cite{schutz1980geometrical}. For instance, if $\partial V$ is described by $q_1=\text{constant}$ for a set of curvilinear coordinates $(q_1,q_2,q_3)$, then $\mathcal{V}$ only involves derivatives $(\partial_{q_2},\partial_{q_3})$ on the tangent plane at a given point.} to $\partial V$. 
To solve system (\ref{eq:pressureelliptic}), we introduce the Dirichlet-to-Neumann operator $\mathcal{N}$ in Proposition \ref{prop:D2N}. 

\begin{proposition}
If $\Psi:\partial V \to \mathbb{C}$ is a field defined on the boundary, then the Dirichlet-to-Neumann operator is defined by $\mathcal{N} (\Psi) := \partial_{n^\star} \Phi$ where $\partial_{n^\star} = \boldsymbol{n}^\star \boldsymbol{\cdot} \nabla$ is the normal derivative with respect to the vector $\boldsymbol{n}^\star$, and where $\Phi$ is the solution of the Dirichlet problem in $V$ given by $\mathcal{P}_\omega (\Phi) = 0$ in $V$ and $\left . \Phi \right |_{\partial V}= \Psi$ on $\partial V$.
\label{prop:D2N}
\end{proposition}
Note that $\mathcal{N}$ is a self-adjoint operator acting on functions defined in the space $\mathrm{L}^2(\partial V, \mathrm{d}\sigma)$ for the area $\mathrm{d} \sigma$ defined by the metric $\mathfrak{g}$ restricted on $\partial V$ \cite{CdV2024spectrum}. 
To solve problem (\ref{eq:pressureelliptic}), we first solve the boundary problem $\mathcal{N}(\Psi) + \mathcal{V} (\Psi)=0$ on $\partial V$, and then consider the Dirichlet problem in $V$ given the boundary field $\Psi$.
In practice, this amounts to solving the problem on the boundary $\partial V$ and, then, propagating the boundary solution in the interior $V$ by using the Dirichlet problem.
Actually, the nature of the boundary-value problem crucially depends on the nature of the boundary equation $\mathcal{N}(\Psi) + \mathcal{V} (\Psi)=0$ on $\partial V$.
This can be analysed using microlocal analysis. 
Indeed, $\mathcal{N}$ is a nice pseudo-differential operator of principal symbol is given in Proposition \ref{eq:symbolD2N}.
The ellipticity of boundary-value problem (\ref{eq:pressureelliptic}) is then given by Proposition \ref{theo:ellipticBC}. 
Finally, the nature is given by Proposition \ref{theo:discretespec} below. 

\begin{proposition}
The principal symbol of the Dirichlet-to-Neumann operator $\mathcal{N}$, defined for an elliptic operator with a positive principal symbol $\mathfrak{p}>0$, is given by $\sqrt{\left . \mathfrak{g}^\star\right |_{\partial V}}$, where $\left . \mathfrak{g}^\star\right |_{\partial V}$ is the dual of the restriction on $\partial V$ of the metric $\mathfrak{g}$ induced by $\mathfrak{p}$ in $V$. 
\label{eq:symbolD2N}
\end{proposition}
\vspace{-0.75em}
\begin{proof}
It is given in the Appendices of Taylor's book \cite{taylor2013partial}. Note that the principal symbol is defined as $-\sqrt{-\left . \mathfrak{g} \right |_{\partial V}}$ in \cite{taylor2013partial}. The first minus sign comes from the fact that an inward vector $\boldsymbol{n}^\star$ was considered, and the second from the fact that the other convention $\mathfrak{p}\leq 0$ was used.
\end{proof}

\begin{proposition}
Let us denote by $\mathfrak{v}$ the principal symbol of $\mathcal{V}$.
Then, boundary-value problem (\ref{eq:pressureelliptic}) is elliptic if and only if the principal symbol $\sqrt{\left . \mathfrak{g}^\star\right |_{\partial V}} + \mathfrak{v}$ is elliptic on $\partial V$, that is if it does not vanish for non-zero covectors $\boldsymbol{k} \in \mathbb{R}^2$ on the boundary $\partial V$. 
\label{theo:ellipticBC}
\end{proposition}
\vspace{-0.75em}
\begin{proof}
See Chapter 9 in \cite{grubb2008distributions}. This is also known as the Shapiro-Lopatinskii
condition \cite{grubb2008distributions}.
\end{proof}

\begin{encart}[float=t,label={box:normalmetric}]{Normal vector induced by the metric on $\partial V$}
\noindent The normal vector $\widehat{\boldsymbol{n}}$ defined from the metric on $\partial V$ is computed as follows.
The principal symbol $\mathfrak{p}$ defines a metric given by $\mathfrak{g}^\star := (\boldsymbol{\mathfrak{g}}^\star)^{ij} k_i k_j$ for the covectors $\boldsymbol{k}\in\mathbb{R}^3$ in the cotangent space, whose inverse $\boldsymbol{\mathfrak{g}} := (\boldsymbol{\mathfrak{g}}^\star)^{-1}$ induces a metric in $V$ defined by $\mathfrak{g} := \boldsymbol{\mathfrak{g}}_{ij} \mathrm{d}x_i \mathrm{d}x_j$. 
Then, the normal vector $\boldsymbol{n}^\star$ computed with respect to $\boldsymbol{\mathfrak{g}}^\star$ is given by
\begin{equation*}
\boldsymbol{n}^\star \propto \boldsymbol{\mathfrak{g}}^\star \boldsymbol{n}, \quad \boldsymbol{\mathfrak{g}}_{ij} n^\star_i n^\star_j =1,
 \end{equation*}
where $\boldsymbol{n}=(n_1,n_2,n_3)^\top$ is the unit vector normal to $\partial V$ in Cartesian coordinates with $n_1^2+n_2^2+n_3^2 = 1$.
The constant of proportionality is such that $\boldsymbol{n}^\star$ has a unit norm with respect to the metric $\boldsymbol{\mathfrak{g}}$. 
The vector $\boldsymbol{n}^\star$ is also called the \textbf{conormal} vector (with respect to $\boldsymbol{\mathfrak{g}}^\star$) in the study of partial differential equations \cite{miranda1970partial}. 
Note that $\boldsymbol{\mathfrak{g}}^\star$ is a contravariant metric tensor associated with the operator $(\boldsymbol{\mathfrak{g}}^\star)^{ij} \partial_{x_i} \partial_{x_j}$.
Thus, the conormal $\boldsymbol{n}^\star$ is the contravariant description (in the space with the metric defined by $\boldsymbol{\mathfrak{g}}$) of the covariant normal vector $\boldsymbol{n}$ to $\partial V$ (in the space with the Euclidean metric). 
\end{encart}

% \url{https://encyclopediaofmath.org/wiki/Elliptic_operator}

\begin{proposition}
If boundary-value problem (\ref{eq:pressureelliptic}) is elliptic in the interval $[\omega_1,\omega_2]$, then its spectrum is discrete (or empty) in this interval. Otherwise, the spectrum is essential in $[\omega_1,\omega_2]$.
\label{theo:discretespec}
\end{proposition}
\vspace{-0.75em}
\begin{proof}
If the boundary-value pressure problem is elliptic, the operator $\mathcal{P}_\omega$ admits a parametrix $\mathcal{L}$ (i.e. an approximate inverse) such that $\mathcal{P}_\omega \circ \mathcal{L} = \mathrm{Id} + \mathcal{R}$, where $\mathcal{R}$ is a pseudo-differential operator of degree $-1$. 
Hence, $\mathcal{R}$ is a compact operator. 
Moreover, $\mathcal{P}_\omega$ and $\mathcal{R}$ are also analytic functions.
Then, the result follows from the Fredholm analytic theorem (see Theorem VI.14 in \cite{reed1972mathsphys}). Otherwise, the values $\omega \in [\omega_1, \omega_2]$ belong to the essential spectrum \cite{colin2020spectral}.
\end{proof}

In practice, the difficult step is to evaluate the principal symbol of the Dirichlet-to-Neumann operator (Proposition \ref{eq:symbolD2N}). 
In practice, this can be done as follows.
The principal symbol $\mathfrak{p}$ induces a metric in $V$ given by $\mathfrak{g} := \boldsymbol{\mathfrak{g}}_{ij} \mathrm{d}x_i \mathrm{d}x_j$.
Let us assume that $x_i = x_i(q_j)$ with the curvilinear coordinates $(q_1,q_2,q_3)$ such that the boundary $\partial V$ is given by $q_1 = \text{constant}$ (for instance). 
Then, we can rewrite the two-dimensional metric restricted on $\partial V$ as $\left. \mathfrak{g} \right |_{\partial V} = \widetilde{\boldsymbol{\mathfrak{g}}}_{ij} \mathrm{d}q_i \mathrm{d}q_j$. 
Then, the dual of the metric induced by $\boldsymbol{\mathfrak{g}}$ on $\partial V$ is given by $\left. \mathfrak{g}^\star \right |_{\partial V} = (\widetilde{\boldsymbol{\mathfrak{g}}}^\star)^{ij} k_i k_j$ with $\widetilde{\boldsymbol{\mathfrak{g}}}^\star = \widetilde{\boldsymbol{\mathfrak{g}}}^{-1}$. 

%-----------------------------------------------------------------------
\subsection{Back to the pressure equation}
\label{sec:IGMs}
%-----------------------------------------------------------------------
We aim to apply the above microlocal analysis to the boundary-value pressure problem when $N$ is constant, which is elliptic in $V$ when $0 < |\omega| < \omega_-$. 
The principal symbol $\mathfrak{p}$, defined in equation (\ref{eq:symbolPw}), induces the metric $\mathfrak{g}$ in $V$ and its dual $\mathfrak{g}^\star$ characterised by the metric tensors
\begin{subequations}
\label{eq:gg*}
\begin{equation}
    \boldsymbol{\mathfrak{g}}^\star = \begin{pmatrix}
    A & F & E \\
    F & B & D \\
    E & D & C \\
    \end{pmatrix}, \quad \boldsymbol{\mathfrak{g}} = \frac{1}{\det (\boldsymbol{\mathfrak{g}}^\star)}
\begin{pmatrix}
{B C - D^{2}} & {- C F + D E} & {- B E + D F} \\
{- C F + D E} & {A C - E^{2}} & {- A D + E F} \\
{- B E + D F} & - A D + E F & A B - F^{2}
\end{pmatrix},
\tag{\theequation a,b}
\end{equation}
\end{subequations}
with $\det (\boldsymbol{\mathfrak{g}}^\star) = A B C - A D^{2} - B E^{2} - C F^{2} + 2 D E F$. 
For a given normal vector $\boldsymbol{n} = (n_1,n_2,n_3)^\top$ in Cartesian coordinates, the conormal vector $\boldsymbol{n}^\star$ is thus given by
\begin{equation}
   \boldsymbol{n}^\star = \frac{1}{\alpha} \left [ n_1 \begin{pmatrix}
        A \\ F \\ E \\
    \end{pmatrix} + n_2 \begin{pmatrix}
        F \\ B \\ D \\
    \end{pmatrix} + n_3 \begin{pmatrix}
        E \\ D \\ C \\
    \end{pmatrix} \right ],
    \label{eq:conormalP}
\end{equation}
with $\alpha = [ A n_1^2 + B n_2^2 + C n_3^2 + 2 (D n_2 n_3 + E n_1 n_3 + F n_1 n_2)]^{1/2}$ such that $\boldsymbol{\mathfrak{g}}_{ij} n_i^\star n_j^\star = 1$.  
Thanks to expression (\ref{eq:conormalP}), we clearly see that the conormal naturally appears in equations (\ref{eq:fxfyfz}a-c).
Thus, we can apply Proposition \ref{theo:ellipticBC} to determine the nature of the boundary-value pressure problem.

Let us consider here the aligned case with $\boldsymbol{\Omega}=(0,0,\Omega_z)^\top$, for which the equations are then greatly simplified (see Box \ref{box:aligned}). 
In this case, the two metrics reduce to
\begin{subequations}
\begin{equation}
    \mathfrak{g}^\star = k_x^2 + k_y^2 + C k_z^2, \quad \mathfrak{g} = \mathrm{d}x^2 + \mathrm{d}y^2 + C^{-1}\mathrm{d}z^2,
    \tag{\theequation a,b}
\end{equation}
\end{subequations}
and the conormal vector is given by $\boldsymbol{n}^\star = (1/\alpha) \, (n_1, n_2, n_3C)^\top$ with $\alpha = \sqrt{1+n_3^2(C-1)}$, where the Euclidian normal vector is $\boldsymbol{n} = (n_1, n_2, n_3)^\top$ with $n_1^2 + n_2^2 + n_3^2 = 1$. 
To demonstrate that the pressure problem is not elliptic for a given frequency $0<|\omega| < \omega_-$, it is sufficient to find one location on $\partial V$ where Proposition \ref{theo:ellipticBC} is not valid. 
We look at the behaviour at the equator, where the tangent plane is vertical with $\boldsymbol{n} = \boldsymbol{1}_x$ (for instance) such that the pressure BC is simply $f_x=0$. 
The restriction of the metrics on $\partial V$ gives 
\begin{subequations}
\begin{equation}
    \left . \mathfrak{g} \right |_{\partial V} = \mathrm{d}y^2 + C^{-1}\mathrm{d}z^2, \quad \left . \mathfrak{g}^\star \right |_{\partial V} = k_y^2 + C k_z^2,
    \tag{\theequation a,b}
\end{equation}
\end{subequations}
and the conormal vector is simply given by $\boldsymbol{n}^\star = \boldsymbol{1}_x$. 
The principal symbol of the pressure BC is then given by
\begin{equation}
\sqrt{k_y^2 + C k_z^2} = -\frac{2\Omega_z}{\omega} k_y,
\end{equation}
which reduces to
\begin{equation}
    \underbrace{\left ( \frac{4 \Omega_z^2 -\omega^2}{\omega^2} \right ) }_{\geq 0} k_y^2 = \underbrace{C}_{\geq 0} k_z^2
\end{equation}
where the two prefactors are positive when $0 < |\omega| < \omega_- = \min(N,2\Omega_z)$ in the aligned case. 
Therefore, we can find a real-valued covector $\boldsymbol{k} = (0,k_y,k_z)^\top$ such that the principal symbol of the BC vanishes. 
The boundary-value pressure problem is thus not elliptic in this case (Proposition \ref{eq:symbolD2N}), and the low-frequency spectrum is essential (Proposition \ref{theo:discretespec}).
Similarly, we conjecture that the spectrum is always essential when $0 < |\omega| < \omega_-$ when $\boldsymbol{\Omega}$ and $\boldsymbol{1}_z$ are misaligned.
The corresponding calculations (which are more lengthy) will be included in a forthcoming paper \cite{CdV2024spectrum}, in which we will also further explore the properties of the low-frequency (surface) modes.

% -----------------------------------------------------------------------
\section{Boundary-layer analysis}
\label{appendix:BLT}
% -----------------------------------------------------------------------
We introduce diffusion in the problem to investigate the diffusive decay rates of IGMs for a constant BV frequency in an ellipsoid.
We use the same model as in the main text except that, here, the fluid has a non-zero kinematic viscosity $\nu$ and diffusivity $\kappa$.
We employ boundary-layer theory (BLT) to approach the low-diffusive regime (relevant for geophysical vortices), and introduce dimensionless variables below. 
We take $\Omega_s^{-1}$ as the time scale, $L$ as the length scale (e.g. either $a$ or $c$), and $\rho_* L N^2/g$ as the density scale. 
For the sake of concision, the dimensionless variables are written below using the same symbols as the dimensional variables in the main text. 
We seek modal solutions as $[ \boldsymbol{v}, \pi, \rho ] (\boldsymbol{r},t)~=~[ \boldsymbol{u}, \Phi, \zeta ] (\boldsymbol{r}) \exp (\lambda t)$, where $\lambda \in \mathbb{C}$ is the diffusive eigenvalue.
The dimensionless linearised Navier-Stokes equations are then given in the rotating frame by
\begin{subequations}
\label{eq:NSBLT}
\begin{equation}
    \lambda \boldsymbol{u} + 2 (\boldsymbol{1}_\Omega \times \boldsymbol{u}) = - \nabla \Phi - \widetilde{N}^2 \zeta \boldsymbol{1}_z + E \nabla^2 \boldsymbol{u}, \quad  \lambda \zeta = (\boldsymbol{u} \boldsymbol{\cdot} \boldsymbol{1}_z) + E_\kappa \nabla^2 \zeta,
   \tag{\theequation a,b}
\end{equation}
\end{subequations}
with the (dimensionless) Ekman number $E = \nu/(\Omega_s L^2)$ and $E_\kappa = E/Pr$ where $Pr=\nu/\kappa$ is the Prandtl number, and with the normalised BV frequency $\widetilde{N} = N/\Omega_s$. 
Multiplying equation (\ref{eq:NSBLT}a) by $\boldsymbol{u}^\dagger$ and the complex-conjugate of equation (\ref{eq:NSBLT}b) by $\widetilde{N}^2 \zeta$, we obtain after volume integration
\begin{equation}
     \lambda || \boldsymbol{u}||^2 + \lambda^\dagger \widetilde{N}^2 || \zeta ||^2 = E \langle \boldsymbol{u}, \nabla^2 \boldsymbol{u} \rangle + \widetilde{N}^2 E_\kappa \langle \nabla^2 \zeta, \zeta \rangle,
     \label{eq:decayC1}
\end{equation}
where the scalar product between two scalar quantities is defined as $\langle \zeta_1, \zeta_2\rangle = \int_V\zeta_1^\dagger \boldsymbol{\cdot} \zeta_2 \, \mathrm{d} V$. 
% the vector inner product defined in equation (\ref{eq:innerproduct}).
Note that formula (\ref{eq:decayC1}) is also valid in a sphere with a radial gravity \cite{dintrans1999gravito}.
Finally, the problem is supplemented with BC. 
We assume that the velocity field satisfies the stress-free BCs
\begin{subequations}
\label{eq:SFBC}
\begin{equation}
    \left . \boldsymbol{u} \boldsymbol{\cdot} \boldsymbol{n} \right |_{\partial V} = 0, \quad \left . [ \boldsymbol{\epsilon}(\boldsymbol{u}) \boldsymbol{\cdot} \boldsymbol{n} ] \times \boldsymbol{n} \right |_{\partial V} = \boldsymbol{0}
    \tag{\theequation a,b}
\end{equation}
\end{subequations}
where $\boldsymbol{\epsilon} (\boldsymbol{u}) = (1/2) \left [ \nabla \boldsymbol{u} + (\nabla \boldsymbol{u})^\top \right]$ is the strain-rate tensor.
SF-BCs (\ref{eq:SFBC}a,b) allow the tangential velocity to slip on $\partial V$.
For the density perturbation, we enforce a Neumann BC $\left . \nabla \zeta \boldsymbol{\cdot} \boldsymbol{n} \right |_{\partial V}~=~0$. 
A density jump is thus possible between the vortex and the ambient fluid (as considered in idealised models \cite{hassanzadeh2012universal,aubert2012universal,lemasquerier2020remote}), but without exchange of mass between the two fluids. 

We can simplify equation (\ref{eq:SFBC}) using BLT when $E \ll 1$.
Formally, classical BLT for rotating fluids is not expected to be strongly modified by stratification when $\widetilde{N}$ is not too large (e.g. for rigid boundaries \cite{friedlander1989asymptotic}).
Moreover, it is known that the boundary-layer flow is $E^{1/2}$ smaller than the bulk flow for SF-BCs \cite{rieutord1997ekman,vidal2023precession}.
This considerably simplifies the BLT, because explicit expressions for the boundary-layer solutions are no longer required to estimate the decay rate of the modes. 
We seek the variables at the leading order in $E$ for the SF-BCs as
\begin{subequations}
\label{eq:BLTansatz}
\begin{equation}
    \lambda \simeq \mathrm{i} \omega_0 + E \lambda_1, \quad \left [ \boldsymbol{u}, \zeta \right ] \simeq \left [ \boldsymbol{u}_0, \zeta_0 \right ] + E^{1/2} \left [ \boldsymbol{u}_1, \zeta_1 \right ],
    \tag{\theequation a,b}
\end{equation}
\end{subequations}
where $[\omega_0,\boldsymbol{u}_0]$ is the eigenvalue-eigenvector pair of a diffusionless IGM, $\lambda_1 \in \mathbb{C}$ is the first-order correction of the eigenvalue, and $\left [ \boldsymbol{u}_1, \zeta_1 \right ]$ are the first-order corrections within the boundary layer such that $\boldsymbol{u}_0 + E^{1/2} \boldsymbol{u}_1$ satisfies SF-BCs (\ref{eq:SFBC}a,b) and $\zeta_0~+~E^{1/2}\zeta_1$ satisfies the Neumann BC for the density. 
We substitute the above asymptotic expansions into equation (\ref{eq:decayC1}) and we obtain
\begin{equation}
\lambda_1 ||\boldsymbol{u}_0||^2 + \lambda_1^\dagger \widetilde{N}^2 ||\zeta_0||^2 \simeq - \left ( 2 \int_V  \boldsymbol{\epsilon} (\boldsymbol{u}_0) : \boldsymbol{\epsilon} (\boldsymbol{u}_0^\dagger) \, \mathrm{d} V + \frac{\widetilde{N}^2}{Pr} \int_V (\nabla \zeta_0)^2 \, \mathrm{d} V \right )
\label{eq:decaySF1}
\end{equation}
at the order $E$, where we have used the SF-BCs and the Neumann BC to simplify the volume integrals on the right-hand side.
Since the latter is real-valued and negative, there is no frequency correction at the leading order in $E$ due to SF-BCs such that $\Im_m(\lambda_1) = 0$ and $\Re_e(\lambda_1)=\tau_1 \leq 0$.
We recover from equation (\ref{eq:decaySF1}) the viscous decay rate of pure inertial modes with SF-BCs \cite{vidal2023precession}.

Finally, we can crudely estimate whether forced IGMs would quickly decay or not in large-scale stratified vortices by plugging geophysical estimates into equation (\ref{eq:decaySF1}). 
Typical values for the kinematic viscosity are
$\nu \sim 10^{-6}$~m${}^2$.s${}^{-1}$ for water and $\nu \sim 4 \times 10^{-7}$~m${}^2$.s${}^{-1}$ for gas giants (according to ab-initio simulations \cite{french2012ab}).
Using Table 1 in the main text, typical values for the Ekman number are thus $E \sim 10^{-8} - 10^{-12}$ for Mediterranean eddies and $E \sim 10^{-12} - 10^{-17}$ for Jovian vortices. 
Estimating the Prandtl number depends on whether the stratification is due to thermal effects or compositional ones (for which $Pr \gg 1$ in both water and gas giants).
Heat diffusion in Mediterranean eddies is characterised by $Pr \simeq 0.7$, and typical values $Pr \sim 10^{-2} - 1$ are expected for thermal diffusion in gas giants \cite{french2012ab}.
Hence, we have $\widetilde{N} E_\kappa \sim 10^{-4}-10^{-8}$ for Mediterranean eddies and $\widetilde{N} E_\kappa \sim 10^{-6}-10^{-11}$ for Jovian vortices. 
Consequently, the largest-scale IGMs are not expected to be strongly damped by diffusion in stratified vortices.

\bibliographystyle{RS}
\bibliography{supp}